\documentclass{LMCS}

\def\dOi{9(3:17)2013}
\lmcsheading%
{\dOi}
{1--34}
{}
{}
{Jul.~22, 2012}
{Sep.~17, 2013}
{}

\usepackage{amssymb,stmaryrd}
\usepackage{hyperref}
\usepackage{breakurl}
\usepackage{enumerate}
\usepackage{mathrsfs}
\usepackage{fancybox}

\newfont{\bbb}{bbm10 scaled 1100}        
\newfont{\bbbs}{bbm10 scaled 800}        
\def\text#1{\textrm{#1}}

\def\precond#1{{\vphantom{#1}}^\bullet #1}
\def\postcond#1{{#1}^\bullet}

\def\Production#1{\stackrel{#1}{\Longrightarrow}}
\def\production#1{\stackrel{#1}{\longrightarrow}}
\def\equivalent{\Leftrightarrow}
\newfont{\fsc}{eusm10 scaled 1100}      
\def\powermultiset#1{\nat^{#1}}
\def\implies{\Rightarrow}

\def\equivalent{\Leftrightarrow}
\def\mathrlap{\mathpalette\mathrlapinternal}
\def\mathrlapinternal#1#2{%
  \rlap{$\mathsurround=0pt#1{#2}$}}
\def\mathllap{\mathpalette\mathllapinternal}
\def\mathllapinternal#1#2{%
  \llap{$\mathsurround=0pt#1{#2}$}}

\def\into{\rightarrow}

\def\defitem#1{\emph{#1}}

\def\rpair#1{\langle#1\rangle}

\def\AA{\text{\it AA}}

\def\onespace#1{\let\argument=#1\ifx\onespace#1\else~\fi\argument}
\def\failureset{\mathcal{F}}

\def\structuralM{\mbox{\sf M}}
\let\origmin\min
\def\min{\mathord{\origmin}}
\let\origmax\max
\def\max{\mathord{\origmax}}

\def\quireunderscore{_}
\def\quire#1{%
  \def\tmp{#1}%
  \ifx\tmp\quireunderscore%
    \def\tmp{\quireindexed_}
  \else%
    \def\tmp{\mathcal{Q}#1}
  \fi\tmp}
\def\quireindexed_#1{\mathcal{Q}_{\text{#1}}}

\makeatletter
\def\goesto{\@transition\rightarrowfill}
\def\Goesto{\@transition\Rightarrowfill}
\def\ngoesto{\@transition\nrightarrowfill}
\def\nGoesto{\@transition\nRightarrowfill}
\def\@transition#1{\@ifnextchar[{\@@transition{#1}}{\@@transition{#1}[]}}
\newbox\@transbox
\newbox\@arrowbox
\def\rightarrowfill{$\m@th\mathord-\mkern-6mu%
  \cleaders\hbox{$\mkern-2mu\mathord-\mkern-2mu$}\hfill
  \mkern-6mu\mathord\rightarrow$}
\def\Rightarrowfill{$\m@th\mathord=\mkern-6mu%
  \cleaders\hbox{$\mkern-2mu\mathord=\mkern-2mu$}\hfill
  \mkern-6mu\mathord\Rightarrow$}
\def\@@transition#1[#2]%
   {\setbox\@transbox\hbox
       {\vrule height 1.5ex depth 1ex width 0ex\hskip0.25em$\scriptstyle#2$\hskip0.25em}
   \ifdim\wd\@transbox<1.5em
      \setbox\@transbox\hbox to 1.5em{\hfil\box\@transbox\hfil}\fi
   \setbox\@arrowbox\hbox to \wd\@transbox{#1}
   \ht\@arrowbox\z@\dp\@arrowbox\z@
   \setbox\@transbox\hbox{$\mathop{\box\@arrowbox}\limits^{\box\@transbox}$}
   \ht\@transbox\z@\dp\@transbox\z@
   \mathrel{\box\@transbox}}
\makeatother

\def\alignedcaption[#1&#2]{\mbox{\scriptsize $\mathllap{#1{}}\mathrlap{#2}$}}

\def\ie{i.e.\ }

\def\varnothing{\emptyset}
\def\Act{{\rm Act}}
\def\Loc{\textrm{\upshape Loc}}
\def\concurrent{\smile}

\def\restrictedto{\mathop\upharpoonright}

\newcommand{\dcup}{\stackrel{\mbox{\huge .}}{\cup}}   
\newcommand{\plat}[1]{\raisebox{0pt}[0pt][0pt]{#1}}   
\newcommand{\inp}{\mathbin\in}                        
\def\idx#1#2#3#4#5{
  \def\argone{#1}
  \def\argtwo{#2}
  \def\argthree{#3}
  \def\argfour{#4}
  \def\argfive{#5}
  \def\testprime{'}
  \def\testdprime{''}
  \def\testtprime{'''}
  {\vphantom{\argthree}}_%
    {\vphantom{\argfour}\argone}^%
    {\vphantom{\argfive}{\argtwo}}%
  \argthree_%
    {\vphantom{\argone}\argfour}%
    \ifx\argfive\testprime\argfive\else%
    \ifx\argfive\testdprime\argfive\else%
    \ifx\argfive\testtprime\argfive\else%
    ^{\vphantom{\argtwo}\argfive}\fi\fi\fi%
}

\def\indexset{\mathfrak{K}}

\def\impl#1{\mathcal{I}(#1)}

\newcommand{\visible}{}

\newenvironment{itemise}{\begin{list}{$\bullet$}{\leftmargin 12pt \labelwidth\leftmargin\advance\labelwidth-\labelsep \topsep 4pt \itemsep 2pt \parsep 2pt}}{\end{list}}

\def\justempty{}

\DeclareFontFamily{T1}{la}{}
\DeclareFontShape{T1}{la}{m}{n}{<->s*[0.8571428571]la14}{}

\newcommand{\Rel}{\mathcal{B}\,}

\newenvironment{proofNobox}{\begin{trivlist} \item[\hspace{\labelsep}\bf Proof:]}{\end{trivlist}}
\newcommand{\filledbox}{\rule{1.2ex}{1.2ex}}
\newenvironment{proofclaim}{\begin{trivlist} \item[\hspace{\labelsep}\it Proof:]}{\hfill\filledbox\end{trivlist}}
\newenvironment{proofclaimNobox}{\begin{trivlist} \item[\hspace{\labelsep}\it Proof:]}{\end{trivlist}}

\newcommand{\refdf}[1]{Definition~\ref{df-#1}}
\newcommand{\refthm}[1]{Theorem~\ref{thm-#1}}
\newcommand{\refpr}[1]{Proposition~\ref{pr-#1}}
\newcommand{\reflem}[1]{Lemma~\ref{lem-#1}}

\newcommand{\reffig}[1]{Figure~\ref{fig-#1}}

\newcommand{\refobs}[1]{Observation~\ref{obs-#1}}
\newcommand{\refcl}[1]{Claim~\ref{cl-#1}}
\newcommand{\refsec}[1]{Section~\ref{sec-#1}}
\newcommand{\reftab}[1]{Table~\ref{tab-#1}}
\newcommand{\UI}{\Omega}             
\newcommand{\UIij}{\UI_{i}}          
\newcommand{\ui}{\mbox{\it \i}}      
\newcommand{\txf}[1]{\mbox{\small\sf #1}}
\newcommand{\dist}[1][p]{\txf{distribute}_#1}
\newcommand{\ini}[1][j]{\txf{initialise}_#1}
\newcommand{\trans}[2][h]{\txf{transfer}^#1_#2}
\newcommand{\exec}[2][i]{\txf{execute}^#1_#2}
\newcommand{\fetch}[1][p,c]{\txf{fetch}_{i,j}^{#1}}
\newcommand{\fetched}[2][i]{\txf{fetched}^#1_#2}
\newcommand{\comp}[2][i]{\txf{finalise}^#1}

\newcommand{\fire}{\txf{fire}}
\newcommand{\undo}[1][i]{\txf{undo}_{#1}}
\newcommand{\und}[1][f]{\txf{undo}(#1)}
\newcommand{\undone}{\txf{undone}}
\newcommand{\reset}[1][i]{\txf{reset}_{#1}}
\newcommand{\elide}[1][i]{\txf{elide}_{#1}}
\newcommand{\ack}[1][i]{\txf{ack}_{#1}}
\newcommand{\keep}[1][i]{\rho_{#1}}
\newcommand{\Fired}{\txf{fired}}
\newcommand{\take}{\txf{take}}
\newcommand{\took}{\txf{took}}

\newcommand{\Pre}{\txf{pre}}
\newcommand{\transin}[2][h]{\txf{trans}^#1_#2\txf{-in}}
\newcommand{\transout}[2][h]{\txf{trans}^#1_#2\txf{-out}}
\newcommand{\fetchin}[1][p,c]{\txf{fetch}_{i,j}^{#1}\txf{-in}}
\newcommand{\fetchout}[1][p,c]{\txf{fetch}_{i,j}^{#1}\txf{-out}}

\newcommand{\weight}[1]{\hfill\mbox{\scriptsize $F'(#1)$}}

\newcommand{\leqc}{\leq^\#\!}
\newcommand{\confeq}{\mathbin{\plat{$\stackrel{\#}{=}$}}}
\newcommand{\confeqscript}{\stackrel{\#}{=}}
\newcommand{\nat}{\mbox{\bbb N}}
\newcommand{\nats}{\mbox{\bbbs N}}
\newcommand{\Int}{\mbox{\bbb Z}}
\newcommand{\fin}{\in_{\scriptscriptstyle F}}
\newcommand{\marking}[1]{\llbracket#1\rrbracket}  
\newcommand{\NF}{\mbox{\it NF}}                          
\newcommand{\opt}[1]{\mbox{\tiny\rm(}#1\mbox{\tiny\rm)}} 
\usepackage{pstricks}
\usepackage{pst-node}
\usepackage{graphicx}

\newrgbcolor{darkblue}{0 0 0.5}
\newrgbcolor{darkred}{0.7 0 0}
\newrgbcolor{verylightgray}{0.9 0.9 0.9}

\newcounter{netimage}

\def\p#1:#2;{\cnode #1{0.3}{n\thenetimage-#2}}
\def\P#1:#2;{\p #1:#2;\pscircle*#1{0.1}}
\def\q#1:#2:#3;{\p #1:#2;\rput#1{\rput[l](0.45,0){\large #3}}}
\def\Q#1:#2:#3;{\P #1:#2;\rput#1{\rput[l](0.45,0){\large #3}}}
\def\qq#1:#2:#3;{\p #1:#2;\rput#1{\rput[t](0,-0.5){\large #3}}}
\def\ql#1:#2:#3;{\p #1:#2;\rput#1{\rput[r](-0.45,0){\large #3}}}
\def\qr#1:#2:#3;{\p #1:#2;\rput#1{\rput[l](0.45,0){\large #3}}}
\def\qt#1:#2:#3;{\p #1:#2;\rput#1{\rput[b](0,0.45){\large #3}}}
\def\qts#1:#2:#3;{\p #1:#2;\rput#1{\rput[b](0,0.35){\large #3}}}
\def\qbs#1:#2:#3;{\p #1:#2;\rput#1{\rput[t](0,-0.35){\large #3}}}
\def\qb#1:#2:#3;{\p #1:#2;\rput#1{\rput[t](0,-0.45){\large #3}}}
\def\Ql#1:#2:#3;{\P #1:#2;\rput#1{\rput[r](-0.45,0){\large #3}}}
\def\Qr#1:#2:#3;{\P #1:#2;\rput#1{\rput[l](0.45,0){\large #3}}}
\def\Qt#1:#2:#3;{\P #1:#2;\rput#1{\rput[b](0,0.45){\large #3}}}
\def\Qb#1:#2:#3;{\P #1:#2;\rput#1{\rput[t](0,-0.45){\large #3}}}
\def\qx#1:#2:#3:#4;{\p #1:#2;\rput#1{\rput#4{\large #3}}}
\def\QXX#1:#2:#3:#4:#5;{\p #1:#2;\rput#1{\rput#4{\large #3}}\pscircle*#5{0.1}}
\def\s#1:#2:#3;{\p #1:#2;\rput#1{\rput(-0.03,0){\large #3}}}
\def\t#1:#2:#3;{\rput#1{\rnode{n\thenetimage-#2}{\psframebox{%
  \vbox to 0.6cm{\vfil\hbox to 0.6cm{\hfil\Large #3\hfil}\vfil}}}}}
\def\u#1:#2:#3:#4;{\rput#1{\rnode{n\thenetimage-#2}{\psframebox{%
  \vbox to 0.6cm{\vfil\hbox to 0.6cm{\hfil\Large #3\hfil}\vfil}}}}%
  \rput#1{\rput[l](0.6,0){\large #4}}}
\def\uts#1:#2:#3:#4;{\rput#1{\rnode{n\thenetimage-#2}{\psframebox{%
  \vbox to 0.6cm{\vfil\hbox to 0.6cm{\hfil\Large #3\hfil}\vfil}}}}%
  \rput#1{\rput[b](0,0.4){\large #4}}}
\def\ut#1:#2:#3:#4;{\rput#1{\rnode{n\thenetimage-#2}{\psframebox{%
  \vbox to 0.6cm{\vfil\hbox to 0.6cm{\hfil\Large #3\hfil}\vfil}}}}%
  \rput#1{\rput[b](0,0.6){\large #4}}}
\def\uT#1:#2:#3:#4;{\rput#1{\rnode{n\thenetimage-#2}{\psframebox{%
  \vbox to 0.6cm{\vfil\hbox to 0.6cm{\hfil\Large #3\hfil}\vfil}}}}%
  \rput#1{\rput[b](0,0.8){\large #4}}}
\def\ubs#1:#2:#3:#4;{\rput#1{\rnode{n\thenetimage-#2}{\psframebox{%
  \vbox to 0.6cm{\vfil\hbox to 0.6cm{\hfil\Large #3\hfil}\vfil}}}}%
  \rput#1{\rput[t](0,-0.45){\large #4}}}
\def\ub#1:#2:#3:#4;{\rput#1{\rnode{n\thenetimage-#2}{\psframebox{%
  \vbox to 0.6cm{\vfil\hbox to 0.6cm{\hfil\Large #3\hfil}\vfil}}}}%
  \rput#1{\rput[t](0,-0.6){\large #4}}}
\def\ul#1:#2:#3:#4;{\rput#1{\rnode{n\thenetimage-#2}{\psframebox{%
  \vbox to 0.6cm{\vfil\hbox to 0.6cm{\hfil\Large #3\hfil}\vfil}}}}%
  \rput#1{\rput[r](-0.6,0){\large #4}}}
\def\ur#1:#2:#3:#4;{\rput#1{\rnode{n\thenetimage-#2}{\psframebox{%
  \vbox to 0.6cm{\vfil\hbox to 0.6cm{\hfil\Large #3\hfil}\vfil}}}}%
  \rput#1{\rput[l](0.6,0){\large #4}}}
\def\a#1->#2;{\ncline{->}{n\thenetimage-#1}{n\thenetimage-#2}}
\def\aBack#1->#2;{\ncline{-}{n\thenetimage-#1}{n\thenetimage-#2}\ncline[nodesep=0.2]{-<}{n\thenetimage-#1}{n\thenetimage-#2}}
\def\aEarly#1->#2;{\ncline{->}{n\thenetimage-#1}{n\thenetimage-#2}\ncline[nodesep=0.25]{-<}{n\thenetimage-#1}{n\thenetimage-#2}}
\def\aLate#1->#2;{\ncline{->}{n\thenetimage-#1}{n\thenetimage-#2}\ncline[nodesep=0.6]{-<}{n\thenetimage-#1}{n\thenetimage-#2}}
\def\aUndo#1->#2;{\ncline[arrowinset=0]{->}{n\thenetimage-#1}{n\thenetimage-#2}}
\def\aReset#1->#2;{\ncline[arrowinset=0]{->}{n\thenetimage-#1}{n\thenetimage-#2}\ncline[linecolor=white,arrowinset=0,arrowscale=0.6,linestyle=none,nodesep=0.07]{->}{n\thenetimage-#1}{n\thenetimage-#2}}
\let\aFar\aLate

\let\aAck\aReset
\def\A#1->#2;{\ncarc[arcangle=22]{->}{n\thenetimage-#1}{n\thenetimage-#2}}
\def\AA#1->#2;{\ncarc[arcangle=7]{->}{n\thenetimage-#1}{n\thenetimage-#2}}
\def\B#1->#2;{\ncarc[arcangle=-18]{->}{n\thenetimage-#1}{n\thenetimage-#2}}
\def\BB#1->#2;{\ncarc[arcangle=-32]{->}{n\thenetimage-#1}{n\thenetimage-#2}}
\def\BBB#1->#2;{\ncarc[arcangle=-60]{->}{n\thenetimage-#1}{n\thenetimage-#2}}
\def\avlinearc{0.2}
\def\av#1[#2]-#3->[#4]#5;{
  \SpecialCoor
  \psline[linearc=\avlinearc]{->}([angle=#2]n\thenetimage-#1)#3([angle=#4]n\thenetimage-#5)
}
\def\interface#1:#2:#3:#4;{\rput#1{\rnode{n\thenetimage-#2}{\psframebox*[fillstyle=solid,fillcolor=black]{\vbox to 0.5cm{\vfill\hbox to 0.05cm{}}}}\rput[lt](0.05,-0.1){#4}}\ncline{-}{n\thenetimage-#2}{n\thenetimage-#3}}

\makeatletter
\long\def\petrinet(#1)#2\end{\psscalebox{0.8}{\pspicture(#1)\stepcounter{netimage}#2\endpspicture}\end}

\makeatother

\psset{arrowsize=5pt}

\begin{document}

\title[On Characterising Distributability]{On Characterising Distributability\rsuper*}

\author[R.~van Glabbeek]{Rob van Glabbeek\rsuper a}
\address{{\lsuper a}NICTA, Sydney, Australia\vspace{-8 pt}}
\address{School of Computer Science and Engineering\\ Univ.\ of New South Wales, Sydney, Australia}
\email{rvg@cs.stanford.edu}
\thanks{{\lsuper a}NICTA is funded by the Australian Government as represented
    by the Department of Broadband, Communications and the Digital Economy and the
    Australian Research Council through the ICT Centre of Excellence program.}

\author[U.~Goltz]{Ursula Goltz\rsuper b}
\address{{\lsuper{b,c}}Institute for Programming and Reactive Systems\\ TU Braunschweig, Germany}
\email{goltz@ips.cs.tu-bs.de, drahflow@gmx.de}

\author[J.-W.~Schicke-Uffmann]{Jens-Wolfhard Schicke-Uffmann\rsuper c}
\address{\vskip-6 pt}


\thanks{This work was partially supported by the DFG
    (German Research Foundation).}

\ACMCCS{[{\bf Theory of computation}]: Models of
  computation---Concurrency---Distributed computing models}

\titlecomment{{\lsuper*}This paper is adapted from \cite{GGS12}; it characterises distributability for a
    slightly larger range of semantic equivalence relations, and incorporates various
    remarks stemming from \cite{glabbeek08syncasyncinteractionmfcs}.
    An extended abstract appeared in L.~Birkedal, ed.: Proc.\ 15th
    Int.\ Conf.\ on Foundations of Software Science and
    Computation Structures (FoSSaCS 2012), LNCS 7213,
    Springer, 2012, pp. 331--345, doi:\href{http://dx.doi.org/10.1007/978-3-642-28729-9\_22}
    {10.1007/978-3-642-2872-9\_22}.}

\begin{abstract}
We formalise a general concept of distributed systems as sequential
components interacting asynchronously.  We define a corresponding
class of Petri nets, called LSGA nets, and precisely characterise
those system specifications which can be implemented as LSGA nets up
to branching ST-bisimilarity with explicit divergence.\vspace{-7pt}
\end{abstract}

\maketitle

\section{Introduction}
The aim of this paper is to contribute to a fundamental understanding
of the concept of a distributed reactive system and the paradigms of
synchronous and asynchronous interaction. We start by giving an
intuitive characterisation of the basic features of distributed
systems. In particular we assume that distributed systems consist of
components that reside on different locations, and that any signal
from one component to another takes time to travel.  Hence the only
interaction mechanism between components is asynchronous
communication.

Our aim is to characterise which system specifications may be implemented as
distributed systems.  In many formalisms for system specification or design,
synchronous communication is provided as a basic notion; this happens for
example in process algebras.  Hence a particular challenge is that it may be
necessary to simulate synchronous communication by asynchronous communication.
 
Trivially, any system specification may be implemented distributedly by locating
the whole system on one single component. Hence we need to pose some additional
requirements. One option would be to specify locations for system activities and
then to ask for implementations satisfying this distribution and still
preserving the behaviour of the original specification. This is done in
\cite{BCD02}. Here we pursue a different approach. We add another requirement to
our notion of a distributed system, namely that its components only allow
sequential behaviour. We then ask whether an arbitrary system specification may
be implemented as a distributed system consisting of sequential components in an
optimal way, that is without restricting the concurrency of the original
specification. This is a particular challenge when synchronous communication
interacts with concurrency in the specification of the original system. We will
give a precise characterisation of the class of distributable systems, which
answers in particular under which conditions synchronous communication may be
implemented in a distributed setting.
  
For our investigations we need a model which is expressive enough to represent
concurrency. It is also useful to have an explicit representation of the
distributed state space of a distributed system, showing in particular the local
control states of components. We choose Petri nets, which offer these
possibilities and additionally allow finite representations of infinite
behaviours. We focus on the class of \emph{structural conflict nets}
\cite{glabbeek11ipl}---a proper generalisation of the class of
one-safe place/transition systems, where conflict and concurrency are clearly separated.

For comparing the behaviour of systems with their distributed implementation we
need a suitable equivalence notion. Since we think of open systems interacting
with an environment, and since we do not want to restrict concurrency in
applications, we need an equivalence that respects branching time and
concurrency to some degree. Our implementations use transitions which are
invisible to the environment, and this should be reflected in the equivalence
by abstracting from such transitions. However, we do not want implementations to
introduce divergence. In the light of these requirements we work with two
semantic equivalences. \emph{Step failures equivalence} is one of the
weakest equivalences that captures branching time, concurrency and divergence to some
degree; whereas \emph{branching ST-bisimilarity with explicit divergence}
fully captures branching time, divergence, and those aspects of concurrency that
can be represented by concurrent actions overlapping in time.
We obtain the same characterisation for both notions of equivalence, and thus
implicitly for all notions in between these extremes.

We model distributed systems consisting of sequential components as an
appropriate class of Petri nets, called \emph{LSGA nets}.  These are obtained by
composing nets with sequential behaviour by means of an asynchronous parallel
composition. We show that this class corresponds exactly to a more abstract
notion of distributed systems, formalised as \emph{distributed nets}
\cite{glabbeek08syncasyncinteractionmfcs}.

We then consider distributability of system specifications which are represented
as structural conflict nets. A net $N$ is \emph{distributable} if there exists a
distributed implementation of $N$, that is a distributed net which is
semantically equivalent to $N$.  In the implementation we allow unobservable
transitions, and labellings of transitions, so that single actions of the
original system may be implemented by multiple transitions. However, the system
specifications for which we search distributed implementations are \emph{plain}
nets without these features. This restriction is motivated in the conclusion.

We give a precise characterisation of distributable nets in terms of a
semi-structural property.  This characterisation provides a formal proof that
the interplay between choice and synchronous communication is a key issue for
distributability.

To establish the correctness of our characterisation we develop a new
method for rigorously proving the equivalence of two Petri nets, one of 
which known to be plain, up to branching ST-bisimilarity with explicit divergence.
 
\section{Basic Notions}
\label{sec-basic}
In this paper we employ \emph{signed multisets}, which generalise
multisets by allowing elements to occur in it with a negative multiplicity.

\begin{defi}\label{df-multiset}
Let $X$ be a set.
\begin{enumerate}[$-$]
\item A {\em signed multiset} over $X$ is a function $A\!:X \rightarrow \Int$,
  \ie $A\in \Int^{X}$.\\
  It is a \emph{multiset} iff $A\in \nat^X$, \ie iff $A(x)\geq 0$ for all $x\in X$.
\item $x \in X$ is an \defitem{element of} a signed multiset $A\in\Int^X$, notation $x \in
  A$, iff $A(x) \neq 0$. 
\item For signed multisets $A$ and $B$ over $X$ we write $A \leq B$ iff
 \mbox{$A(x) \leq B(x)$} for all $x \inp X$;
\\ $A\cup B$ denotes the signed multiset over $X$ with $(A\cup B)(x):=\text{max}(A(x), B(x))$,
\\ $A\cap B$ denotes the signed multiset over $X$ with $(A\cap B)(x):=\text{min}(A(x), B(x))$,
\\ $A + B$ denotes the signed multiset over $X$ with $(A + B)(x):=A(x)+B(x)$,
\\ $A - B$ denotes the signed multiset over $X$ with $(A - B)(x):=A(x)-B(x)$, and\\
for $k\inp\Int$ the signed multiset $k\cdot A$ is given by $(k \cdot A)(x):=k\cdot A(x)$.
\item The function $\emptyset\!:X\rightarrow\nat$, given by
  $\emptyset(x):=0$ for all $x \inp X$, is the \emph{empty} multiset over $X$.
\item If $A$ is a signed multiset over $X$ and $Y\subseteq X$ then
  $A\restrictedto Y$ denotes the signed multiset over $Y$ defined by
  $(A\restrictedto Y)(x) := A(x)$ for all $x \inp Y$.
\item The cardinality $|A|$ of a signed multiset $A$ over $X$ is given by
  $|A| := \sum_{x\in X}|A(x)|$.
\item A signed multiset $A$ over $X$ is \emph{finite}
  iff $|A|<\infty$, i.e.,
  iff the set $\{x \mid x \inp A\}$ is finite.\\
  We write $A\fin\Int^X$ or $A\fin\nat^X$ to indicate that $A$ is a finite (signed)
  multiset over $X$.
\item Any function $f:X\rightarrow\Int$ or $f:X\rightarrow\Int^Y$ from
  $X$ to either the integers or the signed multisets over some set $Y$
  extends to the finite signed multisets $A$ over $X$ by $f(A)=\sum_{x\in X}A(x)\cdot f(x)$.
\end{enumerate}
\end{defi}
\noindent
Two signed multisets $A\!:X \rightarrow \Int$ and $B\!:Y\rightarrow \Int$
are \emph{extensionally equivalent} iff
$A\restrictedto (X\cap Y) = B\restrictedto (X\cap Y)$,
$A\restrictedto (X\setminus Y) = \emptyset$, and
$B \restrictedto (Y\setminus X) = \emptyset$.
In this paper we often do not distinguish extensionally equivalent
signed multisets. This enables us, for instance, to use $A + B$ even
when $A$ and $B$ have different underlying domains.
A multiset $A$ with $A(x) \in\{0,1\}$ for all $x$ is
identified with the set $\{x \mid A(x)=1\}$.
A signed multiset with elements $x$ and $y$, having
multiplicities $-2$ and $3$, is denoted as $-2\cdot\{x\}+3\cdot\{y\}$.

We consider here general labelled place/transition systems with arc weights. Arc weights
are not necessary for the results of the paper, but are included for the sake of generality.

\begin{defi}\label{df-Petri net}
  Let \Act{} be a set of \emph{visible actions} and
  $\tau\mathbin{\not\in}\Act$ be an \emph{invisible action}. Let $\Act_\tau:=\Act \dcup \{\tau\}$.
  A (\emph{labelled}) \defitem{Petri net} (\emph{over $\Act_\tau$}) is a tuple
  $N = (S, T, F, M_0, \ell)$ where
  \begin{enumerate}[$-$]
    \item $S$ and $T$ are disjoint sets (of \defitem{places} and \defitem{transitions}, together
      called the \emph{elements} of $N$),
    \item $F: (S \times T \cup T \times S) \rightarrow \nat$
      (the \defitem{flow relation} including \defitem{arc weights}),
    \item $M_0 : S \rightarrow \nat$ (the \defitem{initial marking}), and
    \item \plat{$\ell: T \into \Act_\tau$} (the \defitem{labelling function}).
  \end{enumerate}
\end{defi}

\noindent
Petri nets are depicted by drawing the places as circles and the
transitions as boxes, containing their label.  Identities of places
and transitions are displayed next to the net element.  When
$F(x,y)>0$ for $x,y \inp S\cup T$ there is an arrow (\defitem{arc})
from $x$ to $y$, labelled with the \emph{arc weight} $F(x,y)$.
Weights 1 are elided.  When a Petri net represents a concurrent
system, a global state of this system is given as a \defitem{marking},
a multiset $M$ of places, depicted by placing $M(s)$ dots
(\defitem{tokens}) in each place $s$.  The initial state is $M_0$.

The behaviour of a Petri net is defined by the possible moves between
markings $M$ and $M'$, which take place when a finite multiset $G$ of
transitions \defitem{fires}.  In that case, each occurrence of a
transition $t$ in $G$ consumes $F(s,t)$ tokens from each 
place $s$.  Naturally, this can happen only if $M$ makes all these
tokens available in the first place.  Next, each $t$ produces $F(t,s)$ tokens
in each $s$.  \refdf{firing} formalises this notion of behaviour.

\begin{defi}\label{df-preset}
Let $N = (S, T, F, M_0, \ell)$ be a Petri net and $x\inp S\cup T$.\\
The multisets $\precond{x},~\postcond{x}: S\cup T \rightarrow
\nat$ are given by $\precond{x}(y)=F(y,x)$ and
$\postcond{x}(y)=F(x,y)$ for all $y \inp S \cup T$.
If $x\in T$, the elements of $\precond{x}$ and $\postcond{x}$ are
called \emph{pre-} and \emph{postplaces} of $x$, respectively, and if
$x\in S$ we speak of \emph{pre-} and \emph{posttransitions}.
The \emph{token replacement function} $\marking{\_\!\_}:T\rightarrow \Int^S$
is given by $\marking{t}=\postcond{t}-\precond{t}$ for all $t\in T$.
These functions extend to finite signed multisets
as usual (see \refdf{multiset}).
\end{defi}

\begin{defi}\label{df-firing}
Let $N \mathbin= (S, T, F, M_0, \ell)$ be a Petri net,
$G \inp \nat^T\!$, $G$ non-empty and finite, and $M, M' \in \nat^S$.\\
$G$ is a \defitem{step} from $M$ to $M'$,
written \plat{$M~[G\rangle_N~ M'$}, iff
\begin{enumerate}[$-$]
  \item $\precond{G} \leq M$ ($G$ is \defitem{enabled}) and
  \item $M' = (M - \precond{G}) + \postcond{G} = M + \marking{G}$.
\end{enumerate}
\end{defi}
\noindent
Note that steps are (finite) multisets, thus allowing self-concurrency,
\ie the same transition can occur multiple times in a single step.
We write $M~[t\rangle_N~ M'$ for $M\mathrel{[\{t\}\rangle_N} M'$, whereas
$M [G\rangle_N$ abbreviates $\exists M'.~ M \mathrel{[G\rangle_N} M'$.
We may omit the subscript $N$ if clear from context.

In our nets transitions are labelled with \emph{actions} drawn from a
set \plat{$\Act \dcup \{\tau\}$}. This makes it possible to see these
nets as models of \defitem{reactive systems} that interact with their
environment. A transition $t$ can be thought of as the occurrence of
the action $\ell(t)$. If $\ell(t)\inp\Act$, this occurrence can be
observed and influenced by the environment---we call such transitions
\defitem{external} or \defitem{visible}, but if $\ell(t)\mathbin=\tau$,
it cannot and $t$ is an \defitem{internal} or \defitem{silent} transition.
Transitions whose occurrences cannot be distinguished by the
environment carry the same label. In particular, since
the environment cannot observe the occurrence of internal
transitions at all, they are all labelled $\tau$.

The labelling function $\ell$ extends to finite signed multisets of transitions $G\in\Int^T$
by $\ell(G):=\sum_{t\in T}G(t)\cdot\{\ell(t)\}$. For $A,B\in\Int^{\Act_\tau}$
we write $A\equiv B$ iff $\ell(A)(a)=\ell(B)(a)$ for all $a\in\Act$, i.e.\ iff $A$ and $B$
contain the same (numbers of) visible actions, allowing $\ell(A)(\tau)\neq \ell(B)(\tau)$.
Hence $\ell(G)\equiv\emptyset$ indicates that $\ell(t)=\tau$ for all transitions $t\in T$ with $G(t)\neq 0$.

\begin{defi}\label{df-onesafe}
  Let $N = (S, T, F, M_0, \ell)$ be a Petri net.
\begin{enumerate}[$-$]
\item
  The set $[M_0\rangle_N$ of \defitem{reachable markings of $N$} is defined as the
  smallest set containing $M_0$ that is closed under $[G\rangle_N$, meaning that if
  $M \inp [M_0\rangle_N$ and $M \mathrel{[G\rangle_N} M'$ then $M' \inp [M_0\rangle_N$.
\item
  $N$ is \defitem{one-safe}
  iff $M \in [M_0\rangle_N \implies \forall s \in S.~ M(s) \leq 1$.
\item
  The \defitem{concurrency relation} $\mathord{\concurrent} \subseteq
  T^2$ is given by $t \concurrent u \equivalent \exists
  M \inp [M_0\rangle.~ M [\{t\}\mathord+\{u\}\rangle$.
\item
  $N$ is a \hypertarget{scn}{\defitem{structural conflict net}} iff
  for all $t,u\in T$ with $t\smile u$ we have $\precond{t} \cap \precond{u} = \emptyset$.
\end{enumerate}
\end{defi}
\noindent
We use the term \hypertarget{plain}{\defitem{plain nets}} for Petri nets where $\ell$ is
injective and no transition has the label $\tau$, \ie essentially unlabelled nets. 

This paper first of all aims at studying finite Petri nets: nets with finitely many places
and transitions. Additionally, our work also applies to infinite nets with the properties that
$\precond{t} \ne \varnothing$ for all transitions $t\in T$, and
any reachable marking (a) is finite, and (b) enables only finitely many transitions.
Henceforth, we call such nets \hypertarget{finitary}{\emph{finitary}}.
Finitariness can be ensured by requiring $|M_0| \mathbin< \infty \wedge \forall t \in T.\,
\precond{t} \ne \varnothing \wedge \forall x \in S\cup T.\, |\postcond{x}| < \infty$, \ie
that the initial marking is finite, no transition has an empty set of preplaces, and each
place and transition has only finitely many outgoing arcs. Our characterisation
of distributability pertains to finitary plain structural conflict nets, and our
distributed implementations are again structural conflict nets,
but they need not be finitary (nor plain). However,
our distributed implementations of finite nets are again finite.

\section{Semantic Equivalences}\label{sec-equivalences}\enlargethispage{\baselineskip}

In this section, we give an overview on some semantic equivalences for reactive systems. Most of these may be defined  formally for Petri nets in a uniform way, by first defining equivalences for transition systems and then associating different transition systems with a Petri net. This yields in particular different non-interleaving equivalences for Petri nets.

\newcommand{\lts}{\mathfrak{L}}
\newcommand{\st}{\mathfrak{S}}
\newcommand{\tr}{\mathfrak{T}}
\newcommand{\inist}{\mathfrak{M_0}}
\newcommand{\mm}{\mathfrak{M}}
\newcommand{\act}{\mathfrak{Act}}

\begin{defi}\label{df-LTS}
Let $\act$ be a set of \emph{visible actions} and
$\tau\mathbin{\not\in}\act$ be an \emph{invisible action}. Let $\act_\tau:=\act \mathbin{\dcup} \{\tau\}$.
A \emph{labelled transition system} (LTS) (\emph{over $\act_\tau$}) is a triple
$(\st,\tr,\inist)$ with
\begin{enumerate}[$-$]
\item $\st$ a set of \emph{states},
\item $\tr\subseteq \st\times \act_\tau \times \st$ a \emph{transition relation}
\item and $\inist\in\st$ the \emph{initial state}.
\end{enumerate}
\end{defi}
\noindent
Given an LTS $(\st,\tr,\inist)$ with $\mm,\mm'\in\st$ and $\alpha\in\act_\tau$,
we write $\mm \goesto[\alpha] \mm'$ for $(\mm,\alpha,\mm')\in \tr$.
We write $\mm \goesto[\alpha]$ for $\exists \mm'.~ \mm \goesto[\alpha] \mm'$ and
$\mm \arrownot\goesto[\alpha]$ for $\nexists \mm'.~ \mm \goesto[\alpha] \mm'$.
Furthermore, $\mm \goesto[\opt{\alpha}] \mm'$ denotes
$\mm \goesto[\alpha] \mm' \vee (\alpha\mathbin=\tau \wedge \mm\mathbin=\mm')$,
meaning that in case \mbox{$\alpha\mathbin=\tau$} performing a $\tau$-transition is optional.
      For $\,a_1 a_2 \cdots a_n \in \act^*$ we write
      $\mm \Goesto[\,a_1 a_2 \cdots a_n~] \mm'$ when
      \[
      \mm
      \Goesto \production{a_1}
      \Goesto \production{a_2}
      \Goesto \cdots
      \Goesto \production{a_n}
      \Goesto
      \mm'
      \]
      where $\Goesto$ denotes the reflexive and transitive closure of $\goesto[\tau]$.
  A state $\mm \in \st$ is said to be \defitem{reachable} iff there is a
  $\sigma \in \act^*$ such that $\inist \Production{\sigma} \mm$. The set of all
  reachable states is denoted by $[\inist\rangle$.
  In case there is an infinite sequence of states
  $(\mathfrak{M}^k)_{k\in\nats}$ such that $\mathfrak{M}^0\in [\inist\rangle$ and
  $\mathfrak{M}^k \goesto[\tau] \mathfrak{M}^{k+1}$ for all $k\in\nat$,
  the LTS is said to display \emph{divergence}.

Many semantic equivalences on LTSs that in some way abstract from internal transitions are
defined in the literature; an overview can be found in \cite{vanglabbeek93linear}.  On
divergence-free LTSs, the most discriminating semantics in the spectrum of equivalences of
\cite{vanglabbeek93linear}, and the only one that fully respects the branching structure
of related systems, is \emph{branching bisimilarity}, proposed in \cite{GW89}.

\begin{defi}\label{df-branching LTS}
Two LTSs $(\st_1,\tr_1,\inist_1)$ and $(\st_2,\tr_2,\inist_2)$ are
\emph{branching bisimilar} iff there exists a relation $\Rel
\subseteq \st_1 \times \st_2$---a \emph{branching bisimulation}---such
that, for all $\alpha\inp\act_\tau$:
\begin{enumerate}[1.]
\item $\mathfrak{M_0}_1\Rel \mathfrak{M_0}_2$;
\item if $\mathfrak{M}_1\Rel \mathfrak{M}_2$ and
  $\mathfrak{M}_1\!\goesto[\alpha]\mathfrak{M}'_1$
  then $\exists \mathfrak{M}^\dagger_2,\mathfrak{M}'_2$ such that
  $\mathfrak{M}_2\Goesto[] \mathfrak{M}^\dagger_2 \!\goesto[\opt{\alpha}] \mathfrak{M}'_2$,
  ~$\mathfrak{M}_1\Rel \mathfrak{M}^\dagger_2$ and $\mathfrak{M}'_1\Rel \mathfrak{M}'_2$;
\item if $\mathfrak{M}_1\Rel \mathfrak{M}_2$ and
  $\mathfrak{M}_2\!\goesto[\alpha]\mathfrak{M}'_2$
  then $\exists \mathfrak{M}^\dagger_1,\mathfrak{M}'_1$ such that
  $\mathfrak{M}_1\Goesto[] \mathfrak{M}^\dagger_1 \!\goesto[\opt{\alpha}] \mathfrak{M}'_1$,
  ~$\mathfrak{M}^\dagger_1\Rel \mathfrak{M}_2$ and $\mathfrak{M}'_1\Rel \mathfrak{M}'_2$.
\end{enumerate}
\end{defi}
\noindent
Branching bisimilarity with explicit divergence \cite{vanglabbeek93linear,GW96,GLT09} is a variant of
branching bisimilarity that fully respects the diverging behaviour of related systems.
It is the most discriminating semantics in the spectrum of equivalences of \cite{vanglabbeek93linear}.
\begin{defi}\label{df-explicit divergence}
Two LTSs $(\st_1,\tr_1,\inist_1)$ and $(\st_2,\tr_2,\inist_2)$ are
branching bisimilar \emph{with explicit divergence} iff there exists a branching bisimulation
$\Rel \subseteq \st_1 \mathbin\times \st_2$ such that furthermore
\begin{enumerate}[1.]
\item[4.] if $\mathfrak{M}_1\Rel \mathfrak{M}_2$ and there is an infinite sequence of states
  $(\mathfrak{M}_1^k)_{k\in\nats}$ such that $\mathfrak{M}_1=\mathfrak{M}_1^0$,
  $\mathfrak{M}_1^k \goesto[\tau] \mathfrak{M}_1^{k+1}$ and $\mathfrak{M}_1^k \Rel \mathfrak{M}_2$
  for all $k\in\nat$, then there exists an infinite sequence of states
  $(\mathfrak{M}_2^\ell)_{\ell\in\nats}$ such that $\mathfrak{M}_2=\mathfrak{M}_2^0$,
  $\mathfrak{M}_2^\ell \goesto[\tau] \mathfrak{M}_2^{\ell+1}$ for all $\ell\in\nat$,
  and $\mathfrak{M}_1^k \Rel \mathfrak{M}_2^\ell$ for all $k,\ell\in\nat$;
\item[5.] if $\mathfrak{M}_1\Rel \mathfrak{M}_2$ and there is an infinite sequence of states
  $(\mathfrak{M}_2^\ell)_{\ell\in\nats}$ such that $\mathfrak{M}_2=\mathfrak{M}_2^0$,
  $\mathfrak{M}_2^\ell \goesto[\tau] \mathfrak{M}_2^{\ell+1}$ and $\mathfrak{M}_1 \Rel \mathfrak{M}_2^\ell$
  for all $\ell\in\nat$, then there exists an infinite sequence of states
  $(\mathfrak{M}_1^k)_{k\in\nats}$ such that $\mathfrak{M}_1=\mathfrak{M}_1^0$,
  $\mathfrak{M}_1^k \goesto[\tau] \mathfrak{M}_1^{k+1}$ for all $k\in\nat$,
  and $\mathfrak{M}_1^k \Rel \mathfrak{M}_2^\ell$ for all $k,\ell\in\nat$.
\end{enumerate}
\end{defi}

\noindent
Since in this paper we mainly compare systems of which one admits no divergence at all, the
definition simplifies to the requirement that the other system may not diverge either.

\begin{prop}\label{pr-explicit divergence}
Let $\mathfrak{L}_1,\,\mathfrak{L}_2$ be two LTSs, of which $\mathfrak{L}_2$ does not display divergence.
Then $\mathfrak{L}_1$ and $\mathfrak{L}_2$ are branching bisimilar with explicit divergence
iff $\mathfrak{L}_1$ and $\mathfrak{L}_2$ are branching bisimilar and
$\mathfrak{L}_1$ does not display divergence either.
\end{prop}
\begin{proof}
``If'': In case neither $\mathfrak{L}_1$ nor $\mathfrak{L}_2$ display divergence,
any branching bisimulation $\Rel$ between $\mathfrak{L}_1$ and $\mathfrak{L}_2$,
when restricted to the reachable states of $\mathfrak{L}_1$ and $\mathfrak{L}_2$,
trivially satisfies Clauses 4 and 5 above.

``Only if'': Suppose that $\Rel$ is a branching bisimulation between
$\mathfrak{L}_1=(\st_1,\tr_1,\inist_1)$ and $\mathfrak{L}_2=(\st_2,\tr_2,\inist_2)$ that satisfies Clauses 4 and 5 above,
and suppose $\mathfrak{L}_1$ displays divergence, \ie there is an infinite sequence of states
$(\mathfrak{M}_1^k)_{k\in\nats}$ such that $\mathfrak{M_0}_1\Goesto[\sigma]\mathfrak{M}_1^0$
for some $\sigma\in\mathfrak{Act}^*$ and $\mathfrak{M}_1^k \goesto[\tau] \mathfrak{M}_1^{k+1}$
for all $k\in\nat$. By \refdf{branching LTS}, Clauses 1 and 2, there exists an infinite sequence of states
  $(\mathfrak{M}_2^k)_{k\in\nats}$ such that $\mathfrak{M_0}_2\Goesto[\sigma]\mathfrak{M}_2^0$,
  $\mathfrak{M}_2^k \Goesto \mathfrak{M}_2^{k+1}$
  and $\mathfrak{M}_1^k \Rel \mathfrak{M}_2^k$ for all $k\in\nat$.
  In case infinitely many of those $\mathfrak{M}_2^k$ are different,
  this sequence constitutes a divergence of $\mathfrak{L}_2$.
  Otherwise, there is an $k_0\geq 0$ such that all $\mathfrak{M}_2^k$ for $k\geq k_0$ are equal,
  and then $\mathfrak{L}_2$ has a divergence by Clause 4.
\end{proof}

One of the semantics reviewed in \cite{vanglabbeek93linear} that respects branching time
and divergence only to a minimal extent, is \emph{(stable) failures equivalence},
proposed in \cite{BKO87} and further elaborated in \cite{Roscoe98}.
It is a variant of the failures equivalence of \cite{BHR84}, only differing in the
treatment of divergence.\footnote{When comparing two systems without divergence, the
  stable failure equivalence coincides with the failures equivalence of \cite{BHR84}.
  When comparing systems of which one is known to be divergence-free---as we will do in
  this paper---the stable failures semantics is strictly less discriminating than the
  failures equivalence of \cite{BHR84}---only the latter guarantees that the other system
  is divergence-free as well. As a less discriminating equivalence will give rise to
  stronger results about the absence of distributed implementations of certain systems, we
  will use a version of the stable failures equivalence, rather than of the failures
  equivalence from \cite{BHR84}.}
\begin{defi}\label{df-failures}
  Let $\lts = (\st,\tr,\inist)$ be an LTS, $\sigma \in \act^*$ and $X \subseteq \act$,\
  $X$ finite.\footnote{Although the version without the restriction that $X$ be finite has
    arguably better properties, we here use the version with this restriction---the
    \emph{finite failures equivalence} of \cite{vanglabbeek93linear}---since it is less
    discriminating.}\\
  $\sigma$ is a \defitem{trace} of $\lts$ iff $\exists \mm.~ \inist \Production{\sigma} \mm$.\\
$\rpair{\sigma, X}$ is a \defitem{failure pair} of $\lts$ iff
  $\exists \mm.~ \inist \Production{\sigma} \mm \wedge \mm \arrownot\production{\tau}
  \wedge \, \forall a\inp X.~ \mm \arrownot\goesto[a].$\\
  We write $\mathfrak{T}(\lts)$ for the set of all traces, and $\mathfrak{F}(\lts)$ for
  the set of all failure pairs of $\lts$.\\
  Two LTSs $\lts_1$ and $\lts_2$ are \defitem{failures equivalent}
  iff $\mathfrak{T}(\lts_1) = \mathfrak{T}(\lts_2)$ and $\mathfrak{F}(\lts_1) = \mathfrak{F}(\lts_2)$.
\end{defi}

As indicated in \cite{vanglabbeek01refinement}, see in particular the
diagram on Page 317 (or 88), equivalences on LTSs have been ported to
Petri nets and other causality respecting models of concurrency
chiefly in five ways: we distinguish \emph{interleaving semantics},
\emph{step semantics}, \emph{split semantics}, \emph{ST-semantics} and
\emph{causal semantics}.  Causal semantics fully respect the causal
relationships between the actions of related systems, whereas
interleaving semantics fully abstract from this information.  Step
semantics differ from interleaving semantics by taking into account
the possi\-bility of multiple actions to occur simultaneously (in
\emph{one step}); this carries a minimal amount of causal information.
ST-semantics respect causality to the extent that it can be expressed
in terms of the possibility of durational actions to overlap in
time. They are formalised by executing a visible action $a$ in two
phases: its start $a^+$ and its termination $a^-$.  Moreover,
terminating actions are properly matched with their starts. Split
semantics are a simplification of ST-semantics in which the matching
of starts and terminations is dropped.

Interleaving semantics on Petri nets can be formalised by associating to each net
$N=(S,T,F,M_0,\ell)$ the LTS $(\st,\tr,M_0)$ with $\st$ the set of markings of $N$
and $\tr$ given by 
$$M_1 \production{\alpha} M_2 :\equivalent \exists\, t \mathbin\in T
      .~ \ell(t) \mathbin= \alpha \wedge M_1~[t\rangle~ M_2.$$
Here we take $\act := \Act$.
Now each equivalence on LTSs from \cite{vanglabbeek93linear} induces a corresponding
interleaving equivalence on nets by declaring two nets equivalent iff the associated LTSs are.
For example, \emph{interleaving branching bisimilarity} is the relation of \refdf{branching LTS}
with the $\mm$'s denoting markings, and the $\alpha$'s actions from $\Act_\tau$.

Step semantics on Petri nets can be formalised by associating another LTS to each net.
Again we take $\st$ to be the markings of the net, and $\inist$ the initial marking,
but this time $\act$ consists of the \emph{steps} over $\Act$,
the non-empty, finite multisets $A$ of visible actions from $\Act$,
and the transition relation $\tr$ is given by
      $$M_1 \production{A} M_2 :\equivalent \exists\, G \fin\nat^T
      .~ \ell(G) = A \wedge \tau \notin \ell(G) \wedge M_1~[G\rangle~ M_2$$
with $\tau$-transitions defined just as in the interleaving case:
$$
M_1 \production{\tau} M_2 :\equivalent \exists\, t \mathbin\in T
      .~ \ell(t) \mathbin= \tau \wedge M_1~[t\rangle~ M_2.
$$
In particular, the step version of failures equivalence would be the relation of \refdf{failures}
with the $\mm$'s denoting markings, the $a$'s steps over $\Act$, the $X$'s sets of
steps, and the $\sigma$'s sequences of steps.
This form of step failures semantics, but based on the failures semantics of \cite{BHR84}
rather than the stable failures semantics of \refdf{failures}, has been studied in \cite{TaubnerV89}.
However, variations in this type of definition are possible.
In this paper we employ a form of step failures semantics that is 
a bit closer to interleaving semantics, thereby coarsening the equivalence and strengthening the final result: $\sigma$ is
a sequence of single actions, whereas the set $X$ of impossible continuations after
$\sigma$ is a set of steps. Moreover, we drop the comparison of the sets of traces.
We define this notion directly on Petri nets, without using intermediate LTSs.

\begin{defi}\label{df-step failures}
  {Let $N = (S, T, F, M_0, \ell)$ be a Petri net, $\sigma \in \Act^*$ and
  $X \subseteq \powermultiset{\Act}$, $X$ finite.}\\
  $\rpair{\sigma, X}$ is a \defitem{step failure pair} of $N$ iff\vspace{-3pt}
  $$\exists M. M_0 \Production{\sigma} M \wedge M \arrownot\production{\tau}
  \wedge \, \forall A\inp X.~ M \arrownot\goesto[A].$$\\[-3ex]
  We write $\failureset(N)$ for the set of all step failure pairs of $N$.\\
  Two Petri nets $N_1$ and $N_2$ are \defitem{step failures equivalent},
  $N_1 \approx_\mathscr{F} N_2$, iff $\failureset(N_1) = \failureset(N_2)$.
\end{defi}

Next we propose a general definition on Petri nets of ST-versions of each of the semantics
of \cite{vanglabbeek93linear}. Again we do this through a mapping from nets to a suitable LTS\@.
An \emph{ST-marking} of a net $(S,T,F,M_0,\ell)$ is a pair $(M,U)\inp\nat^S \mathord\times T^*$
of a normal marking, together with a sequence of visible transitions \emph{currently firing}.
The \emph{initial} ST-marking is $\mathfrak{M_0}:=(M_0,\epsilon)$.
The elements of $\Act^\pm:=\{a^+,\, a^{-n} \mid a \inp \Act, ~n\mathbin> 0\}$ are called
\emph{visible action phases}, and \plat{$Act^\pm_\tau:=\Act^\pm\dcup\{\tau\}$}.
For $U\in T^*$, we write $t\in^{(n)}U$ if $t$ is the  \plat{$n^{\it th}$}
element of $U$. Furthermore $U^{-n}$ denotes $U$ after removal of the \plat{$n^{\it th}$}
transition.

\begin{defi}\label{df-ST-marking}
Let $N=(S,T,F,M_0,\ell)$ be a Petri net, labelled over \plat{$\Act_\tau$}.

The \emph{ST-transition relations} $\goesto[\eta]$ for $\eta\inp\Act^\pm_\tau$ between ST-markings are given by\vspace{1pt}

$(M,U)\goesto[a^+](M',U')$ iff $\exists t\inp T.~ \ell(t)=a \wedge M[t\rangle
\wedge M'=M-\precond{t} \wedge U'=U t$.

$(M,U)\goesto[a^{-n}](M',U')$ iff $\exists t\in^{(n)} U.~\ell(t)=a \wedge
 U'=U^{-n} \wedge M'=M+\postcond{t}$.

$(M,U)\goesto[\tau](M',U')$ iff $M\goesto[\tau]M' \wedge U'=U$.
\end{defi}
\noindent
Now the ST-LTS associated to a net $N$ is $(\st,\tr,\inist)$ with $\st$ the set of
ST-markings of $N$, $\act:=\Act^\pm$, $\tr$ as defined in \refdf{ST-marking}, and $\inist$ the initial ST-marking.
Again, each equivalence on LTSs from \cite{vanglabbeek93linear} induces a corresponding
ST-equivalence on nets by declaring two nets equivalent iff their associated LTSs are.
In particular, \emph{branching ST-bisimilarity} is the relation of \refdf{branching LTS}
with the $\mm$'s denoting ST-markings, and the $\alpha$'s action phases from $\Act^\pm_\tau$.
We write $N_1\approx^\Delta_{bSTb}N_2$ iff $N_1$ and $N_2$ are branching ST-bisimilar with explicit divergence.

\emph{ST-bisimilarity} was originally proposed in \cite{GV87}. It was extended to a setting
with internal actions in \cite{Vo93}, based on the notion of \emph{weak bisimilarity} of
\cite{Mi89}, which is a bit less discriminating than branching bisimilarity.
The above can be regarded as a reformulation of the same idea; the notion of weak
ST-bisimilarity defined according to the recipe above agrees with the ST-bisimilarity of \cite{Vo93}.

The next proposition says that branching ST-bisimilarity with explicit
divergence is more discriminating than (\ie \emph{stronger} than,
\emph{finer} than, or included in) step failures equivalence.

\begin{prop}\label{pr-step failures ST}
Let $N_1$ and $N_2$ be Petri nets. If $N_1\approx^\Delta_{bSTb}N_2$ then $N_1 \approx_\mathscr{F} N_2$.
\end{prop}
\begin{proof}
Suppose $N_1\approx^\Delta_{bSTb}N_2$ and $\rpair{\sigma, X}\in\failureset(N_1)$.
By symmetry it suffices to show that $\rpair{\sigma, X}\in\failureset(N_2)$.

Since $N_1\approx^\Delta_{bSTb}N_2$,
there must be a branching bisimulation $\Rel$ between the ST-markings of
$N_1=(S_1,T_1,F_1,{M_0}_1,\ell_1)$ and $N_2=(S_2,T_2,F_2,{M_0}_2,\ell_2)$.
In particular, $({M_0}_1,\epsilon)\Rel({M_0}_2,\epsilon)$.
Let $\sigma =: a_1 a_2 \cdots a_n \in \Act^*$.\vspace{-2pt}
Then \plat{$
      {M_0}_1
      \Goesto \production{a_1}
      \Goesto \production{a_2}
      \Goesto \cdots
      \Goesto \production{a_n}
      \Goesto
      M'_1
      $}
for a marking $M'_1\inp\nat^{S_1}$ with
$M'_1 \arrownot\production{\tau}$ and $\forall A\inp X.~ M'_1 \arrownot\goesto[A]$.
So\vspace{-4pt} $
      ({M_0}_1,\epsilon)
      \Goesto \production{a_1^+}\production{a_1^{-1}}
      \Goesto \production{a_2^+}\production{a_2^{-1}}
      \Goesto \cdots
      \Goesto \production{a_n^+}\production{a_n^{-1}}
      \Goesto
      (M'_1,\epsilon)
      $.
Thus, using the properties of a branching bisimulation on the ST-LTSs associated to $N_1$
\vspace{-3pt}and $N_2$, there must be a marking \plat{$M'_2\inp\nat^{S_2}$} such that $
      ({M_0}_2,\epsilon)\!
      \Goesto \production{a_1^+} \Goesto \production{a_1^{-1}}
      \Goesto \production{a_2^+} \Goesto \production{a_2^{-1}}
      \Goesto \cdots
      \Goesto \production{a_n^+} \Goesto \production{a_n^{-1}}
      \Goesto\!
      (M'_2,\epsilon)
      $
and $(M'_1,\epsilon)\Rel(M'_2,\epsilon)$.
Since \plat{$(M'_1,\epsilon) \arrownot\production{\tau}$}, the ST-marking $(M'_1,\epsilon)$
admits no divergence. As \plat{$\approx^\Delta_{bSTb}$} respects this property (cf.\ the proof
of \refpr{explicit divergence}), also $(M'_2,\epsilon)$
admits no divergence, and there must be an $M''_2\inp\nat^{S_2}$ with \plat{$M''_2 \arrownot\production{\tau}$}
and $(M'_2,\epsilon)\Goesto(M''_2,\epsilon)$. Clause 3.\ of a branching bisimulation
gives $(M'_1,\epsilon)\Rel(M''_2,\epsilon)$, and
\refdf{ST-marking} yields ${M_0}_2 \Goesto[\sigma] M''_2$.\vspace{2pt}
Here we use that if \plat{$(M,U)\production{a^{-1}}\production{\tau}(M',U')$} then
\plat{$(M,U)\production{\tau}\production{a^{-1}}(M',U')$}.

Now let $B=\{b_1,\ldots,b_m\}\in X$. Then \plat{$M'_1\arrownot\production{B}$}.\vspace{-4pt}
Suppose, towards a contradiction, that $M''_2\production{B}$.
Then $(M''_2,\epsilon)\production{b_1^+}\production{b_2^+}\cdots \production{b_m^+}$.\vspace{-4pt}
Property 2.\ of a branching bisimulation implies
$(M'_1,\epsilon)\production{b_1^+}\production{b_2^+}\cdots \production{b_m^+}$
and hence $M'_1\production{B}$. This is a contradiction, so $M''_2 \arrownot\production{B}$.
It follows that $\rpair{\sigma, X}\in\failureset(N_2)$.
\end{proof}

In this paper we employ both step failures equivalence and 
branching ST-bisimilarity with explicit divergence.
Fortunately it will turn out that for our purposes the latter equivalence coincides
with its split version
(since always one of the compared nets is plain, see \refpr{split}).

A \emph{split marking} of a net $N=(S,T,F,M_0,\ell)$ is a pair $(M,U)\in\nat^S \times\nat^T$
of a normal marking $M$, together with  a multiset of visible transitions currently firing.
The \emph{initial} split marking is $\mathfrak{M_o}:=(M_0,\emptyset)$.
A split marking can be regarded as an abstraction from an ST-marking, in which the total
order on the (finite) multiset of transitions that are currently firing has been dropped.
Let $\Act^\pm_{\rm split}:=\{a^+,\, a^- \mid a \in \Act\}$.

\begin{defi}\label{df-split marking}
Let $N=(S,T,F,M_0,\ell)$ be a Petri net, labelled over $\Act_\tau$.
\\
The \emph{split transition relations} $\goesto[\zeta]$ for $\zeta\inp\Act^\pm_{\rm split}\dcup\{\tau\}$
between split markings are given by

$(M,U)\goesto[a^+](M',U')$ iff $\exists t\inp T.~ \ell(t)=a \wedge M[t\rangle
\wedge M'=M-\precond{t} \wedge U'=U + \{t\}$.

$(M,U)\goesto[a^{-}](M',U')$ iff $\exists t\inp U.~\ell(t)=a \wedge
 U'=U-\{t\} \wedge M'=M+\postcond{t}$.

$(M,U)\goesto[\tau](M',U')$ iff $M\goesto[\tau]M' \wedge U'=U$.
\end{defi}
\noindent
Note that $(M,U)\goesto[a^+]$ iff $M\goesto[a]$, whereas $(M,U)\goesto[a^-]$
iff $a\in\ell(U)$.
With induction on reachability of markings it is furthermore easy to check that
$(M,U) \in [\inist\rangle$ iff $\tau \notin \ell(U)$ and $M+\!\precond{U}\in [M_0\rangle$.

The split LTS associated to a net $N$ is $(\st,\tr,\inist)$ with $\st$ the set of split
markings of $N$, $\act:=\Act^\pm$, $\tr$ as defined in \refdf{split marking}, and $\inist$ the initial split marking.
Again, each equivalence on LTSs from \cite{vanglabbeek93linear} induces a corresponding
split equivalence on nets by declaring two nets equivalent iff their associated LTSs are.
In particular, \emph{branching split bisimilarity} is the relation of \refdf{branching LTS}
with the $\mm$'s denoting split markings, and the $\alpha$'s action phases from
\plat{$\Act^\pm_{\rm split}\dcup\{\tau\}$}.
\vspace{2pt}

For $\mathfrak{M}=(M,U)\in \nat^S\times T^*$ an ST-marking, let
$\overline{\mathfrak{M}}=(M,\overline{U})\in \nat^S \times \nat^T$ be the split marking obtained by
converting the sequence $U$ into the multiset $\overline{U}$, where $\overline{U}(t)$ is
the number of occurrences of the transition $t\in T$ in $U$.
Moreover, define $\ell(\mathfrak{M})$ by $\ell(M,U) := \ell(U)$ and
$\ell(t_1 t_2 \cdots t_k) := \ell(t_1) \ell(t_2) \cdots \ell(t_k)$.
Furthermore, for $\eta\in\Act^\pm_\tau$, let \plat{$\overline{\eta}\in \Act^\pm_{\rm split}\dcup\{\tau\}$}
be given by \plat{$\overline{a^+}:= a^+$}, \plat{$\overline{a^{-n}}:= a^-$} and $\overline{\tau}:= \tau$.

\begin{obs}\label{obs-match}
Let $\mathfrak{M},\mathfrak{M}'$ be ST-markings, $\mathfrak{M}^\dagger$ a split marking,
$\eta\inp\Act^\pm_\tau$ and {$\zeta\inp\Act^\pm_{\rm split}\linebreak[2]\cup\{\tau\}$}. Then
\begin{enumerate}[(1)]
\item $\mathfrak{M}\in \nat^S\times T^*$ is the initial ST-marking of $N$ iff
$\overline{\mathfrak{M}}\in \nat^S\times \nat^T$ is the initial split marking of $N$;
\item if $\mathfrak{M} \goesto[\eta] \mathfrak{M}'$ then
$\overline{\mathfrak{M}} \goesto[\overline{\eta}] \overline{\mathfrak{M}'}$;
\item if $\overline{\mathfrak{M}} \goesto[\zeta] \mathfrak{M}^\dagger$ then there
is a $\mathfrak{M}'\in \nat^S\times T^*$ and \plat{$\eta\in\Act^\pm_\tau$} such that $\mathfrak{M} \goesto[\eta]
\mathfrak{M}'$, $\overline{\eta}=\zeta$ and $\overline{\mathfrak{M}'}=\mathfrak{M}^\dagger$;
\item if $\mathfrak{M} \goesto[\opt{\eta}] \mathfrak{M}'$ then
$\overline{\mathfrak{M}} \goesto[\opt{\overline{\eta}}] \overline{\mathfrak{M}'}$;
\item if $\overline{\mathfrak{M}} \goesto[\opt{\zeta}] \mathfrak{M}^\dagger$ then there
is a $\mathfrak{M}'\in \nat^S\times T^*$ and \plat{$\eta\in\Act^\pm_\tau$} such that $\mathfrak{M} \goesto[\opt{\eta}]
\mathfrak{M}'$, $\overline{\eta}=\zeta$ and $\overline{\mathfrak{M}'}=\mathfrak{M}^\dagger$;
\item if $\mathfrak{M} \Goesto \mathfrak{M}'$ then
$\overline{\mathfrak{M}} \Goesto \overline{\mathfrak{M}'}$;
\item if $\overline{\mathfrak{M}} \Goesto \mathfrak{M}^\dagger$ then there
is a $\mathfrak{M}'\in \nat^S\times T^*$ such that $\mathfrak{M} \Goesto
\mathfrak{M}'$ and $\overline{\mathfrak{M}'}=\mathfrak{M}^\dagger$.\qed
\end{enumerate}
\end{obs}

\begin{lem}\label{lem-label sequence}
Let $N_1=(S_1,T_1,F_1,{M_0}_1,\ell)$ and $N_2=(S_2,T_2,F_2,{M_0}_2,\ell_2)$ be two nets, $N_2$ being
\hyperlink{plain}{plain};
let $\mathfrak{M}_1,\mathfrak{M}'_1$ be ST-markings of $N_1$, and
$\mathfrak{M}_2,\mathfrak{M}'_2$ ST-markings of $N_2$.
  If $\ell(\mathfrak{M}_2)=\ell(\mathfrak{M}_1)$,
  $\mathfrak{M}_1\goesto[\eta] \mathfrak{M}'_1$ and
  $\mathfrak{M}_2\goesto[\opt{\eta'}] \mathfrak{M}'_2$ with $\overline{\eta'}=\overline{\eta}$,
  then there is an $\mathfrak{M}''_2$ with $\mathfrak{M}_2\goesto[\opt{\eta}] \mathfrak{M}''_2$,
  $\ell(\mathfrak{M}''_2)=\ell(\mathfrak{M}'_1)$,
  and $\overline{\mathfrak{M}''_2}=\overline{\mathfrak{M}'_2}$.
\end{lem}

\begin{proof}
If $\mathfrak{M}\goesto[\eta] \mathfrak{M}'$ or $\mathfrak{M}\goesto[\opt{\eta}] \mathfrak{M}'$
then $\ell( \mathfrak{M}')$ is completely determined by $\ell(\mathfrak{M})$ and $\eta$.
For this reason the requirement
$\ell(\mathfrak{M}''_2)=\ell(\mathfrak{M}'_1)$ will hold as soon as
the other requirements are met. 

First suppose $\eta$ is of the form $\tau$ or $a^+$. Then
$\overline{\eta}=\eta$ and moreover $\overline{\eta'}=\overline{\eta}$ implies $\eta'=\eta$.
Thus we can take $\mathfrak{M}''_2:=\mathfrak{M}'_2$.

Now suppose $\eta:=a^{-n}$ for some $n>0$. Then $\eta'=a^{-m}$ for
some $m>0$. As $\mathfrak{M}_1\goesto[\eta]$, the $n^{\it th}$ element
of $\ell(\mathfrak{M}_1)$ must (exist and) be $a$.
Since $\ell(\mathfrak{M}_2)=\ell(\mathfrak{M}_1)$, also the $n^{\it th}$ element
of $\ell(\mathfrak{M}_2)$ must be $a$, so there is an
$\mathfrak{M}''_2$ with $\mathfrak{M}_2\goesto[\opt{\eta}]
\mathfrak{M}''_2$. Let $\mathfrak{M}_2:=(M_2,U_2)$. Then $U_2$ is a
sequence of transitions of which the $n^{\it th}$ and the $m^{\it th}$ elements
are both labelled $a$. Since the net $N_2$ is plain, those two
transitions must be equal. Let $\mathfrak{M}'_2:=(M'_2,U'_2)$ and $\mathfrak{M''}_2:=(M''_2,U''_2)$.
We find that $M''_2\mathbin=M'_2$ and $\overline{U''_2}\mathbin=\overline{U'_2}$. It
follows that $\overline{\mathfrak{M}''_2}=\overline{\mathfrak{M}'_2}$.
\end{proof}

\begin{obs}\label{obs-label sequence}
If $\mathfrak{M}\Goesto \mathfrak{M}'$ for ST-markings
$\mathfrak{M},\mathfrak{M}'$ then $\ell(\mathfrak{M}')=\ell(\mathfrak{M})$.
\end{obs}

\begin{obs}\label{obs-end-phase determinism 1}
  If $\ell(\mathfrak{M}_1)=\ell(\mathfrak{M}_2)$ and
  $\mathfrak{M}_2\goesto[a^{-n}]$ for some $a\in\Act$ and $n>0$, then $\mathfrak{M}_1\goesto[a^{-n}]$.
\end{obs}

\begin{obs}\label{obs-end-phase determinism 2}
  If $\mathfrak{M}\goesto[a^{-n}]\mathfrak{M}'$ and $\mathfrak{M}\goesto[a^{-n}]\mathfrak{M}''$
 for some $a\in\Act$ and $n>0$, then $\mathfrak{M}'=\mathfrak{M}''$.
\end{obs}

\begin{prop}\label{pr-split}
Let $N_1=(S_1,T_1,F_1,{M_0}_1,\ell)$ and
$N_2=(S_2,T_2,F_2,{M_0}_2,\ell_2)$ be two nets, $N_2$ being \hyperlink{plain}{plain}.
Then $N_1$ and $N_2$ are branching ST-bisimilar (with explicit
divergence) iff they are branching split bisimilar (with explicit divergence).
\end{prop}

\begin{proof}
Suppose $\Rel $ is a branching ST-bisimulation between $N_1$ and $N_2$.
Then, by \refobs{match}, the relation $\Rel _{\rm split} := \{(\overline{\mathfrak{M}_1},\overline{\mathfrak{M}_2}) \mid
(\mathfrak{M}_1,\mathfrak{M}_2)\in \Rel \}$ is a branching split bisimulation between $N_1$ and $N_2$.

Now let $\Rel $ be a branching split bisimulation between $N_1$ and $N_2$.
Then, using \refobs{match}, the relation $\Rel _{\rm ST} := \{(\mathfrak{M}_1,\mathfrak{M}_2) \mid
\ell_1(\mathfrak{M}_1)=\ell_2(\mathfrak{M}_2) \wedge 
(\overline{\mathfrak{M}_1},\overline{\mathfrak{M}_2})\in \Rel \}$ turns out to be a
branching ST-bisimulation between $N_1$ and $N_2$:
\begin{enumerate}[1.]
\item $\mathfrak{M_0}_1\Rel_{\rm ST} \mathfrak{M_0}_2$ follows from \refobs{match}(1), since
  $\overline{\mathfrak{M_0}_1}\Rel \overline{\mathfrak{M_0}_2}$ and
  $\ell(\mathfrak{M_0}_1)\mathbin=\ell(\mathfrak{M_0}_2)\mathbin=\epsilon$.
\item Suppose $\mathfrak{M}_1\Rel_{\rm ST} \mathfrak{M}_2$ and
  $\mathfrak{M}_1\!\goesto[\eta]\mathfrak{M}'_1$.
  Then $\overline{\mathfrak{M}_1}\Rel \overline{\mathfrak{M}_2}$ and
  $\overline{\mathfrak{M}_1}\!\goesto[\overline\eta]\overline{\mathfrak{M}'_1}$.
  Hence $\exists \mathfrak{M}^\dagger_2,\mathfrak{M}^\ddagger_2$ such that
  $\overline{\mathfrak{M}_2}\Goesto[] \mathfrak{M}^\dagger_2 \!\goesto[\opt{\overline\eta}] \mathfrak{M}^\ddagger_2$,
  ~$\overline{\mathfrak{M}_1}\Rel \mathfrak{M}^\dagger_2$ and
  $\overline{\mathfrak{M}'_1}\Rel \mathfrak{M}^\ddagger_2$.
  As $N_2$ is plain, $\mathfrak{M}^\dagger_2=\overline{\mathfrak{M}_2}$.
  By \refobs{match}(5), using that $\overline{\mathfrak{M}_2}\goesto[\opt{\overline\eta}] \mathfrak{M}^\ddagger_2$,
  $\exists \mathfrak{M}'_2,\,\eta'$ such that
  $\mathfrak{M}_2\goesto[\opt{\eta'}] \mathfrak{M}'_2$, $\overline{\eta'}=\overline{\eta}$ and
  $\overline{\mathfrak{M}'_2} = \mathfrak{M}^\ddagger_2$.
  By \reflem{label sequence}, there is an ST-marking $\mathfrak{M}''_2$ such that
  $\mathfrak{M}_2\goesto[\opt{\eta}] \mathfrak{M}''_2$,
  $\ell(\mathfrak{M}''_2)=\ell(\mathfrak{M}'_1)$,
  and $\overline{\mathfrak{M}''_2}=\overline{\mathfrak{M}'_2} = \mathfrak{M}^\ddagger_2$.
  It follows that $\mathfrak{M}'_1\Rel_{\rm ST} \mathfrak{M}''_2$.
\item Suppose $\mathfrak{M}_1\Rel_{\rm ST} \mathfrak{M}_2$ and
  $\mathfrak{M}_2\!\goesto[\eta]\mathfrak{M}'_2$.
  Then $\overline{\mathfrak{M}_1}\Rel \overline{\mathfrak{M}_2}$ and
  $\overline{\mathfrak{M}_2}\!\goesto[\overline\eta]\overline{\mathfrak{M}'_2}$.
  Hence $\exists \mathfrak{M}^\dagger_1,\mathfrak{M}^\ddagger_1$ such that
  $\overline{\mathfrak{M}_1}\Goesto[] \mathfrak{M}^\dagger_1 \!\goesto[\opt{\overline\eta}] \mathfrak{M}^\ddagger_1$,
  ~$\mathfrak{M}^\dagger_1\Rel \overline{\mathfrak{M}_2}$ and
  $\mathfrak{M}^\ddagger_1\Rel \overline{\mathfrak{M}'_2}$.
  By \refobs{match}(7), $\exists \mathfrak{M}^*_1$ such that
  $\mathfrak{M}_1\Goesto[] \mathfrak{M}^*_1$ and
  $\overline{\mathfrak{M}^*_1} = \mathfrak{M}^\dagger_1$.
  By \refobs{label sequence},
  $\ell(\mathfrak{M}^*_1)=\ell(\mathfrak{M}_1)=\ell(\mathfrak{M}_2)$,
  so $\mathfrak{M}^*_1\Rel_{\rm ST} \mathfrak{M}_2$.
  Since $N_2$ is plain, $\eta\neq\tau$.
  \begin{iteMize}{$\bullet$}
  \item Let $\eta=a^+$ for some $a\in\Act$.
  Using that $\overline{\mathfrak{M}^*_1}\goesto[\opt{\overline\eta}] \mathfrak{M}^\ddagger_1$,
  by \refobs{match}(5) $\exists \mathfrak{M}'_1,\,\eta'$ such that
  $\mathfrak{M}^*_1\goesto[\opt{\eta'}] \mathfrak{M}'_1$, $\overline{\eta'}=\overline{\eta}$ and
  $\overline{\mathfrak{M}'_1} = \mathfrak{M}^\ddagger_1$.
  It must be that $\eta'=\eta=a^+$ and
  $\ell(\mathfrak{M}'_1)=\ell(\mathfrak{M}^*_1)a=\ell(\mathfrak{M}_2)a=\ell(\mathfrak{M}'_2)$.
  Hence $\mathfrak{M}'_1\Rel_{\rm ST} \mathfrak{M}'_2$.
\item Let $\eta=a^{-n}$ for some $a\in\Act$ and $n>0$.
  By \refobs{end-phase determinism 1}, $\exists \mathfrak{M}'_1$ with
  $\mathfrak{M}^*_1\goesto[\eta]\mathfrak{M}'_1$.
  By Part 2.\ of this proof, $\exists \mathfrak{M}''_2$ such that
  $\mathfrak{M}_2\goesto[\opt{\eta}] \mathfrak{M}''_2$ and
  $\mathfrak{M}'_1\Rel_{\rm ST} \mathfrak{M}''_2$.
  By \refobs{end-phase determinism 2} $\mathfrak{M}''_2=\mathfrak{M}'_2$.
  \end{iteMize}
\end{enumerate}
Since the net $N_2$ is plain, it has no divergence. In such a case,
the requirement ``with explicit divergence'' requires $N_1$ to be free
of divergence as well, regardless of whether split or ST-semantics is
used.
\end{proof}
\noindent
In this paper we will not consider causal semantics.  The reason is
that our distributed implementations will not fully preserve the
causal behaviour of nets.  We will further comment on this in the
conclusion.

\section{Distributed Systems}
\label{sec-distributed systems}
In this section, we stipulate what we understand by a distributed
system, and subsequently formalise a model of distributed systems in
terms of Petri nets.
\begin{enumerate}[$-$]
\item A distributed system consists of components residing on different locations.
\item Components work concurrently.
\item Interactions between components are only possible by explicit communications.
\item Communication between components is time consuming and asynchronous.
\end{enumerate}
Asynchronous communication is the
only interaction mechanism in a
distributed system for exchanging signals or information.
\begin{enumerate}[$-$]
\item The sending of a message happens always strictly before its receipt (there is a causal relation between sending and receiving a message).
\item A sending component sends without regarding the state of the
  receiver; in particular there is no need to synchronise with a receiving component.
  After sending the sender continues its behaviour independently of receipt of the message.
\end{enumerate}
As explained in the introduction, we will add another requirement to
our notion of a distributed system, namely that its components only
allow sequential behaviour.

\subsection{LSGA nets}

Formally, we model distributed systems as nets consisting of component
nets with sequential behaviour and interfaces in terms of input and
output places.

\begin{defi}\label{df-component}
  Let $N \mathbin= (S, T, F, M_0, \ell)$ be a Petri net,
  $I, O \mathbin\subseteq S$,
  $I\mathop\cap O\mathbin=\emptyset$ and
  $\postcond{O} = \varnothing$.
  \begin{enumerate}[1.]
    \item $(N, I, O)$ is a \defitem{component with interface $(I, O)$}.
    \item $(N, I, O)$ is a \defitem{sequential} component with interface $(I,
      O)$ iff\\ $\exists Q \mathbin\subseteq S \mathord\setminus (I \cup O)$ with
      $\forall t \in T. |\precond{t} \restrictedto Q| = 1 \wedge
          |\postcond{t}\! \restrictedto Q| = 1$ and
        $|M_0 \restrictedto Q| = 1$.
  \end{enumerate}
\end{defi}

\noindent
An input place $i\inp I$ of a component $\mathcal{C}\mathbin=(N,I,O)$ can be regarded as a mailbox of
$\mathcal{C}$ for a specific type of messages.  An output place $o\inp
O$, on the other hand, is an address outside $\mathcal{C}$ to which $\mathcal{C}$
can send messages. Moving a token into $o$ is like posting a
letter. The condition $\postcond{o}=\varnothing$ says that a message,
once posted, cannot be retrieved by the component.\footnote{%
We could have required that $\precond{I}=\varnothing$, thereby
disallowing a component to put messages in its own mailbox.
This would not lead to a loss of generality in the class of distributed systems
that can be obtained as the asynchronous parallel composition
of sequential components, defined below. However, this property
is not preserved under asynchronous parallel composition (defined below),
and we like the composition of a set of (sequential) components to be
a component itself (but not a sequential one).}

A set of places like $Q$ above is a special case of an \emph{$S$-invariant}.
The requirements guarantee that the number of tokens in these places
remains constant, in this case $1$. It follows that no two transitions
can ever fire concurrently (in one step).
Conversely, whenever a net is sequential, in the sense that no two
transitions can fire in one step, it is easily converted into a
behaviourally equivalent net with the required $S$-invariant, namely by
adding a single marked place with a self-loop to all transitions.
This modification preserves virtually all semantic equivalences on
Petri nets from the literature, including $\approx^\Delta_{bSTb}$.

Next we define an operator for combining components with asynchronous
communication by fusing input and output places.

\begin{defi}\label{df-parcomp}
  Let $\indexset$ be an index set.\\
  Let $((S_k, T_k, F_k, {M_0}_k, \ell_k), I_k, O_k)$ with $k \in \indexset$
  be components with interface such that
  $(S_k \cup T_k) \cap (S_l \cup T_l) = (I_k \cup O_k) \cap (I_l \cup O_l)$
  for all $k, l \in \indexset$ with $k\neq l$
  (components are disjoint except for interface places)
  and $I_k \cap I_l = \varnothing$ for all $k, l \in \indexset$ with $k\neq l$
  (mailboxes cannot be shared; any message has a unique recipient).

\noindent
  Then the \defitem{asynchronous parallel composition} of
  these components is defined by\vspace{-.5ex}
  \[
  \Big\|_{i \in \indexset} ((S_k, T_k, F_k, {M_0}_k, \ell_k), I_k, O_k) =
    ((S, T, F, {M_0}, \ell), I, O)\vspace{-.5ex}
  \]
  with
  $S \mathord= \bigcup_{k \in \indexset} S_k,~
  T \mathord= \bigcup_{k \in \indexset} \!T_k,~
  F \mathord= \bigcup_{k \in \indexset} F_k,~
  M_0 \mathord= \sum_{k \in \indexset} {M_0}_k,~
  \ell \mathord= \bigcup_{k \in \indexset} \ell_k$
  (componentwise union of all nets),
  $I \mathord= \bigcup_{k \in \indexset} I_k$
  (we accept additional inputs from outside), and
  $O \mathord= \bigcup_{k \in \indexset} O_k \setminus \bigcup_{k \in \indexset} I_k$
  (once fused with an input, $o\inp O_I$ is no longer an output).
\end{defi}

\noindent
Note that the asynchronous parallel composition of components with interfaces is again
a component with interface.

\begin{obs}\label{obs-associativity}
  $\|$ is associative.
\end{obs}
\noindent
  This follows directly from the associativity of the (multi)set union
  operator.\hfill$\Box$\\[1.5ex]
We are now ready to define the class of nets representing systems 
of asynchronously communicating sequential components.

\begin{defi}\label{df-LSGA}
  $\!$A Petri net $N$ is an \defitem{LSGA net} (a \defitem{locally sequential
  globally asynchronous net}) iff there exists an index set
  $\indexset$ and sequential components with interface
  $\mathcal{C}_k,~ k \inp \indexset$, such that $(N, I, O) = \|_{k \in \indexset} \mathcal{C}_k$
  for some $I$ and $O$.
\end{defi}

\noindent
Up to $\approx^\Delta_{bSTb}$---or any reasonable equivalence
preserving causality and branching time but abstracting from internal
activity---the same class of LSGA systems would have been obtained if
we had imposed, in \refdf{component} of sequential components, that $I$, $O$ and $Q$ form
a partition of $S$ and that $\precond{I}=\emptyset$.\footnote{%
First of all, any $i\in I$ with $\precond{i} \neq \emptyset$ can be split into a
pure input place, receiving tokens only from outside the component, and an
internal place, which is the target of all arcs that used to go to $i$.
Any transition $t$ with $i\in\precond{t}$ now needs to be split into one
that takes its input token from the pure input place and one that takes it from
the internal incarnation of $i$. In fact, if $F(i,t)=n$ then $t$ needs to be
split into $n\mathord+1$ copies. The result of this transformation is that $\precond{I}=\emptyset$.

Next, any component $\mathcal{C}=((S,T,F,M_0,\ell),I,O)$ with $\precond{I}=\emptyset$ can be
replaced by an equivalent component $((S',T',F',M'_0,\ell'),I,O)$ whose places $S'$ are
\plat{$I\dcup O \dcup Q$}, where $Q$ is the set of markings of $\mathcal{C}$, each restricted to
the places outside $I$ and $O$.
For each transition $t$ and markings $M,M'$ of the component such that
$M \mathrel{[t\rangle} M'$, writing $q:=M\restrictedto (S \setminus (I\cup O))$
and $q':=M'\restrictedto (S \setminus (I\cup O))$, there will be a transition $t_q \in T'$ with
$F'(i,t_q) = F(i,t)$ for all $i\inp I$,
$F'(t_q,o) = F(t,o)$ for all $o\inp O$,
$F'(q,t) = F'(t,q') =1$, and $F'(p,t)=F'(t,p)=0$ otherwise.
Moreover, $\ell'(t_q)=\ell(t)$ and $M'_0$ consists of the single place
$M_0\restrictedto (S \setminus (I\cup O))$. 
This component clearly has the required properties.
}
However, it is essential that our definition allows multiple
transitions of a component to read from the same input place.

\subsection{Distributed nets}

In the remainder of this section we give a more abstract
characterisation of Petri nets representing distributed systems,
namely as \emph{distributed} Petri nets, which we introduced in
\cite{glabbeek08syncasyncinteractionmfcs}. This will be useful in
\refsec{distributable}, where we investigate
distributability using this more semantic characterisation. We show
below that the concrete characterisation of distributed systems as
LSGA nets and this abstract characterisation agree.

Following \cite{BCD02}, to arrive at a class of nets representing distributed systems,
we associate \emph{localities} to the elements of a net $N=(S,T,F,M_0,\ell)$.
We model
this by a function \mbox{$D: S\cup T \rightarrow\Loc$}, with $\Loc$ a
set of possible locations.  We refer to such a
function as a \defitem{distribution} of $N$.  Since the identity of
the locations is irrelevant for our purposes, we can just as well
abstract from $\Loc$ and represent $D$ by the equivalence relation
$\equiv_D$ on $S\cup T$ given by $x \equiv_D y$ iff $D(x)=D(y)$.

Following \cite{glabbeek08syncasyncinteractionmfcs}, we impose a fundamental
restriction on distributions, namely that when two transitions
can occur in one step, they cannot be co-located. This reflects our
assumption that at a given location \visible actions can only occur
sequentially.

In \cite{glabbeek08syncasyncinteractionmfcs} we observed that
Petri nets incorporate a notion of synchronous interaction, in that a
transition can fire only by synchronously taking the tokens from all
of its preplaces. In general the behaviour of a net would change
radically if a transition would take its input tokens one by one---in
particular deadlocks may be introduced. Therefore we insist that in a
distributed Petri net, a transition and all its input places reside on
the same location. There is no reason to require the same for the
output places of a transition, for the behaviour of a net would not
change significantly if transitions were to deposit their output tokens
one by one \cite{glabbeek08syncasyncinteractionmfcs}.

This leads to the following definition of a distributed Petri net.

\begin{defi}\label{df-distributed}$\!$\cite{glabbeek08syncasyncinteractionmfcs}\,
    A Petri net $N = (S, T, F, M_0, \ell)$ is \defitem{distributed}
    iff there exists a distribution $D$ such that
\begin{enumerate}[(1)]
    \item
      $\forall s \in S, ~t \in T.~\hspace{1pt}s \in \precond{t}
      \implies t \equiv_D s$,
    \item
     $\forall t,u\in T.~t \concurrent u \implies t\not\equiv_D u$.
  \end{enumerate}
\end{defi}

\noindent
A typical example of a net which is not distributed is shown in
\reffig{fullM} on Page \pageref{fig-fullM}.
Transitions $t$ and $v$ are concurrently executable
and hence should be placed on different locations. However,
both have preplaces in common with $u$ which would enforce putting all
three transitions on the same location. In fact, distributed nets can
be characterised in the following semi-structural way.

\begin{obs}\label{obs-distributed}
A Petri net is distributed iff there is no sequence $t_0,\ldots,t_n$ of
transitions with $t_0 \smile t_n$ and
$\precond{t_{i-1}}\cap\precond{t_{i}}\neq\emptyset$ for $i=1,\ldots,n$.\hfill$\Box$
\end{obs}

\noindent
Since a \hyperlink{scn}{structural conflict net} is defined as a net without such a sequence
with $n\mathbin=1$ (cf.\ \refdf{onesafe}), we obtain:

\begin{obs}\label{obs-distributed-structuralconflict}
  Every distributed Petri net is a \hyperlink{scn}{structural conflict net}.\hfill$\Box$
\end{obs}

\noindent
Further on, we use  a more liberal definition of a
distributed net, called \emph{essentially distributed}.
We will show that up to $\approx^\Delta_{bSTb}$ any essentially
distributed net can be converted into a distributed net.
In \cite{glabbeek08syncasyncinteractionmfcs} we employed an even more liberal definition of a
distributed net, which we call here \emph{externally distributed}.
Although we showed that up to step failures equivalence any externally
distributed net can be converted into a distributed net, this does not hold for
$\approx^\Delta_{bSTb}$.

\begin{defi}\label{df-externally distributed}
  A net $N=(S,T,F,M_0,\ell)$ is \emph{essentially distributed} iff
  there exists a distribution $D$ satisfying (1) of
  \refdf{distributed} and
\begin{enumerate}[$(1')$]
    \item[($2'$)] $\forall t,u\in T.~t \concurrent u \wedge \ell(t)\neq\tau \implies t\not\equiv_D u$.
  \end{enumerate}
  It is \emph{externally distributed} iff there exists a distribution $D$ satisfying (1) and
\begin{enumerate}[$(1'')$]
    \item[($2''$)] $\forall t,u\in T.~t \concurrent u \wedge \ell(t),\ell(u)\neq\tau \implies t\not\equiv_D u$.
  \end{enumerate}
\end{defi}
\noindent
Instead of ruling out co-location of concurrent transitions in general,
essentially distributed nets permit concurrency of internal transitions---labelled $\tau$---at the same location.
Externally distributed nets even allow concurrency between visible and silent transitions at the same location.
If the transitions $t$ and $v$ in the net of \reffig{fullM} would both be labelled $\tau$,
the net would be essentially distributed, although not distributed; in case only $v$ would
be labelled $\tau$ the net would be externally distributed but not essentially distributed.
Essentially distributed nets need not be structural conflict nets; in fact, \emph{any} net
without visible transitions is essentially distributed.

\begin{defi}\label{df-canonical}
Given any Petri net $N$, the \emph{canonical co-location relation} $\equiv_C$ on $N$ is the
equivalence relation on the places and transitions of $N$
\emph{generated} by Condition (1) of \refdf{distributed},
i.e.\ the smallest equivalence relation $\equiv_D$ satisfying (1).
The \emph{canonical distribution} of $N$ is the distribution $C$ that maps
each place or transition to its $\equiv_C$-equivalence class.
\end{defi}
\begin{obs}\label{obs-canonical}
A Petri net that is distributed (resp.\ essentially or externally
distributed) w.r.t.\ any distribution $D$, is
distributed (resp.\ essentially or externally distributed)
w.r.t.\ its canonical distribution.
\end{obs}
\noindent
This follows because whenever a co-location relation $\equiv_D$ satisfies Condition (2) of
\refdf{distributed} (resp.\ Condition ($2'$) or ($2''$) of \refdf{externally distributed}),
then so does any smaller co-location relation.
Hence a net is distributed (resp.\ essentially or externally distributed) iff its canonical
distribution $D$ satisfies (2) (resp.\ ($2'$) or ($2''$)).

\subsection{Correspondence between LSGA nets and distributed nets}

We proceed to show that the classes of LSGA nets, distributable
nets and essentially distributable nets essentially coincide.

That every LSGA net is distributed follows
because we can place each sequential component on a
separate location. The following two lemmas constitute a formal argument.
Here we call a component with interface $(N,I,O)$ distributed iff $N$
is distributed.

\begin{lem}\label{lem-sequential component distributed}
Any sequential component with interface is distributed.
\end{lem}
\begin{proof}
As a sequential component displays no concurrency,
it suffices to co-locate all places and transitions.
\end{proof}

\noindent
\reflem{parcompdistributed} states that the class of distributed nets is
closed under asynchronous parallel composition.

\begin{lem}\label{lem-parcompdistributed}
  Let $\mathcal{C}_k=(N_k, I_k, O_k)$, $k \inp \indexset$, be components with
  interface, satisfying the requirements of \refdf{parcomp},
  which are all distributed.
  Then $\|_{k \in \indexset} \mathcal{C}_k$ is distributed.
\end{lem}
\begin{proof}
  We need to find a distribution $D$ satisfying the requirements of
  \refdf{distributed}.

  Every component $\mathcal{C}_k$
  is distributed and hence comes with a distribution $D_k$.
  Without loss of generality the codomains of all $D_k$ can be assumed
  disjoint.

  Considering each $D_k$ as a function from net elements onto locations,
  a partial function $D_k'$ can be defined which does not map any places in
  $O_k$, denoting that the element may be located arbitrarily, and
  behaves as $D_k$ for all other elements.
  As an output place has no posttransitions within a component, any total
  function larger than (i.e.\ a superset of) $D_k'$ is still a valid
  distribution for $N_k$.

  Now $D' = \bigcup_{k \in \indexset} D_k'$ is a (partial) function, as every place
  shared between components is an input place of at most one.
  The required distribution $D$ can be chosen as any total function extending $D'$;
  it satisfies the requirements of \refdf{distributed} since
  the $D_k$'s do.
\end{proof}

\begin{cor}\label{cor-LSGA distributed}
Every LSGA net is distributed.\qed
\end{cor}

\begin{cor}\label{cor-LSGA-structuralconflict}
  Every LSGA net is a structural conflict net.\qed
\end{cor}

\noindent
Conversely, any distributed net $N$, and even any essentially
distributed net $N$, can be transformed in an LSGA net by
choosing co-located transitions with their pre- and postplaces as
sequential components and declaring any place that belongs to multiple
components to be an input place of component $N_k$ if it is a preplace
of a transition in $N_k$, and an output place of component $N_l$ if it
is a postplace of a transition in $N_l$ and not an input place of $N_l$.
As transitions sharing a preplace are co-located, a place will be an input
place of at most one component.
Furthermore, in order to guarantee that the components are sequential
in the sense of \refdf{component}, an explicit control place is
added to each component---without changing behaviour---as explained
below \refdf{component}. It is straightforward to check that the
asynchronous parallel composition of all so-obtained components is an
LSGA net, and that it is equivalent to $N$ (using $ \approx_\mathscr{F}$,
$\approx^\Delta_{bSTb}$, or any other reasonable equivalence).

\begin{thm}\label{thm-bothdistributedequal}
  For any essentially distributed net $N$ there is an LSGA net $N'$ with $N'\approx^\Delta_{bSTb} N$.
\end{thm}
\begin{proof}
  Let $N = (S, T, F, M_0, \ell)$ be an essentially distributed net with a distribution $D$.
  Then an equivalent LSGA net $N'$ can be constructed
  by composing sequential components with interfaces as follows.
  
  For each equivalence class $[x]$ of net elements according to $D$ a
  sequential component $(N_{[x]}, I_{[x]}, O_{[x]})$ is created.
  Each such component contains one new and initially marked place $p_{[x]}$
  which is connected via self-loops to all transitions in $[x]$.
  The interface of the component is formed by
  $I_{[x]} := (S \cap [x])$\footnote{Alternatively, we could take
  $I_{[x]} := \postcond{(T\backslash[x])}\cap [x]$.}
  and
  $O_{[x]} := \postcond{([x] \cap T)} \setminus [x]$.
  Formally,
  $N_{[x]} := (S_{[x]}, T_{[x]}, F_{[x]}, {M_0}_{[x]}, \ell_{[x]})$
  with
  \begin{iteMize}{$\bullet$}
    \item $S_{[x]} = ((S \cap [x]) \cup O_{[x]} \cup \{p_{[x]}\}$,
    \item $T_{[x]} = T \cap [x]$,
    \item $F_{[x]} = F \restrictedto (S_{[x]} \cup T_{[x]})^2 \cup
      \{(p_{[x]}, t), (t, p_{[x]}) \mid t \in T_{[x]}\}$,
    \item ${M_0}_{[x]} = (M_0 \restrictedto [x]) \cup \{p_{[x]}\}$, and
    \item $\ell_{[x]} = \ell \restrictedto [x]$.
  \end{iteMize}
  All components overlap at interfaces only, as the sole
  places not in an interface are the newly created $p_{[x]}$.
  The $I_{[x]}$ are disjoint as the equivalence classes $[x]$ are, so
  $(N',I',O') := \|_{[x] \in (S \cup T) / D} (N_{[x]}, O_{[x]}, I_{[x]})$ is well-defined.
  It remains to be shown that $N' \approx^\Delta_{bSTb} N$.
  The elements of $N'$ are exactly those of $N$ plus the new places $p_{[x]}$,
  which stay marked continuously except when a transition from $[x]$ is firing,
  and never connect two concurrently enabled transitions.

  As we cannot have concurrently firing visible transitions on a single location,
  $|\overline{U} \cap [x]| \leq 1$ for
  any reachable ST-marking $(M, U)$ of $N$ and
  any $x \in S \cup T$, \ie for any location $[x]$.
  Here $\overline{U}$ is the multiset representation of the sequence $U$,
  defined in \refsec{equivalences}.
  The relation
  $$
  \left\{ \left((M, U), (M \mathord\cup P_U, U)\right) \mid (M, U) \text{ is a reachable ST-marking of } N,
    P_U \mathbin= \{ p_{[x]} \mid \overline{U} \mathord\cap [x] \mathbin= \varnothing \}\! \right\}
  $$
  is a bijection between the reachable ST-markings of $N'$ and $N$ that preserves the
  ST-transition relations between them.
  In particular, if $(M,U)\goesto[\tau](M',U')$, using a silent transition
  that belongs to the equivalence class $[x]$, then $U'=U$ and
  $\overline{U} \mathord\cap [x] \mathbin= \varnothing$, \ie
  no transition at location $[x]$ is currently firing, using that $N$ is essentially distributed.
  Hence $p_{[x]}\in P_U$ and thus $(M\mathord\cup P_U,U)\goesto[\tau](M'\mathord\cup P_U,U)$.
  (This argument does not extend to externally distributed nets $N$.)
  From this it follows that \plat{$N' \approx^\Delta_{bSTb} N$}.
\end{proof}

\begin{exa}\label{essdist2LSGA}
In \reffig{fullMbusy} appears an example of an essentially distributed net;
the location borders are indicated. This net is not distributed, and thus not an LSGA net, because
the two topmost $\tau$-transitions are co-located but can be fired concurrently.
Applying the construction in the proof of \refthm{bothdistributedequal} turns this net into the
distributed net of \reffig{fullMbusy-LSGA}.%
\begin{figure}[tb]
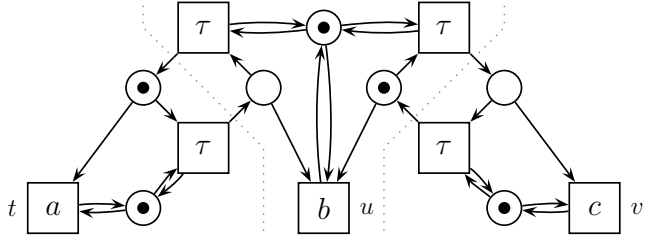

  \begin{center}
    \begin{petrinet}(10,4.5)
      \P (2,3):p;
      \p (4,3):pprime;
      \P (6,3):q;
      \p (8,3):qprime;
      \P (5,4):plsga;

      \P (2,1):pa;
      \P (8,1):pc;

      \ul (0.5,1):a:$a$:$t$;
      \u (5,1):b:$b$:$u$;
      \u (9.5,1):c:$c$:$v$;

      \t (3,2):ptau1:$\tau$;
      \t (3,4):ptau2:$\tau$;
      \t (7,2):qtau1:$\tau$;
      \t (7,4):qtau2:$\tau$;

      \a p->ptau1; \a ptau1->pprime; \a pprime->ptau2; \a ptau2->p;
      \a q->qtau2; \a qtau2->qprime; \a qprime->qtau1; \a qtau1->q;

      \AA plsga->ptau2; \AA ptau2->plsga;
      \AA plsga->qtau2; \AA qtau2->plsga;
      \AA plsga->b; \AA b->plsga;
      
      \AA pa->a; \AA a->pa;
      \AA pc->c; \AA c->pc;

      \AA pa->ptau1; \AA ptau1->pa;
      \AA pc->qtau1; \AA qtau1->pc;

      \a p->a;
      \a pprime->b;
      \a q->b;
      \a qprime->c;

      {%
        \psset{linecolor=gray}
        \psset{linestyle=dotted}
        \psline(2,4.5)(2,4)(4,2)(4,0.5)
        \psline(8,4.5)(8,4)(6,2)(6,0.5)
      }
    \end{petrinet}
  \end{center}
\vspace{-1em}
  \caption{\hspace{-1pt}The LSGA net obtained from converting the essentially distributed net of \reffig{fullMbusy}.}
  \label{fig-fullMbusy-LSGA}
\end{figure}
\end{exa}

\noindent
Likewise, up to $\approx_\mathscr{F}$ any externally distributed net can be converted into a distributed net.
\begin{prop}\label{pr-externallydistributedequal}\cite{glabbeek08syncasyncinteractionmfcs}\,\,
  For any externally distributed net $N$ there is a distributed net $N'$ with $N'\approx_\mathscr{F} N$.
\end{prop}
\begin{proof}
  The same construction applies. 
  The relation
  $$
  \left\{ (M, M \mathord\cup P) \mid M \text{ is a reachable marking of } N,\;
    P \mathbin= \{ p_{[x]} \mid [x] \textrm{ is a location} \}\! \right\}
  $$
  is a bijection between the reachable markings of $N'$ and $N$ that preserves the
  step transition relations between them.
  Here we use that the transitions in the associated LTS involve either a multiset of
  concurrently firing \emph{visible} transitions (that all reside on
  different locations and thus do not share a preplace $p_{[x]}$), or a single internal one.
  It follows that \plat{$N' \approx_\mathscr{F} N$}.
\end{proof}

\begin{figure}[ht]
\vspace{-3pt}
  \begin{center}
    \begin{petrinet}(10,3,4)
      \Q (1.5,3):p1:$p$;
      \Q (4.5,3):p2:$q$;
      \u (0,1):t1:$a$:$t$;
      \u (3,1):t2:$b$:$u$;
      \ub (6,1):t3:$\tau$:$v$;
      \qb (8,1):p3:$r$;
      \ub (10,1):t4:$c$:$w$;

      \a p1->t1;
      \a p1->t2;
      \a p2->t2;
      \a p2->t3;
      \a t3->p3;
      \a p3->t4;

      \psset{linestyle=dotted}
      \psline(7,3.3)(7,.4)
    \end{petrinet}
  \end{center}
  \vspace{-3ex}
  \caption{Externally distributed, but not
   convertible into a distributed net up to $\approx^\Delta_{bSTb}$.}
  \label{fig-externally distributed}
  \vspace{1ex}
  \begin{center}
    \begin{petrinet}(10,3,4)
      \Q (1.5,3):p1:$p$;
      \Q (4.5,3):p2:$q$;
      \ul (0,1):t1:$a$:$t$;
      \u (3,1):t2:$b$:$u$;
      \ub (6,1):t3:$\tau$:$v$;
      \qb (8,1):p3:$r$;
      \ub (10,1):t4:$c$:$w$;

      \P (3,3):plsga;
      \P (10,3):pc;

      \AA plsga->t1; \AA t1->plsga;
      \AA plsga->t2; \AA t2->plsga;
      \AA plsga->t3; \AA t3->plsga;

      \AA pc->t4; \AA t4->pc;

      \a p1->t1;
      \a p1->t2;
      \a p2->t2;
      \a p2->t3;
      \a t3->p3;
      \a p3->t4;

      \psset{linestyle=dotted}
      \psline(7,3.3)(7,.4)
    \end{petrinet}
  \end{center}
  \vspace{-3ex}
  \caption{The LSGA net obtained from converting the externally distributed net of \reffig{externally distributed}.}
  \label{fig-externally distributed-LSGA}
\end{figure}

\begin{exa}\label{extdist2LSGA}
\reffig{externally distributed} shows an externally distributed net;
the (canonical) location borders are dotted. It is not essentially distributed,
because the transitions $t$ and $v$ are co-located but can be
fired concurrently, while $\ell(t)\neq\tau$.
Applying the construction in the proof of \refpr{externallydistributedequal} turns this net into the
step failures equivalent LSGA net of \reffig{externally distributed-LSGA}.
\end{exa}

\noindent
The counterexample in \reffig{externally distributed} shows that up to $\approx^\Delta_{bSTb}$ not
all externally distributed nets can be converted into distributed nets.
Sequentialising the component with actions $a$, $b$
and $\tau$ (as happens in \reffig{externally distributed-LSGA}) would disable the execution
$\goesto[a^+]\Goesto\goesto[c^+]$.

\section{Distributable Systems}
\label{sec-distributable}

We now consider Petri nets as specifications of concurrent
systems and ask the question which of those specifications can be
implemented as distributed systems. This question can be formalised as
\begin{quote}\em
Which Petri nets are semantically equivalent to distributed nets?
\end{quote}
Of course the answer depends on the choice of a suitable
semantic equivalence. 
Here we will answer this question using the two equivalences discussed in the
introduction.
We will give a precise characterisation of those nets for which we can find semantically equivalent distributed nets. For the negative part of this characterisation, stating that certain nets are not distributable, we will use step failures
equivalence, which is one of the simplest and least discriminating
equivalences imaginable that abstracts from internal actions, but
preserves branching time, concurrency and divergence to some small
degree.\footnote{In \cite{GGS12} we used \emph{step readiness equivalence}, a
  slightly more discriminating equivalence with roughly the same properties. By moving
  to step failures equivalence we strengthen our result.}
Giving up on any of these latter three properties would make any Petri net
distributable, but in a rather trivial and unsatisfactory way:
\begin{iteMize}{$\bullet$}
\item
Every net can be converted into an essentially distributed net by refining every
transition
\psscalebox{0.7}{\begin{pspicture}(4,0)
  \def\thenetimage{trans}
  \u (2,.1):t:$a$:;
  \pnode(.5,.3){ntrans-d}
  \pnode(.5,-.1){ntrans-e}
  \pnode(3.5,.4){ntrans-a}
  \pnode(3.5,.1){ntrans-b}
  \pnode(3.5,-.2){ntrans-c}
  \a d->t;
  \a e->t;
  \a t->a;
  \a t->b;
  \a t->c;
\end{pspicture}}
into the net segment
\psscalebox{0.7}{\begin{pspicture}(8,0)
  \def\thenetimage{ref}
  \u (2,.1):t1:$\tau$:;
  \q (4,.1):p:;
  \u (6,.1):t2:$a$:;
  \pnode(.5,.3){nref-d}
  \pnode(.5,-.1){nref-e}
  \pnode(7.5,.4){nref-a}
  \pnode(7.5,.1){nref-b}
  \pnode(7.5,-.2){nref-c}
  \a t1->p;
  \a p->t2;
  \a d->t1;
  \a e->t1;
  \a t2->a;
  \a t2->b;
  \a t2->c;
\end{pspicture}}.\vspace{2pt}
This construction appears in \cite{BD11} where it is criticised for putting
``all relevant choice resolutions'' on one location. The construction does not introduce or
remove concurrency or divergence.
So it preserves even causality respecting linear time equivalences like pomset trace equivalence \cite{vanglabbeek01refinement}.
It does not preserve branching time equivalences,
because a choice between two visible transitions $a$ and $b$ in the
original net is implemented by a choice between two internal transitions preceding $a$ and $b$.
The resulting net is essentially distributed because all
new $\tau$-transitions can be placed on the same location, whereas all other transitions
get allocated a location of their own. Hence, using \refthm{bothdistributedequal},
it can be converted into an equivalent distributed net.
\item
When working in interleaving semantics, any net can be converted into
an equivalent distributed net by removing all concurrency
between transitions. This can be accomplished by adding a new,
initially marked place, with an arc to and from every transition in
the net.
\item
When fully abstracting from divergence, even when respecting causality
and branching time, the net of \reffig{fullM} is equivalent to the
essentially
distributed net of \reffig{fullMbusy}, and in fact it is not hard to
see that this type of implementation is possible for any given net.
Yet, the implementation may diverge, as the nondeterministic choices
might consistently be decided in an unhelpful way.
This argument is elaborated in \refsec{characterising} below.
The clause \plat{$M \arrownot\production{\tau}$} in \refdf{step failures} is
strong enough to rule out this type of implementation, even though our
step failures semantics abstracts from other forms of divergence.

\begin{figure}[tb]
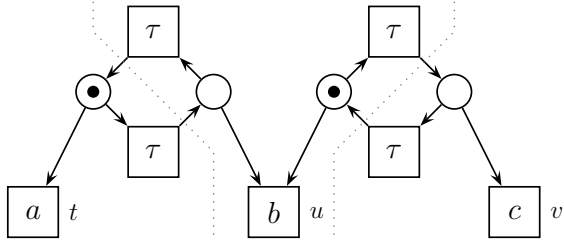

  \begin{center}
    \begin{petrinet}(10,4.5)
      \P (2,3):p;
      \p (4,3):pprime;
      \P (6,3):q;
      \p (8,3):qprime;

      \u (1,1):a:$a$:$t$;
      \u (5,1):b:$b$:$u$;
      \u (9,1):c:$c$:$v$;

      \t (3,2):ptau1:$\tau$;
      \t (3,4):ptau2:$\tau$;
      \t (7,2):qtau1:$\tau$;
      \t (7,4):qtau2:$\tau$;

      \a p->ptau1; \a ptau1->pprime; \a pprime->ptau2; \a ptau2->p;
      \a q->qtau2; \a qtau2->qprime; \a qprime->qtau1; \a qtau1->q;

      \a p->a;
      \a pprime->b;
      \a q->b;
      \a qprime->c;

      {%
        \psset{linecolor=gray}
        \psset{linestyle=dotted}
        \psline(2,4.5)(2,4)(4,2)(4,0.5)
        \psline(8,4.5)(8,4)(6,2)(6,0.5)
      }
    \end{petrinet}
  \end{center}
\vspace{-1em}
  \caption{A busy-wait implementation of the net in \reffig{fullM}, location borders dotted.}
  \label{fig-fullMbusy}
\end{figure}
\end{iteMize}
For the positive part, namely that all other nets are indeed distributable, 
we will use the most discriminating equivalence for which our implementation works, namely branching
ST-bisimilarity with explicit divergence, which is finer than
step failures equivalence. Hence we will obtain the strongest possible results for both directions and it turns out that the concept of distributability is fairly robust w.r.t.\ the choice of
a suitable equivalence: any equivalence notion between step failures
equivalence and branching ST-bisimilarity with explicit divergence will yield
the same characterisation.

\begin{defi}\label{df-distributable}
  A Petri net $N'$ is \defitem{distributable} up to an equivalence
  $\approx$ iff there exists a distributed net $N$ with $N \approx N'$.
\end{defi}

\noindent
Formally we give our characterisation of distributability by classifying which
finitary plain structural conflict nets can be implemented as distributed nets,
and hence as LSGA nets. In such implementations, we use invisible
transitions. We study the concept ``distributable'' for plain nets only, but in
order to get the largest class possible we allow non-plain implementations,
where a given transition may be split into multiple transitions carrying the
same label.

\subsection{Characterising Distributability}\label{sec-characterising}

It is well known that sometimes a global protocol is necessary to implement
synchronous interactions present in system specifications. In
particular, this may be needed for deciding choices in a coherent
way, when these choices require agreement of multiple components.
The simple net in \reffig{fullM} shows a typical situation of this
kind. Independent decisions of the two choices might lead to incorrect
system behaviour. If $p$ and $q$ both decide to send their respective
tokens leftwards, $a$ can fire, yet the token from $q$ gets stuck as
$b$ never receives a second token. Compared to the correct semantics,
a firing of $c$ after $a$ is missing.
It can be argued that for this particular net there exists no
satisfactory distributed implementation that fully respects the
reactive behaviour of the original system: Transitions $t$ and $v$ are
supposed to be concurrently executable (if we do not want to restrict
performance of the system), and hence reside on different locations.
Thus at least one of them, say $t$, cannot be co-located with transition $u$.
However, both transitions are in conflict with $u$. 

As we use nets as models of reactive systems, we allow the environment
of a net to influence decisions at runtime by blocking some of the
possibilities. Equivalently we can say it is the environment that
fires transitions, and this can only happen for transitions that are
currently enabled in the net. If the net decides between $t$ and $u$
before the actual execution of the chosen transition, the environment
might change its mind in between, leading to a state of deadlock.
Therefore we work in a branching time semantics, in which the option
to perform $t$ stays open until either $t$ or $u$ occurs. Hence the
decision to fire $u$ can only be taken at the location of $u$, namely
by firing $u$, and similarly for $t$.  Assuming that it takes time to
propagate any message from one location to another, in no distributed
implementation of this net can $t$ and $u$ be simultaneously enabled,
because in that case we cannot exclude that both of them happen.
Thus, the only possible implementation of the choice between $t$ and
$u$ is to alternate the right to fire between
$t$ and $u$, by sending messages between them (cf.\ \reffig{fullMbusy}).
But if the
environment only sporadically tries to fire $t$ or $u$ it may
repeatedly miss the opportunity to do so, leading to an infinite
loop of control messages sent back and forth, without either
transition ever firing.

\begin{figure}[h]
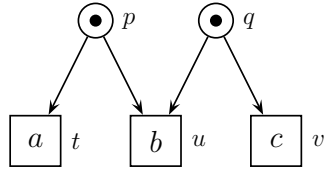

  \begin{center}
    \begin{petrinet}(6,3,4)
      \Q (2,3):p1:$p$;
      \Q (4,3):p2:$q$;
      \u (1,1):t1:$a$:$t$;
      \u (3,1):t2:$b$:$u$;
      \u (5,1):t3:$c$:$v$;

      \a p1->t1;
      \a p1->t2;
      \a p2->t2;
      \a p2->t3;
    \end{petrinet}
  \end{center}
  \vspace{-4ex}
  \caption{A fully reachable pure \structuralM.}
  \label{fig-fullM}
\end{figure}

Indeed such \structuralM-structures, representing
interference between concurrency and choice, turn out to play a crucial
r\^ole for characterising distributability.
To be specific, it is only those {\structuralM}s that are \emph{pure}, i.e.\
don't have extra arcs from their places to their transitions besides those in \reffig{fullM},
and are \emph{fully reachable}, i.e.\ for which there exists a reachable marking enabling
all three transitions at the same time.

\begin{defi}\label{df-fullM}
  Let $N = (S, T, F, M_0, \ell)$ be a Petri net.
  $N$ has a \hypertarget{M}{\defitem{fully reachable \visible pure \structuralM}} iff\\
  $\exists t,u,v \in T.
  \precond{t} \cap \precond{u} \ne \varnothing \wedge
  \precond{u} \cap \precond{v} \ne \varnothing \wedge
  \precond{t} \cap \precond{v} = \varnothing \wedge
  \exists M \in [M_0\rangle.
  \precond{t} \cup \precond{u} \cup \precond{v} \leq M$.
\end{defi}

\noindent
Note that \refdf{fullM} implies that $t\neq u$, $u\neq v$ and $t\neq v$.

\smallskip

\begin{obs}\label{obs-fullmimpliesnotdistr}
  A net with a fully reachable \visible pure {\structuralM} is not distributed. \qed
\end{obs}

\noindent
We now give an upper bound on the class of distributable nets by adapting a
result from \cite{glabbeek08syncasyncinteractionmfcs}:
We show that fully reachable \visible pure {\structuralM}'s that
are present in a plain structural conflict net are preserved under step failures equivalence.
In \cite{glabbeek08syncasyncinteractionmfcs} we showed this for step
readiness equivalence.

\begin{lem}\label{lem-plainfullmimpliesfailm}
  Let $N = (S, T, F, M_0, \ell)$ be a plain structural conflict net.
  If $N$ has a fully reachable \visible pure \structuralM, then there
  are $\sigma\mathbin\in\Act^*$ and $a,b,c\mathbin\in\Act$ with $a\mathbin{\neq} c$, such that
  $\rpair{\sigma,\{\{a,c\}\}}, \rpair{\sigma,\{\{b\}\}} \mathbin{\notin} \failureset(N)$ and
  $\rpair{\sigma,\{\{a,b\},\{b,c\}\}} \mathbin\in \failureset(N)$.
  (It is implied that $a \mathbin{\ne} b \mathbin{\ne} c$.)
\end{lem}
\begin{proof}
  $N$ has a {fully reachable \visible pure \structuralM}, so
  there exist $t,u,v \inp T$ and $M \inp [M_0\rangle$ such that
  $
  \precond{t} \cap \precond{u} \ne \varnothing \wedge
  \precond{u} \cap \precond{v} \ne \varnothing \wedge
  \precond{t} \cap \precond{v} = \varnothing \wedge
  \precond{t} \cup \precond{u} \cup \precond{v} \leq M$.
  Let $\sigma \in \Act^*$ such that \plat{$M_0 \Production{\sigma} M$}.
  Let $a:=\ell(t)$, $b:=\ell(u)$ and $c:=\ell(v)$,
  Then \plat{$M\goesto[\{a,c\}]$} and \plat{$M\goesto[\{b\}]$}.\vspace{1pt}
  Moreover, using that $N$ is a structural conflict net,
  \plat{$M\arrownot\goesto[\{a,b\}]$} and \plat{$M\arrownot\goesto[\{b,c\}]$}.
  Since $N$ is a plain net, \plat{$M \arrownot\production{\tau}$}, and there
  is no $M'\neq M$ with $M_0 \Production{\sigma} M'$.
  Hence $\rpair{\sigma,\{\{a,c\}\}}, \rpair{\sigma,\{\{b\}\}} \notin \failureset(N)$ and
  $\rpair{\sigma,\{\{a,b\},\{b,c\}\}} \in \failureset(N)$.
\end{proof}

\begin{lem}\label{lem-failmimpliesfullm}
  Let $N = (S, T, F, M_0, \ell)$ be a structural conflict net. If there
  are $\sigma\in\Act^*$ and $a,b,c\in\Act$ with $a\neq c$, such that
  $\rpair{\sigma,\{\{a,c\}\}}, \rpair{\sigma,\{\{b\}\}} \notin \failureset(N)$ and
  $\rpair{\sigma,\{\{a,b\},\{b,c\}\}} \in \failureset(N)$,
  then $N$ has a fully reachable \visible pure \structuralM.
\end{lem}
\begin{proof}
  Let $M \mathbin\in\powermultiset{S}$ be the marking that gives rise to the step failure pair\vspace{2pt}
  $\rpair{\sigma,\{\{a,b\},\{b,c\}\}}$, \ie $M_0 \Goesto[\sigma] M$,
  $M \arrownot\goesto[\{a,b\}]$ and $M \arrownot\goesto[\{b,c\}]$.
  Since $\rpair{\sigma,\{a,c\}}\notin \failureset(N)$, it must be that\vspace{1pt}
  $M \goesto[\{a,c\}]$. Likewise, $M \goesto[\{b\}]$.

  As $a \ne b \ne c \ne a$ there must exist three transitions $t, u, v \in T$
  with $\ell(t) \mathbin= a \wedge \ell(u) \mathbin= b\linebreak[2] \wedge \ell(v) \mathbin= c$ and
  $ M [\{t,v\}\rangle \wedge M [\{u\}\rangle \wedge
  \neg(M[\{t,u\}\rangle) \wedge \neg(M[\{u,v\}\rangle)$.
  From $M[\{t,v\}\rangle \wedge M[\{u\}\rangle$ it follows that
  $\precond{t} \cup \precond{u} \cup \precond{v} \leq M$ and
  $\precond{t} \cap \precond{v} = \varnothing$,
  using that $N$ is a structural conflict net.
  From $\neg(M[\{t,u\}\rangle)$ then follows $\precond{t} \cap \precond{u} \ne
  \varnothing$ and analogously for $u$ and $v$.
  Hence $N$ has a fully reachable \visible pure \structuralM.
\end{proof}

\noindent
Note that the lemmas above give a behavioural property that for plain
structural conflict nets is equivalent to having a fully reachable \visible pure \structuralM.

\begin{thm}\label{thm-trulysyngltfullm}
  Let $N$ be a plain structural conflict Petri net. 
  If $N$ has a fully reachable \visible pure {\structuralM}, then $N$ is
  not distributable up to step failures equivalence.
\end{thm}
\begin{proof}
  Let $N$ be a plain structural conflict net which has a fully reachable \visible pure \structuralM.
  Let $N'$ be a net which is step failures equivalent to $N$.
  By \reflem{plainfullmimpliesfailm} and \reflem{failmimpliesfullm}, also
  $N'$ has a fully reachable \visible pure \structuralM.
  By \refobs{fullmimpliesnotdistr}, $N'$ is not distributed. Thus $N$ is
  not distributable up to step failures equivalence.
\end{proof}

\noindent
Since $\approx^\Delta_{bSTb}$ is finer than $\approx_\mathscr{F}$,
this result holds also for distributability up to $\approx^\Delta_{bSTb}$
(and any equivalence between $\approx_\mathscr{F}$ and $\approx^\Delta_{bSTb}$).

In the following, we establish that this upper bound is tight, and hence a
finitary plain structural conflict net is distributable iff it has no fully
reachable \visible
pure \structuralM. For this, it is helpful to first introduce a more compact
graphical notation for Petri nets as well as macros for reversibility of transitions.

\begin{figure}[b]
\vspace{-.5ex}
\begin{minipage}[b]{0.49\linewidth}
  \begin{center}
    \begin{petrinet}(6,3)
      \qb (1,1):s:$s_i$;
      \qb (1,2.5):p:$p$;
      \ub (3,1.75):t:$a$:$t_j$;
      \qb (5,1.75):q:$q_j$;

      \rput[l](1,0){$\forall i \in \{0, 1\}$}
      \rput[l](3,0){$\forall j \in \{2, 3\}$}
       
      \a s->t;
      \a p->t;
      \a t->q;
    \end{petrinet}
  \end{center}
  \caption{A net with quantifiers.}
  \label{fig-quantifiernotation}
\end{minipage}
\begin{minipage}[b]{0.49\linewidth}
  \begin{center}
    \begin{petrinet}(6,3)
      \ql (1,2.75):p:$p$;
      \ql (1,1.75):s0:$s_0$;
      \ql (1,0.75):s1:$s_1$;
      \ub (3,2.5):t2:$a$:$t_2$;
      \ub (3,1):t3:$a$:$t_3$;
      \qb (5,2.5):q2:$q_2$;
      \qb (5,1):q3:$q_3$;

      \a p->t2; \a p->t3;
      \a s0->t2; \a s0->t3;
      \a s1->t2; \a s1->t3;
      \a t2->q2;
      \a t3->q3;
    \end{petrinet}
    \vspace{-1ex}
  \end{center}
  \caption{The same net expanded.}
  \label{fig-quantifiernotationexpanded}
\end{minipage}
\vspace{-1.5ex}
\end{figure}

\subsection{A compressed Petri net notation}

To compress the graphical notation, we allow universal
quantifiers of the form $\forall x. \phi(x)$ to appear
in the drawing 
(cf. Figures~\ref{fig-quantifiernotation} and \ref{fig-quantifiernotationexpanded}).
A quantifier replaces occurrences of $x$ in place and transition identities
with all concrete values for which $\phi(x)$ holds, possibly
creating a set of places, respectively transitions, instead of the depicted single one.
Accordingly, an arc of which only one end is replicated by a given quantifier
results in a fan of arcs, one for each replicated element.
If both ends of an arc are affected by the same quantifier,
an arc is created between pairs of elements corresponding to the same $x$,
but not between elements created due to differing values of $x$.

\subsection{Petri nets with reversible transitions}

A \emph{Petri net with reversible transitions} generalises the notion
of a Petri net; its semantics is given by a translation to an
ordinary Petri net, thereby interpreting the reversible transitions
as syntactic sugar for certain net fragments.  It is defined as a tuple
$(S,T,\UI,\ui,F,M_0,\ell)$ with $S$ a set of places, $T$ a set of (reversible)
transitions, labelled by \plat{$\ell:T\rightarrow\Act\dcup\{\tau\}$}, $\UI$ a set of
\emph{undo interfaces} with the relation $\ui\subseteq \UI\times T$
linking interfaces to transitions, $M_0 \inp \nat^S$ an initial marking, and
$$F\!: (S\times T\times \{{\scriptstyle \it in,~early,~late,~out,~far}\} \rightarrow \nat)$$
the flow relation.
When $F(s,t,{\scriptstyle \it type})>0$ for ${\scriptstyle \it type} \in
\{{\scriptstyle \it in,~early,~late,~out,~far}\}$, this is depicted by drawing an arc from $s$ to
$t$, labelled with its arc weight $F(s,t,{\scriptstyle \it type})$, of the form\
\psscalebox{0.7}{\begin{pspicture}(1.5,0.2)
  \def\thenetimage{arrowexamples}
  \pnode(0,0.1){narrowexamples-a}
  \pnode(1.5,0.1){narrowexamples-b}
  \a a->b;
\end{pspicture}},
\psscalebox{0.7}{\begin{pspicture}(1.5,0.2)
  \def\thenetimage{arrowexamples}
  \pnode(0,0.1){narrowexamples-a}
  \pnode(1.5,0.1){narrowexamples-b}
  \aEarly a->b;
\end{pspicture}},
\psscalebox{0.7}{\begin{pspicture}(1.5,0.2)
  \def\thenetimage{arrowexamples}
  \pnode(0,0.1){narrowexamples-a}
  \pnode(1.5,0.1){narrowexamples-b}
  \aLate a->b;
\end{pspicture}},
\psscalebox{0.7}{\begin{pspicture}(1.5,0.2)
  \def\thenetimage{arrowexamples}
  \pnode(0,0.1){narrowexamples-a}
  \pnode(1.5,0.1){narrowexamples-b}
  \a b->a;
\end{pspicture}},
\psscalebox{0.7}{\begin{pspicture}(1.5,0.2)
  \def\thenetimage{arrowexamples}
  \pnode(0,0.1){narrowexamples-a}
  \pnode(1.5,0.1){narrowexamples-b}
  \aFar b->a;
\end{pspicture}},
respectively. 
For $t\inp T$ and ${\scriptstyle \it type} \in
\{{\scriptstyle \it in,~early,~late,~out,~far}\}$, the multiset of places
$t^{\it type}\inp\nat^S$ is given by $t^{\it type}(s) = F(s,t,{\scriptstyle \it type})$.
When $s\inp t^{\it type}$ for ${\scriptstyle \it type} \in
\{{\scriptstyle \it in,~early,~late}\}$, the place $s$ is called a
\emph{preplace} of $t$ of type {\scriptsize \it type\/};
when $s\inp t^{\it type}$ for ${\scriptstyle \it type} \in
\{{\scriptstyle \it out,~far}\}$, $s$ is called a
\emph{postplace} of $t$ of type {\scriptsize \it type}.
For each undo interface $\omega\inp \UI$ and transition $t$ with $\ui(\omega,t)$ there must be places
$\undo[\omega](t)$, $\reset[\omega](t)$ and $\ack[\omega](t)$ in $S$.
A transition with a nonempty set of interfaces is called \emph{reversible};
the other (\emph{standard}) transitions may have pre- and postplaces
of types {\scriptsize \it in} and {\scriptsize \it out} only---for these
transitions $t^{\it in}\mathbin=\precond{t}$ and $t^{\it out}\mathbin=\postcond{t}$.
In case $\UI=\emptyset$, the net is just a normal Petri net.

A global state of a Petri net with reversible transitions is given by a marking
$M\inp\nat^S$, together with the state of each reversible transition ``currently
in progress''.  Each transition in the net can fire as usual. A reversible transition can
moreover take back (some of) its output tokens, and be \emph{undone} and
\emph{reset}.
(The use in our implementation will be that every reversible transition that fires is undone and reset later.)
When a transition $t$ fires, it consumes $\sum_{{\scriptstyle \it
type}\in\{{\scriptstyle \it in,~early,~late}\}} F(s,t,{\scriptstyle \it type})$
tokens from each of its preplaces $s$ and produces $\sum_{{\scriptstyle \it
type}\in\{{\scriptstyle \it out,~far}\}} F(s,t,{\scriptstyle \it type})$
tokens in each of its postplaces $s$.  A reversible transition $t$ that has fired
can start its reversal by consuming a token from $\undo[\omega](t)$ for one of
its interfaces $\omega$.  Subsequently, it can
take back the tokens from its postplaces of type {\scriptsize \it
far}. After it has retrieved all its output of type
{\scriptsize \it far}, the transition is undone, thereby returning
$F(s,t,{\scriptstyle \it early})$ tokens in each of its preplaces $s$ of type
{\scriptsize \it early}.
Afterwards, by
consuming a token from $\reset[\omega](t)$, for the same interface $\omega$ that
started the undo-process, the transition terminates its chain of activities by
returning $F(s,t,{\scriptstyle \it late})$ tokens in each of its {\scriptsize
\it late} preplaces $s$.  At that occasion it also produces a token in
$\ack[\omega](t)$. Alternatively, two tokens in $\undo[\omega](t)$ and
$\reset[\omega](t)$ can annihilate each other without involving the transition
$t$; this also produces a token in $\ack[\omega](t)$. The latter mechanism comes in action
when trying to undo a transition that has not yet fired.

\begin{figure}[ht]
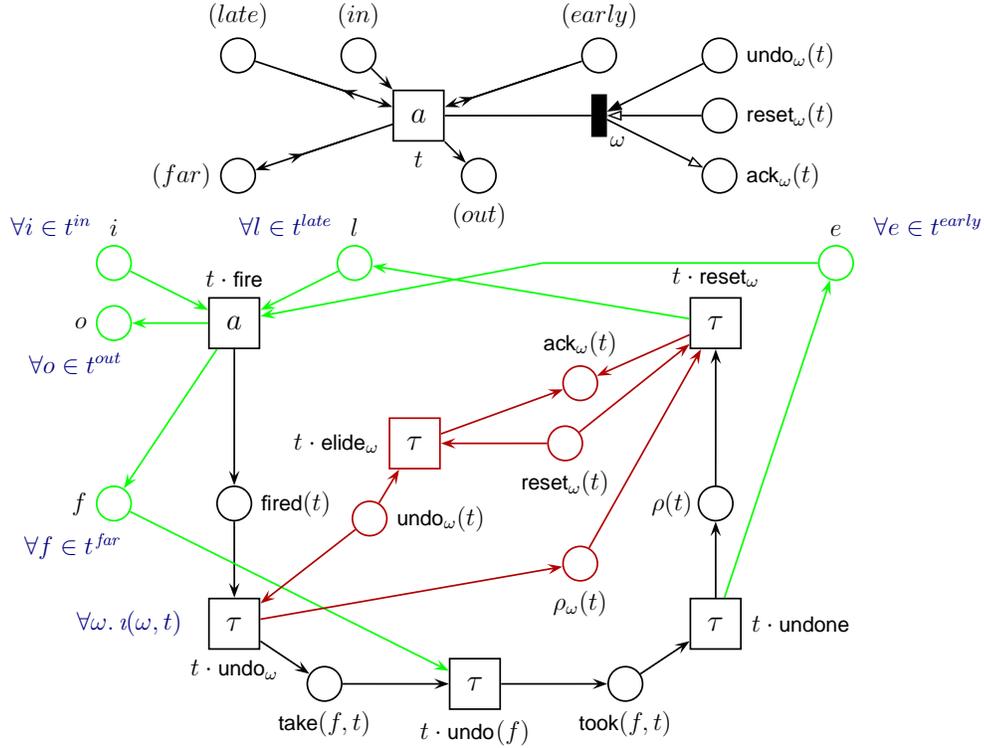

\vspace*{-.5ex}
  \begin{center}
    \begin{petrinet}(10,4)
      \qt(3,3):in:$(in)$;
      \qt(1,3):late:$(late)$;
      \qt(7,3):early:$(early)$;
      \q(9,3):undo:$\undo[\omega](t)$;
      \q(9,2):reset:$\reset[\omega](t)$;
      \q(9,1):ack:$\ack[\omega](t)$;
      \ql(1,1):far:$(far)$;
      \qb(5,1):out:$(out)$;
      \ub(4,2):t:$a$:$t$;

      \interface(7,2):ti:t:$\omega$;

      \a in->t;
      \aLate late->t;
      \aEarly early->t;
      \a t->out;
      \aFar t->far;
      \aUndo undo->ti;
      \aReset reset->ti;
      \aAck ti->ack;
    \end{petrinet}
  \end{center}
  \vspace{1ex}
  \begin{center}
    \begin{petrinet}(14,8)
      \psset{linecolor=green}
      \ql(1,4):far:$f$;
      \ql(1,7):out:$o$;
      \qt(1,8):in:$i$;
      \qt(5,8):late:$l$;
      \qt(13,8):early:$e$;

      \psset{linecolor=black}
      \qb(4.5,1):take:$\take(f,t)$;
      \ub(7,1):tundop:$\tau$:$t \cdot \undo[](f)$;
      \qb(9.5,1):took:$\took(f,t)$;

      \ub(3,2):tundoa:$\tau$:$t \cdot \undo[\omega]$;
      \ur(11,2):tundone:$\tau$:$t \cdot \undone$;

      \qr(3,4):fired:$\Fired(t)$;
      \ql(11,4):p2:$\rho(t)\!$;

      \ut(3,7):tfire:$a$:$t \cdot \fire$;
      \ut(11,7):treseta:$\tau$:$t \cdot \reset[\omega]$;

      \psset{linecolor=darkred}
      \qr(5.25,3.75):undoa:$\undo[\omega](t)$;
      \qb(8.75,3):pa:$\keep[\omega](t)$;
      \ul(6,5):elidea:$\tau$:$t \cdot \elide[\omega]$;
      \qt(8.75,6):acka:$\ack[\omega](t)$;
      \qb(8.5,5):reseta:$\reset[\omega](t)$;

      {
        \darkblue
        \rput[tr](1,3.5){\large $\forall f \in t^{\,\it far}$}
        \rput[tr](1,6.5){\large $\forall o \in t^{out}$}
        \rput[rb](0.5,8.45){\large $\forall i \in t^{in}$}
        \rput[rb](4.5,8.45){\large $\forall l \in t^{late}$}
        \rput[lb](13.5,8.45){\large $\forall e \in t^{early}$}
        \rput[r](2,2){\large $\forall \omega.\, \ui(\omega, t)$}
      }

      \psset{linecolor=black}
      \a tfire->fired;
      \a fired->tundoa;
      \a tundoa->take;
      \a take->tundop;
      \a tundop->took;
      \a took->tundone;
      \a tundone->p2;
      \a p2->treseta;

      \psset{linecolor=green}
      \a in->tfire;
      \a tfire->out;
      \a tfire->far;
      \a far->tundop;
      \a tundone->early;
      \a treseta->late;
      \av early[180]-(8,8)->[15]tfire;
      \a late->tfire;

      \psset{linecolor=darkred}
      \a tundoa->pa;
      \a pa->treseta;
      \a undoa->tundoa;
      \a undoa->elidea;
      \a elidea->acka;
      \a treseta->acka;
      \a reseta->elidea;
      \a reseta->treseta;
    \end{petrinet}
  \end{center}
\vspace{-1ex}
\caption{A reversible transition and its macro expansion.}
\label{fig-reversible}
\end{figure}

\reffig{reversible} shows the translation of a reversible transition
$t$ with $\ell(t)\mathbin=a$ into an ordinary net fragment.
The arc weights on the green (or grey) arcs are inherited from the
untranslated net; the other arcs have weight~1.
Formally, a net $(S,T,\UI,\ui,F,M_0, \ell)$ with reversible
transitions translates into the Petri net containing all places $S$,
all standard transitions in
$T$, labelled according to $\ell$, along with their pre- and
postplaces, and furthermore all net elements mentioned in
\reftab{reversible},
\hypertarget{Tback}{}$T^\leftarrow$ denoting the set of
reversible transitions in $T$.
The initial marking is exactly $M_0$.

\begin{table}[ht]
\[
\begin{array}{@{}l@{}c@{}clll@{}}
\textbf{Transition} & \textrm{at} & \textrm{label} & \textrm{Preplaces} & \textrm{Postplaces} & \textrm{for all} \\
\hline\rule[11pt]{0pt}{1pt}
t\cdot\fire          & t & \ell(t) & t^{in},~ t^{early},~ t^{late} & \Fired(t),~ t^{out}, t^{\,\it far} & t \in T^\leftarrow \\
t\cdot\undo[\omega]  & t\txf{-undo}& \tau & \undo[\omega](t),~\Fired(t) & \keep[\omega](t),~\take(f,t) &
  t \in T^\leftarrow,~\ui(\omega,t),~f \mathbin\in t^{\,\it far}\\
t\cdot\und           & f              & \tau & \take(f,t),~ f & \took(f,t) & t\in T^\leftarrow,~f \in t^{\,\it far}\\
t\cdot\undone        & t\txf{-undo}& \tau & \took(f,t) & \rho(t),~ t^{early} & t\in T^\leftarrow,~f \in t^{\,\it far}\\
t\cdot\reset[\omega] & t\txf{-undo}& \tau & \reset[\omega](t),~\keep[\omega](t),~\rho(t) & t^{late},~ \ack[\omega](t) &
 t \in T^\leftarrow,~\ui(\omega,t)\\
t\cdot\elide[\omega] & t\txf{-undo}& \tau & \undo[\omega](t),~\reset[\omega](t) & \ack[\omega](t) & t \in T^\leftarrow,~\ui(\omega,t)\\
\end{array}
\]
\vspace{-1.5ex}
\caption{Expansion of a Petri net with reversible transitions into a place/transition system.}
\label{tab-reversible}
\end{table}

A distribution of a Petri net with reversible transitions can be given
as a function $D:S \cup T \rightarrow \Loc$. As in Condition (1) of \refdf{distributed}
we require that a transition and its preplaces (of types {\scriptsize \it in},
{\scriptsize \it early} or {\scriptsize \it late}) reside on the same location.
Additionally, for any given transition $t$, all its undo-interface places $\undo[\omega](t)$ and
$\reset[\omega](t)$ for all $\omega\in\UI$ must reside on the same location---we refer to this
location as $t$-\txf{undo}. The second column of \reftab{reversible} indicates how such a
distribution is translated under expansion of reversible transitions into ordinary net fragments:
The location of a reversible transition $t$ is really the location of $t\cdot\fire$; it should be the
same as all preplaces of $t$. Furthermore, the transition $t\cdot\und$ and its preplace $\take(f,t)$
reside on the same location as the place $f\in t^{\,\it far}$. All other net elements that are part
of the macro expansion of $t$, except for $\ack[\omega](t)$, reside at the location $t$-\txf{undo}.
The resulting distribution of the expanded net is now guaranteed to satisfy (1).
Whether a Petri net with reversible translations is (essentially) distributed requires checking Condition (2) of \refdf{distributed} (or Condition ($2'$) of \refdf{externally distributed}) on its expansion.

\subsection{The conflict replicating implementation}\label{sec-implementation}

Now we establish that a finitary plain structural conflict net that has no fully reachable
\visible pure \structuralM\ is distributable.  We do this by proposing the \emph{conflict
replicating implementation} of any such net, and show that this implementation is always
(a) essentially distributed, and (b) equivalent to the original net. In order to get the strongest
possible result, for (b) we use branching ST-bisimilarity with explicit divergence.

\begin{figure}
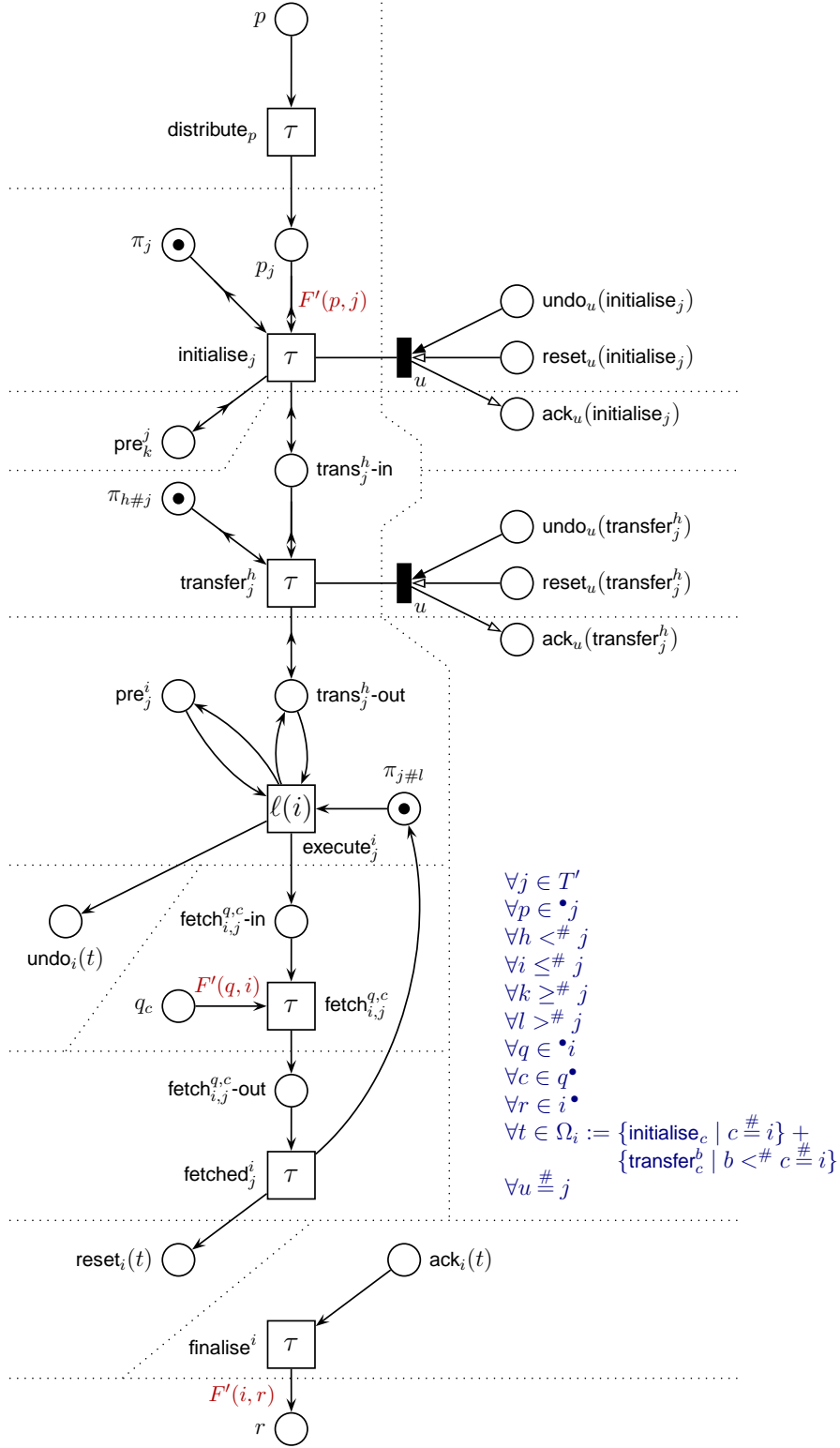

  \begin{center}
    \begin{petrinet}(17,26)
      {
        \darkblue

        \rput[lb](9,5){\large$\begin{array}{l}\displaystyle
          \forall j \in T'\\
          \forall p \in \precond{j}\\
          \forall h <^\# j\\
          \forall i \leq^\# j\\
          \forall k \geq^\# j\\
          \forall l >^\# j\\
          \forall q \in \precond{i}\\
          \forall c \in \postcond{q}\\
          \forall r \in \postcond{i\,}\\
          \forall t \in \UIij := \begin{array}[t]{@{}l@{}}
          \{\ini[c]\mid c\confeq i\} +\mbox{}\\ \{\trans[b]{c} \mid b <^\# c \confeq i\}
          \end{array}\\
          \forall u \confeq j\\
        \end{array}$}
      }

      {
        \darkred

        \rput[l](5.25,21){$F'(p,j)$}
        \rput[r](4.85,1.6){$F'(i,r)$}
        \rput[b](4,8.65){$F'(q,i)$}
      }

      \ql(5,26):p:$p$;
      \ul(5,24):distributep:$\tau$:$\dist$;
      \qx(5,22):pj:$p_j$:(-0.45,-0.45);
      \ql(3,18.5):prejk:$\Pre^j_k$;
      \Ql(3,22):readyinitialisej:$\pi_j$;
      \ul(5,20):initialisej:$\tau$:$\ini$;
      \interface(7,20):initialiseji:initialisej:$u$;
      \qr(9,21):undoinij:$\undo[u](\ini)$;
      \qr(9,20):resetinij:$\reset[u](\ini)$;
      \qr(9,19):ackinij:$\ack[u](\ini)$;
      \qr(5,18):transin:$\transin{j}$;
      \Ql(3,17.5):pconh:$\pi_{h\#j}$;
      \ul(5,16):trans:$\tau$:$\trans{j}$;
      \interface(7,16):transi:trans:$u$;
      \qr(9,17):undotransj:$\undo[u](\trans{j})$;
      \qr(9,16):resettransj:$\reset[u](\trans{j})$;
      \qr(9,15):acktransj:$\ack[u](\trans{j})$;
      \qr(5,14):transout:$\transout{j}$;
      \ql(3,14):prehj:$\Pre^i_j$;
      \Qt(7,12):pconj:$\pi_{j\#l}$;
      \ul(5,12):executehj:\makebox[0pt]{$\ell(i)$}:;
      \rput[l](5.2,11.3){\large $\exec{j}$}
      \qb(1,10):undohjt:$\undo(t)$;
      \ql(5,10):fetchphjin:$\fetchin[q,c]$;
      \ql(3,8.5):pbackbottom:$q_c$;
      \ur(5,8.5):fetchphj:$\tau$:$\fetch[q,c]$;
      \ql(5,7):fetchphjout:$\fetchout[q,c]$;
      \ul(5,5.5):fetchedhj:$\tau$:$\fetched{j}$;
      \qr(7,4):ackhjt:$\ack(t)$;
      \ql(3,4):resethjt:$\reset(t)$;
      \ul(5,2.5):completehj:$\tau$:$\comp{j}$;
      \ql(5,1):r:$r$;

      \a p->distributep;
      \aBack pbacktop->distributep;
      \a distributep->pj;
      \aEarly pj->initialisej;
      \aLate readyinitialisej->initialisej;
      \aFar initialisej->prejk;
      \aFar initialisej->transin;
      \aEarly transin->trans;
      \aLate pconh->trans;
      \aFar trans->transout;
      \B prehj->executehj;
      \B executehj->prehj;
      \A transout->executehj;
      \A executehj->transout;
      \a pconj->executehj;
      \a executehj->undohjt;
      \a executehj->fetchphjin;
      \a fetchphjin->fetchphj;
      \a pbackbottom->fetchphj;
      \a fetchphj->fetchphjout;
      \a fetchphjout->fetchedhj;
      \a fetchedhj->resethjt;
      \a ackhjt->completehj;
      \a completehj->r;
      \BB fetchedhj->pconj;
      \aUndo undoinij->initialiseji;
      \aReset resetinij->initialiseji;
      \aAck initialiseji->ackinij;
      \aUndo undotransj->transi;
      \aReset resettransj->transi;
      \aAck transi->acktransj;

      \psset{linestyle=dotted}
      \psline(6.6,26.3)(6.6,19)(7.3,18.5)(7.3,17.5)(6.6,17)(6.6,15.4)(7.8,14.5)(7.8,4.7)
      \psline(0,23)(6.6,23)
      \psline(0,19.4)(13,19.4)
      \psline(0,18)(3.8,18)(4.65,19.4)
      \psline(7.3,18)(13,18)
      \psline(0,15.4)(13,15.4)
      \psline(0,11)(7.8,11)
      \psline(1,7.7)(3.4,11)
      \psline(0,7.7)(7.8,7.7)
      \psline(0,4.7)(13,4.7)
      \psline(0,1.9)(13,1.9)
      \psline(2,1.9)(5.4,4.7)
    \end{petrinet}
  \end{center}
  \vspace{-3.5ex}
  \caption{The entire conflict replicating implementation, drawn with emphasis on the structure of the component of $j$;
    location borders dotted.}
  \label{fig-conflictrepl}
  \vspace{-1.3ex}
\end{figure}

\begin{figure}
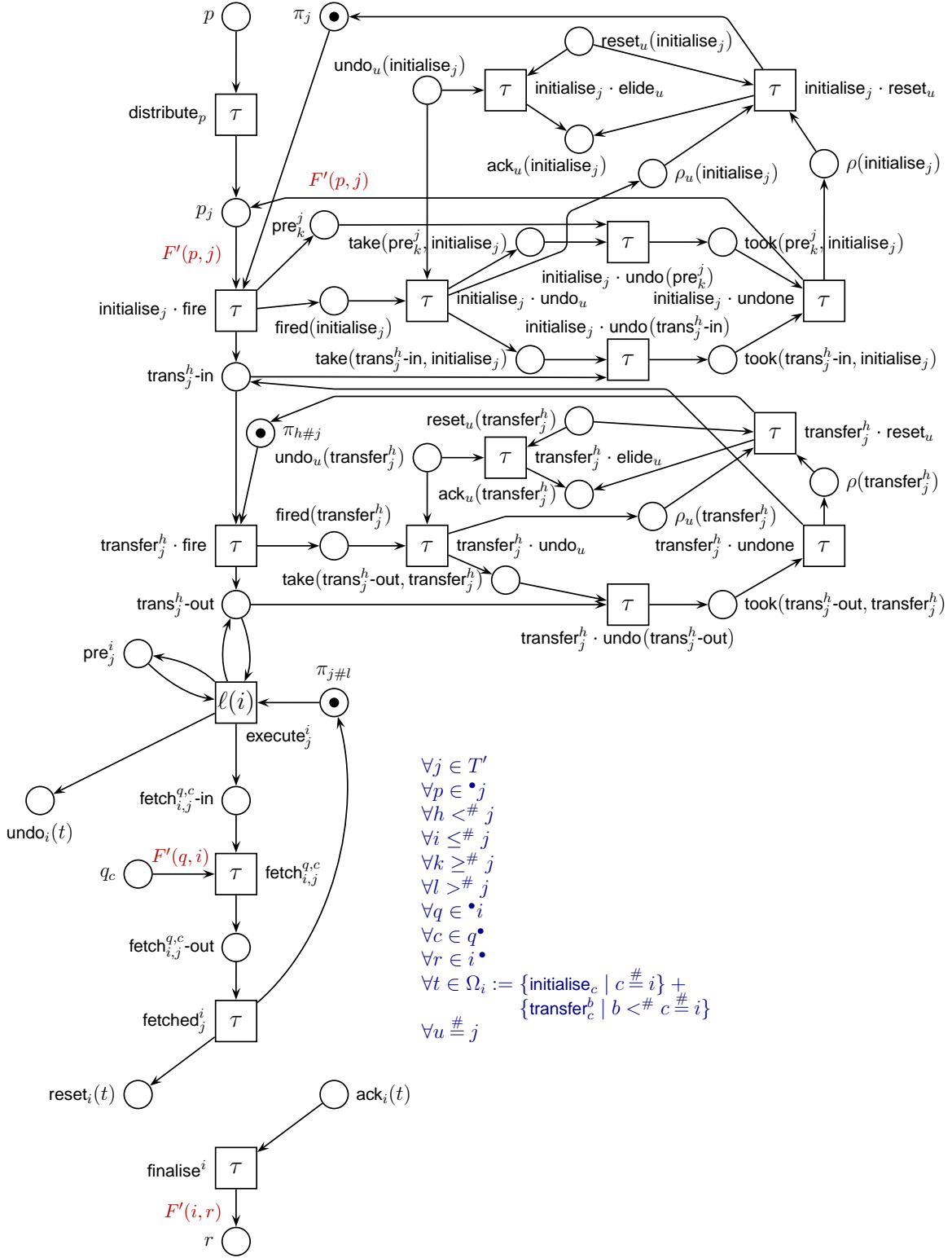

  \hspace{130pt}\makebox[0pt][c]{
    \begin{petrinet}(26,26)
      {
        \darkblue

        \rput[lb](9,5){\large$\begin{array}{l}\displaystyle
          \forall j \in T'\\
          \forall p \in \precond{j}\\
          \forall h <^\# j\\
          \forall i \leq^\# j\\
          \forall k \geq^\# j\\
          \forall l >^\# j\\
          \forall q \in \precond{i}\\
          \forall c \in \postcond{q}\\
          \forall r \in \postcond{i\,}\\
          \forall t \in \UIij := \begin{array}[t]{@{}l@{}}
          \{\ini[c]\mid c\confeq i\} +\mbox{}\\ \{\trans[b]{c} \mid b <^\# c \confeq i\}
          \end{array}\\
          \forall u \confeq j\\
        \end{array}$}
      }

      {
        \darkred

        \rput[r](7.85,22.65){$F'(p,j)$}
        \rput[r](4.85,21.15){$F'(p,j)$}
        \rput[r](4.85,1.6){$F'(i,r)$}
        \rput[b](4,8.65){$F'(q,i)$}
      }

      \ql(5,26):p:$p$;
      \ul(5,24):distributep:$\tau$:$\dist$;
      \ql(5,22):pj:$p_j$;

      \ul(5,20):initialisejfire:$\tau$:$\ini\cdot\fire$;
      \qbs(7,20.2):initialisejfired:$\Fired(\ini)$;
      \ur(8.9,20.2):initialisejundo:$\tau$:$\ini\cdot\undo[u]$;
      \ql(11,21.4):initialisejtakepre:$\take(\Pre^j_k, \ini)$;
      \ql(11,19):initialisejtaketrans:$\take(\transin{j}, \ini)$;
      \ubs(13,21.4):initialisejundopre:$\tau$:$\ini\cdot\undo[](\Pre^j_k)$;
      \uts(13,19):initialisejundotrans:$\tau$:$\ini\cdot\undo[](\transin{j})$;
      \qr(14.93,21.4):initialisejtookpre:$\took(\Pre^j_k, \ini)$;
      \qr(14.93,19):initialisejtooktrans:$\took(\transin{j}, \ini)$;
      \ul(17,20.2):initialisejundone:$\tau$:$\ini\cdot\undone$;
      \qr(17,23):initialisejrho:$\rho(\ini)$;
      \ur(16,24.5):initialisejreset:$\tau$:$\ini\cdot\reset[u]$;
      \ur(10.5,24.5):initialisejelide:$\tau$:$\ini\cdot\elide[u]$;
      \qr(13.5,22.8):initialisejrhou:$\rho_u(\ini)$;

      \ql(6.8,21.75):prejk:$\Pre^j_k\!$;
      \Ql(7,26):readyinitialisej:$\pi_j$;
      \qx(8.9,24.5):undoinij:$\undo[u](\ini)$:(-0.55,0.5);
      \qr(12,25.5):resetinij:$\reset[u](\ini)$;
      \qx(12,23.5):ackinij:$\ack[u](\ini)$:(-0.65,-0.55);
      \ql(5,18.65):transin:$\transin{j}$;
      \Qr(5.5,17.5):pconh:$\pi_{h\#j}$;

      \ul(5,15.2):transjfire:$\tau$:$\trans{j}\cdot\fire$;
      \qts(7,15.2):transjfired:$\Fired(\trans{j})$;
      \ur(8.9,15.2):transjundo:$\tau$:$\trans{j}\cdot\undo[u]$;
      \ql(10.5,14.5):transjtaketrans:$\take(\transout{j}, \trans{j})$;
      \ubs(13,14):transjundotrans:$\tau$:$\trans{j}\cdot\undo[](\transout{j})$;
      \qr(14.93,14):transjtooktrans:$\took(\transout{j}, \trans{j})$;
      \ul(17,15.2):transjundone:$\tau$:$\trans{j}\cdot\undone$;
      \qr(17,16.5):transjrho:$\rho(\trans{j})$;
      \ur(16,17.5):transjreset:$\tau$:$\trans{j}\cdot\reset[u]$;
      \ur(10.5,17):transjelide:$\tau$:$\trans{j}\cdot\elide[u]$;
      \qr(13.5,15.8):transjrhou:$\rho_u(\trans{j})$;

      \ql(8.9,17):undotransj:$\undo[u](\trans{j})$;
      \ql(12,17.75):resettransj:$\reset[u](\trans{j})$;
      \ql(12,16.25):acktransj:$\ack[u](\trans{j})$;
      \ql(5,14):transout:$\transout{j}$;
      \ql(3,13):prehj:$\Pre^i_j$;
      \Qt(7,12):pconj:$\pi_{j\#l}$;
      \ul(5,12):executehj:\makebox[0pt]{$\ell(i)$}:;
      \rput[l](5.2,11.3){\large $\exec{j}$}
      \qb(1,10):undohjt:$\undo(t)$;
      \ql(5,10):fetchphjin:$\fetchin[q,c]$;
      \ql(3,8.5):pbackbottom:$q_c$;
      \ur(5,8.5):fetchphj:$\tau$:$\fetch[q,c]$;
      \ql(5,7):fetchphjout:$\fetchout[q,c]$;
      \ul(5,5.5):fetchedhj:$\tau$:$\fetched{j}$;
      \qr(7,4):ackhjt:$\ack(t)$;
      \ql(3,4):resethjt:$\reset(t)$;
      \ul(5,2.5):completehj:$\tau$:$\comp{j}$;
      \ql(5,1):r:$r$;

      \a p->distributep;
      \a distributep->pj;

      \a pj->initialisejfire;
      \av initialisejundone[135]-(15,22.3)(6,22.3)->[20]pj;
      \a readyinitialisej->initialisejfire;
      \av initialisejreset[120]-(15,26)->[0]readyinitialisej;
      \a initialisejfire->prejk;
      \av prejk[0]-(7.5,21.75)->[140]initialisejundopre;
      \a initialisejfire->transin;
      \av transin[0]-(6.5,18.65)->[220]initialisejundotrans;
      \a undoinij->initialisejundo;
      \a undoinij->initialisejelide;
      \a resetinij->initialisejelide;
      \a resetinij->initialisejreset;
      \a initialisejelide->ackinij;
      \a initialisejreset->ackinij;
      \a initialisejfire->initialisejfired;
      \a initialisejfired->initialisejundo;
      \a initialisejundo->initialisejtakepre;
      \a initialisejundo->initialisejtaketrans;
      \a initialisejtakepre->initialisejundopre;
      \a initialisejtaketrans->initialisejundotrans;
      \a initialisejundopre->initialisejtookpre;
      \a initialisejundotrans->initialisejtooktrans;
      \a initialisejtookpre->initialisejundone;
      \a initialisejtooktrans->initialisejundone;
      \a initialisejundone->initialisejrho;
      \a initialisejrho->initialisejreset;
      \av initialisejundo[15]-(11.75,21.25)(11.75,22)->[210]initialisejrhou;
      \a initialisejrhou->initialisejreset;

      \a transin->transjfire;
      \av transjundone[135]-(13.8,18.4)(7,18.4)->[340]transin;
      \a pconh->transjfire;
      \av transjreset[135]-(15.25,18.25)(7,18.25)->[50]pconh;
      \a transjfire->transout;
      \a transout->transjundotrans;
      \a undotransj->transjundo;
      \a undotransj->transjelide;
      \a resettransj->transjelide;
      \a resettransj->transjreset;
      \a transjelide->acktransj;
      \a transjreset->acktransj;
      \a transjfire->transjfired;
      \a transjfired->transjundo;
      \a transjundo->transjtaketrans;
      \a transjtaketrans->transjundotrans;
      \a transjundotrans->transjtooktrans;
      \a transjtooktrans->transjundone;
      \a transjundone->transjrho;
      \a transjrho->transjreset;
      \av transjundo[30]-(10.5,15.8)->[180]transjrhou;
      \a transjrhou->transjreset;

      \B prehj->executehj;
      \B executehj->prehj;
      \A transout->executehj;
      \A executehj->transout;
      \a pconj->executehj;
      \a executehj->undohjt;
      \a executehj->fetchphjin;
      \a fetchphjin->fetchphj;
      \a pbackbottom->fetchphj;
      \a fetchphj->fetchphjout;
      \a fetchphjout->fetchedhj;
      \a fetchedhj->resethjt;
      \a ackhjt->completehj;
      \a completehj->r;
      \BB fetchedhj->pconj;
    \end{petrinet}
 }
  \vspace{-2.5ex}
  \caption{The entire conflict replicating implementation (with macros expanded).}
  \label{fig-conflictrepl-expanded}
  \vspace{-1.3ex}
\end{figure}

To define the conflict replicating implementation of a net $N'=(S',T',F',M'_0,\ell')$
we fix an arbitrary well-ordering $<$ on its transitions. We let
$b,c,g,h,i,j,k,l,u$ range over these ordered transitions, and write
\begin{enumerate}[$-$]
\item $i\mathbin\#j$ iff ~$i\neq j \wedge \precond{i} \cap \precond{j} \ne \varnothing$
  ~(transitions $i$ and $j$ are \emph{in conflict}),
  ~and $i \confeq j$ iff ~$i\mathbin\# j \vee i\mathbin=j$,
\item $i <^\#\! j$ iff ~$i<j \wedge i\mathbin{\#} j$,
  ~and $i \leqc j$ iff ~$i <^\#\! j \vee i=j$.
\end{enumerate}
\reffig{conflictrepl} shows the conflict replicating implementation of $N'$.  It is presented
as a Petri net $$\impl{N'}=(S,T,F,\UI,\ui,M_0,\ell)$$ with reversible transitions.  The set
$\UI$ of undo interfaces is $T'$, and for $i\inp \UI$ we
have $\ui(i,t)$ iff $t\inp \UIij$, where the sets of transitions
$\UI_i\subseteq T$ are specified in \reffig{conflictrepl}. The implementation $\impl{N'}$
inherits the places of $N'$ (\ie $S\supseteq S'$), and we define
$M_0\mathord\upharpoonright S'$ to be $M'_0$.  Given this, \reffig{conflictrepl} is not merely an
illustration of $\impl{N'}$---it provides a complete and accurate description of it, thereby
defining the conflict replicating implementation of any net. In interpreting this figure
it is important to realise that net elements are completely determined by their name
(identity), and exist only once, even if they show up multiple times in the figure. For
instance, the place $\pi_{h\#j}$ with $h\mathord=2$ and $j\mathord=5$ (when using
natural numbers for the transitions in $T'$) is the same as the place \plat{$\pi_{j\#l}$}
with $j\mathord=2$ and $l\mathord=5$; it is a standard preplace of \plat{$\exec{2}$} (for
all $i\leqc 2$), a
standard postplace of $\fetched{2}$, as well as a late preplace of \plat{$\trans[2]{5}$}.
\reffig{conflictrepl-expanded} depicts the same net after expanding the macros for reversible transitions.
An alternative description of the latter net appears in \reftab{conflictrepl} on Page~\pageref{tab-conflictrepl}.

The r\^ole of the transitions $\dist$ for $p\inp S'$ is to distribute a
token in $p$ to copies $p_j$ of $p$ in the localities of all
transitions $j\inp T'$ with $p\inp \precond{j}$.  In case $j$ is
enabled in $N'$, the transition $\ini$ will become enabled in $\impl{N'}$.
These transitions put tokens in the places \plat{$\Pre^j_k$}, which
are preconditions for all transitions \plat{$\exec[j]{k}$}, which
model the execution of $j$ at the location of $k$.  When two
conflicting transitions $h$ and $j$ are both enabled in $N'$, the first
steps $\ini[h]$ and $\ini$ towards their execution in $\impl{N'}$ can happen
in parallel. To prevent them from executing both, \plat{$\exec[j]{j}$}
(of $j$ at its own location) is only possible after \plat{$\trans{j}$},
which disables $\exec[h]{h}$.
This happens because \plat{$\trans{j}$} takes the initially present
token from the place $\pi_{h\#j}$, which is needed to fire $\exec[h]{h}$.

The main idea behind the conflict replicating implementation is that
a transition $h\inp T'$ is primarily executed by a sequential component of
its own, but when a conflicting transition $j$ gets enabled, the
sequential component implementing $j$ may ``steal'' the possibility to
execute $h$ from the home component of $h$, by putting a token in
\plat{$\transin[h]{j}$} and getting \plat{$\trans[h]{j}$} to fire, and then keep the options to do
$h$ and $j$ open on the home component of $j$ until one of them occurs. To prevent $h$ and $j$ from
stealing each other's initiative, which would result in deadlock, a
global asymmetry is built in by ordering the transitions.
Transition $j$ can steal the initiative from $h$ only when $h<j$.

In case $j$ is also in conflict with a transition $l$, with $j<l$,
the initiative to perform $j$ may subsequently be stolen by $l$.
In that case either $h$ and $l$ are in conflict too---then $l$ takes
responsibility for the execution of $h$ as well---or $h$ and $l$ are
concurrent---in that case $h$ will not be enabled, due to the absence
of fully reachable pure \structuralM s in $N'$.
The absence of fully reachable pure \structuralM s also guarantees that it
cannot happen that two concurrent transitions $j$ and $k$ both
steal the initiative from an enabled transition $h$.

After the firing of $\exec{j}$ all tokens that were left behind in the process of
carefully orchestrating this firing will have to be cleaned up, in order to prepare the
net for the next activity in the same neighbourhood. This is the reason for
the reversibility of the transitions preparing the firing of \plat{$\exec{j}$}.
Hence there is an undo interface for each transition $i\in T'$, cleaning up the mess made in
preparation of firing $\exec{j}$ for some $j\geq^\# i$. $\UIij$ is the set of all transitions $t$
that could possibly have
contributed to this. For each of them the undo interface $i$ is activated, by
\plat{$\exec{j}$} depositing a token in $\undo(t)$. After all preparatory transitions that
have fired are undone, tokens appear in the places $p_c$ for all $p\inp\precond{i}$ and $c\inp\postcond{p}$.
These are collected by $\fetch$, after which all transitions in $\UIij$ get a reset signal.
Those that have fired and were undone are reset, and those that never fired perform $\elide(t)$.
In either case a token appears in $\ack(t)$. These are collected by $\comp{j}$, which
finishes the process of executing $i$ by depositing tokens in its postplaces.

We allow multiple tokens to reside on the same place in the specification. To
ensure that this does never lead to the component implementing a transition $j$
starting the firing protocol again, even though it has not yet completed an
earlier round, we introduce a place $\pi_j$ which only holds a token while the
component is idle.

By means of location boundaries, \reffig{conflictrepl} also displays a distribution of $\impl{N'}$.
It has
\begin{iteMize}{$\bullet$}
\item a location $p$ for every place $p\inp S'$, containing $\dist$ and $p$;
\item locations $\ini$ and $\txf{execute}_j$ for every $j\in T'$---collectively
  referred to as ``the location of $j$''---the latter containing all transitions $\exec{j}$ for $i\leq^\# j \inp T'$;
\item locations $\fetched{j}$ for every $i\leq^\# j \inp T'$;
\item locations $\ini$-\txf{undo} for every $j\in T'$;
\item locations $\trans{j}$-\txf{undo} for every $h<^\#j\in T'$;
\item and locations $\comp{j}$ for every $i\in T'$.
\end{iteMize}
A transition $\trans{j}$ resides at location $\txf{execute}_h$, due to its common preplace
$\pi_{h\#j}$ with $\exec[g]{h}$. Likewise, $\fetch$ resides at location $\ini[c]$. Provided
$N'$ is a finitary plain structural conflict net without a fully reachable pure $\structuralM$,
the proof of \refthm{cri-distributed} will show that this distribution makes $\impl{N'}$ an essentially
distributed net.

\bigskip

\begin{figure}[hbt]
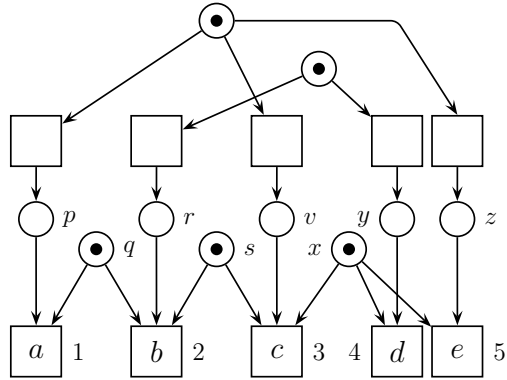

  \begin{center}
    \begin{petrinet}(10, 6)
      \P(4,6):inip1;
      \P(5.7,5.2):inip2;

      \t(1,4):init1:;
      \t(3,4):init2:;
      \t(5,4):init3:;
      \t(7,4):init4:;
      \t(8,4):init5:;

      \qr(1,2.7):p:$p$;
      \Qr(2,2.2):q:$q$;
      \qr(3,2.7):r:$r$;
      \Qr(4,2.2):s:$s$;
      \qr(5,2.7):v:$v$;
      \Ql(6.2,2.2):x:$x$;
      \ql(7,2.7):y:$y$;
      \qr(8,2.7):z:$z$;

      \ur(1,0.5):a:$a$:$1$;
      \ur(3,0.5):b:$b$:$2$;
      \ur(5,0.5):c:$c$:$3$;
      \ul(7,0.5):d:$d$:$4$;
      \ur(8,0.5):e:$e$:$5$;

      \a inip1->init1;
      \a inip1->init3;
      \av inip1[0]-(7,6)->[90]init5;
      \a inip2->init2;
      \a inip2->init4;

      \a init1->p;
      \a init2->r;
      \a init3->v;
      \a init4->y;
      \a init5->z;

      \a p->a;
      \a q->a;
      \a q->b;
      \a r->b;
      \a s->b;
      \a s->c;
      \a v->c;
      \a x->c;
      \a x->d;
      \a y->d;
      \a x->e;
      \a z->e;
    \end{petrinet}
  \end{center}
  \caption{An example net.}
  \label{fig-bigexampleoriginal}
\end{figure}

\begin{landfloat}{figure}{\rotateleft}
  \hfill
  \psscalebox{0.7}{
  \begin{petrinet}(38,22)
    {%
      \psset{linecolor=gray}
      \psset{linestyle=dotted}
      \psline(8,18)(8,20.5)(14,20.5)(14,18.25)(19,18.25)(19,20.5)(26,20.5)(26,18)

      \psline(-1,18)(38,18)
      \psline(-1,14)(38,14)
      \psline(-1,10)(38,10)

      \psline(4,18)(4,14)
      \psline(8,18)(8,14)
      \psline(12,18)(12,14)
      \psline(16,18)(16,14)
      \psline(20,18)(20,14)
      \psline(24,18)(24,14)
      \psline(28,18)(28,14)

      \psline(5,14)(5,10)
      \psline(13,14)(13,10)
      \psline(21,14)(21,10)
      \psline(28,14)(28,10)

      \psline(10,10)(10,6)(5,6)(3.5,5.5)(3.5,-1)
      \psline(18,10)(18,6)(13,6)(11.5,5.5)(11.5,-1)
      \psline(32,10)(32,6.875)(28.5,6)(24.75,6)(19.75,5)(19.75,-1)
      \psline(39,6)(29.75,6)(27.75,5.5)(27.75,-1)
    }

    {%
      \psset{linecolor=darkblue}
      \P(15,21):inip1;
      \P(20,20):inip2;

      \t(1,19):init1:;
      \t(9,19):init2:;
      \t(17,19):init3:;
      \t(25,19):init4:;
      \t(31,19):init5:;
    }%

    \ql(1,17):p:$p$;
    \Ql(5,17):q:$q$;
    \ql(9,17):r:$r$;
    \Ql(13,17):s:$s$;
    \ql(17,17):v:$v$;
    \Ql(21,17):x:$x$;
    \ql(25,17):y:$y$;
    \ql(31,17):z:$z$;

    \ur(1,15):distrp:$\tau$:$\dist[p]$;
    \ur(5,15):distrq:$\tau$:$\dist[q]$;
    \ur(9,15):distrr:$\tau$:$\dist[r]$;
    \ur(13,15):distrs:$\tau$:$\dist[s]$;
    \ur(17,15):distrv:$\tau$:$\dist[v]$;
    \ur(21,15):distrx:$\tau$:$\dist[x]$;
    \ur(25,15):distry:$\tau$:$\dist[y]$;
    \ur(31,15):distrz:$\tau$:$\dist[z]$;

    \ql(1,13):p1:$p_1$;
    \qr(3,13):q1:$q_1$;
    \ql(7,13):q2:$q_2$;
    \ql(9,13):r2:$r_2$;
    \qr(11,13):s2:$s_2$;
    \qr(15,13):s3:$s_3$;
    \qr(17,13):v3:$v_3$;
    \qr(19,13):x3:$x_3$;
    \qr(23,13):x4:$x_4$;
    \qr(25,13):y4:$y_4$;
    \qb(29,13):x5:$x_5$;
    \qr(31,13):z5:$z_5$;

    \ur(1,11):ini1:$\tau$:$\ini[1]$;
    \ur(9,11):ini2:$\tau$:$\ini[2]$;
    \ur(17,11):ini3:$\tau$:$\ini[3]$;
    \ur(25,11):ini4:$\tau$:$\ini[4]$;
    \ur(31,11):ini5:$\tau$:$\ini[5]$;

    {%
      \psset{linecolor=darkblue}
      \qr(7,9):trans12in:$\transin[1]{2}$;
      \qr(15,9):trans23in:$\transin[2]{3}$;
      \qr(23,9):trans34in:$\transin[3]{4}$;
      \qr(29,9):trans35in:$\transin[3]{5}$;
      \qr(33,9):trans45in:$\transin[4]{5}$;

      \ur(7,7):trans12:$\tau$:$\trans[1]{2}$;
      \ur(15,7):trans23:$\tau$:$\trans[2]{3}$;
      \ur(23,7):trans34:$\tau$:$\trans[3]{4}$;
      \ur(29,7):trans35:$\tau$:$\trans[3]{5}$;
      \ur(33,7):trans45:$\tau$:$\trans[4]{5}$;

      \ql(7,5):trans12out:$\transout[1]{2}$;
      \ql(15,5):trans23out:$\transout[2]{3}$;
      \ql(23,5):trans34out:$\transout[3]{4}$;
      \ql(32,5):trans35out:$\transout[3]{5}$;
      \qr(35,5):trans45out:$\transout[4]{5}$;
    }%

    \Ql(1,5):pi12:$\pi_{1\#2}$;
    \Ql(9,5):pi23:$\pi_{2\#3}$;
    \Ql(17,5):pi34:$\pi_{3\#4}$;
    \Ql(19,5):pi35:$\pi_{3\#5}$;
    \Ql(25,5):pi45:$\pi_{4\#5}$;

    \ur(1,3):exec11:$a$:$\exec[1]{1}$;
    \ul(7,3):exec12:$a$:$\exec[1]{2}$;
    \ur(9,3):exec22:$b$:$\exec[2]{2}$;
    \ul(15,3):exec23:$b$:$\exec[2]{3}$;
    \ur(17,3):exec33:$c$:$\exec[3]{3}$;
    \ul(23,3):exec34:$c$:$\exec[3]{4}$;
    \ur(25,3):exec44:$d$:$\exec[4]{4}$;
    \ul(31,3):exec35:$c$:$\exec[3]{5}$;
    \ur(34,3):exec45:$d$:$\exec[4]{5}$;
    \ur(37,3):exec55:$e$:$\exec[5]{5}$;

    \qr(1,1):pre11:$\Pre^1_1$;
    \ql(7,1):pre12:$\Pre^1_2$;
    \qr(9,1):pre22:$\Pre^2_2$;
    \ql(15,1):pre23:$\Pre^2_3$;
    \qr(17,1):pre33:$\Pre^3_3$;
    \ql(23,1):pre34:$\Pre^3_4$;
    \qr(25,1):pre44:$\Pre^4_4$;
    \ql(31,1):pre35:$\Pre^3_5$;
    \qr(34,1):pre45:$\Pre^4_5$;
    \qr(37,1):pre55:$\Pre^5_5$;

    {%
      \psset{linecolor=darkblue}
      \av inip1[180]-(8,21)->[0]init1;
      \a inip1->init3;
      \av inip1[0]-(26,21)->[180]init5;
      \av inip2[180]-(13,20)->[15]init2;
      \a inip2->init4;

      \a init1->p;
      \a init2->r;
      \a init3->v;
      \a init4->y;
      \a init5->z;
    }%

    \a p->distrp;
    \a q->distrq;
    \a r->distrr;
    \a s->distrs;
    \a v->distrv;
    \a x->distrx;
    \a y->distry;
    \a z->distrz;

    \a distrp->p1;
    \a distrq->q1;
    \a distrq->q2;
    \a distrr->r2;
    \a distrs->s2;
    \a distrs->s3;
    \a distrv->v3;
    \a distrx->x3;
    \a distrx->x4;
    \a distrx->x5;
    \a distry->y4;
    \a distrz->z5;

    \a p1->ini1;
    \a q1->ini1;
    \a q2->ini2;
    \a r2->ini2;
    \a s2->ini2;
    \a s3->ini3;
    \a v3->ini3;
    \a x3->ini3;
    \a x4->ini4;
    \a y4->ini4;
    \a x5->ini5;
    \a z5->ini5;

    \a ini2->trans12in;
    \a ini3->trans23in;
    \a ini4->trans34in;
    \a ini5->trans35in;
    \a ini5->trans45in;

    {%
      \psset{linecolor=darkblue}
      \a trans12in->trans12;
      \a trans23in->trans23;
      \a trans34in->trans34;
      \a trans35in->trans35;
      \a trans45in->trans45;

      \a trans12->trans12out;
      \a trans23->trans23out;
      \a trans34->trans34out;
      \a trans35->trans35out;
      \a trans45->trans45out;

      \a pi12->trans12;
      \a pi23->trans23;
      \a pi34->trans34;
      \a pi35->trans35;
      \a pi45->trans45;
    }%

    \a pi12->exec11;
    \a pi23->exec12;
    \a pi23->exec22;
    \a pi34->exec23;
    \a pi34->exec33;
    \a pi35->exec23;
    \a pi35->exec33;
    \a pi45->exec34;
    \a pi45->exec44;

    \a trans12out->exec12;
    \a trans12out->exec22;
    \a trans23out->exec23;
    \a trans23out->exec33;
    \a trans34out->exec34;
    \a trans34out->exec44;
    \a trans35out->exec35;
    \a trans35out->exec45;
    \a trans35out->exec55;
    \a trans45out->exec35;
    \a trans45out->exec45;
    \a trans45out->exec55;

    \a pre11->exec11;
    \a pre12->exec12;
    \a pre22->exec22;
    \a pre23->exec23;
    \a pre33->exec33;
    \a pre34->exec34;
    \a pre44->exec44;
    \a pre35->exec35;
    \a pre45->exec45;
    \a pre55->exec55;

    \av ini1[180]-(-1,9)(-1,0)(6,0)->[235]pre12;
    \av ini1[210]-(-0.8,9)(-0.8,0.2)(0,0.2)->[235]pre11;
    \av ini2[180]-(4,9.1)(4,-0.5)(14,-0.5)->[235]pre23;
    \av ini2[200]-(4.2,9)(4.2,-0.3)(8,-0.3)->[235]pre22;
    \av ini3[205]-(12.2,9)(12.2,0.2)(16,0.2)->[235]pre33;
    \av ini3[195]-(12.1,9.05)(12.1,0.1)(22,0.1)->[235]pre34;
    \av ini3[183]-(12,9.1)(12,0)(30,0)->[235]pre35;
    \av ini4[200]-(20.2,9)(20.2,-0.3)(24,-0.3)->[235]pre44;
    \av ini4[180]-(20,9.1)(20,-0.5)(33,-0.5)->[235]pre45;
    \av ini5[200]-(28,9)(28,-0.25)(36,-0.25)->[235]pre55;
  \end{petrinet}
  }%
  \vspace{4ex}
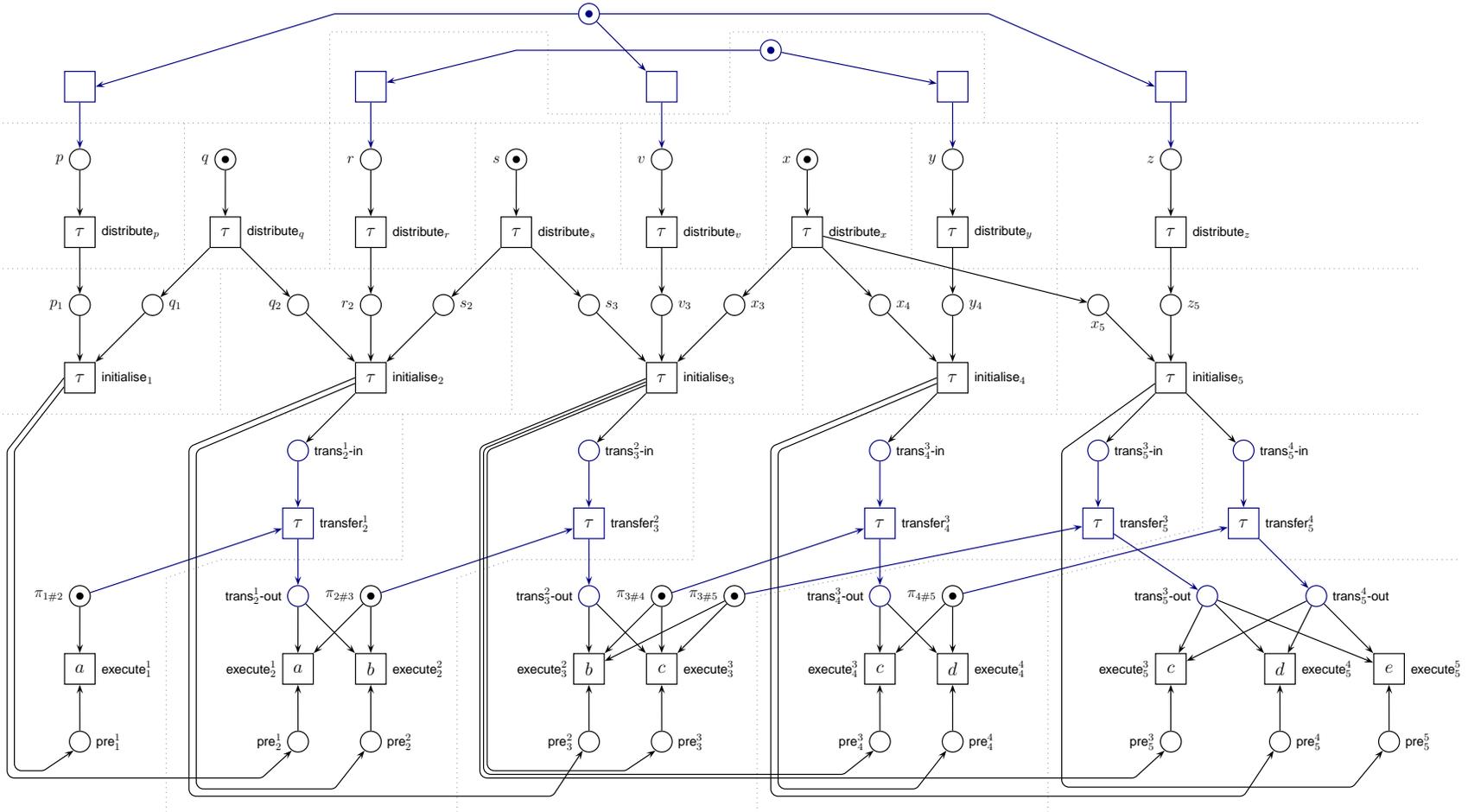
  \caption{The (relevant parts of the) conflict replicating implementation of the net in \reffig{bigexampleoriginal},
    location borders dotted.}
  \label{fig-bigexample}
  \vspace{2ex}
\end{landfloat}

The conflict replicating implementation is illustrated by means of the
finitary plain structural conflict net $N'$ of \reffig{bigexampleoriginal}.
The places and transitions $a$-$q$-$b$-$s$-$c$-$x$-$d$ in this net
constitute a \emph{Long \structuralM}: for each pair $a$-$b$,
$b$-$c$ and $c$-$d$ of neighbouring transitions, as well as for the
pair $a$-$d$ of extremal transitions, there exists a reachable marking
enabling them both. Moreover, neighbouring transitions in the long
\structuralM\ are in conflict: $a\mathbin \# b$, $b\mathbin \# c$ and $c\mathbin \# d$,
whereas the extremal transitions are concurrent: $a \smile d$.
However, $N'$ has no fully reachable pure \structuralM: no \structuralM-shaped triple of transitions
$a$-$b$-$c$, $b$-$c$-$d$ or $b$-$c$-$e$ is ever simultaneously enabled.

In \cite{glabbeek08syncasyncinteractionmfcs} we gave a simpler
implementation, the \emph{transition-controlled choice
  implementation}, that works for all finitary plain 1-safe Petri nets
without such a long \structuralM. Hence $N'$ constitutes an example
where that implementation does not apply, yet the conflict replicating
implementation does. In fact, when leaving out the $z$-$e$-branch it
may be the simplest example with these properties. We have added this
branch to illustrate the situation where three transitions are pairwise in conflict.

\reffig{bigexample} presents relevant parts of the conflict
replicating implementation $\impl{N'}$ of $N'$. What corresponds to the ten places of $N'$
can easily be discerned in $\impl{N'}$, but the transitions of $N'$ are replaced by more
complicated net fragments. In \reffig{bigexample} we have simplified
the rendering of $\impl{N'}$ by simply just copying the five topmost
transitions of $N'$, instead of displaying the net fragments replacing them.
This simplification is possible since the top half of $N'$ is already
distributed. To remind the reader of this, we left those transitions
unlabelled.\footnote{While it is highly desirable in practical applications to use
such simplifications to reduce the implementation size,
we refrained from doing so in the formal definition of our implementation.
It would have become less regular and the proofs correspondingly longer.}

In order to fix a well-ordering $<$ on the remaining transitions, we
named them after the first five positive natural numbers. The ordered
conflicts between those transitions now are $1 \mathord{\leq^\#} 2$,
$2 \mathord{\leq^\#} 3$, $3 \mathord{\leq^\#} 4$, $3 \mathord{\leq^\#} 5$
and $4 \mathord{\leq^\#} 5$. In \reffig{bigexample} we have skipped
all places, transitions and arcs involved in the cleanup of tokens
after firing of a transition. In this example the cleanup is not
necessary, as no place of $N'$ is visited twice. Thus, we displayed only the non-reversible part
of the transitions $\ini$ and \plat{$\trans{j}$}---i.e.\ $\ini\cdot\fire$ and
\plat{$\trans{j}\cdot\fire$}---as well as the transitions $\dist$ and \plat{$\exec{j}$}.
Likewise, we omitted the outgoing arcs of \plat{$\exec{j}$}, the
places $\pi_j$, and those places that have arcs only to omitted transitions.
We leave it to the reader to check this net against the definition in
\reffig{conflictrepl}, and to play the token game on this net, to see 
that it correctly implements $N'$.
\bigskip

  In \refsec{correctness} we will show,
  for any finitary plain structural conflict net $N'$
  without a fully reachable \visible pure \structuralM,
  that $\impl{N'} \approx^\Delta_{bSTb} N'$,
  and that $\impl{N'}$ is essentially distributed.
  Hence $\impl{N'}$ is an essentially distributed implementation of $N'$.
  By \refthm{bothdistributedequal} this implies that
  $N'$ is distributable up to $\approx^\Delta_{bSTb}$.
Together with \refthm{trulysyngltfullm} it follows that, for any equivalence between
$\approx_\mathscr{F}$ and $\approx^\Delta_{bSTb}$,
a finitary plain structural conflict net is distributable
iff it has no fully reachable \visible pure~\structuralM.

Given the complexity of our construction, no techniques known to us were adequate for
performing the equivalence proof. We therefore had to develop an entirely new method for
rigorously proving the equivalence of two Petri nets up to $\approx^\Delta_{bSTb}$, one of
which known to be plain. This method is presented in \refsec{method}.

\section{Proving Implementations Correct}\label{sec-method}

This section presents a method for establishing the equivalence of two Petri nets, one of
which known to be \hyperlink{plain}{plain}, up to branching ST-bisimilarity with explicit divergence.
It appears as \refthm{3ST}. First approximations of this method are presented in
Lemmas~\ref{lem-1ST} and~\ref{lem-2ST}. The progression from \reflem{1ST} to \reflem{2ST}
and to \refthm{3ST} makes the method more specific (so less general) and more powerful.
By means of a simplification a similar method can be obtained, also in three steps, for
establishing the equivalence of two Petri nets up to interleaving branching bisimilarity
with explicit divergence. This is elaborated at the end of this section.

We sometimes illustrate the results of this section in terms of the conflict replicating
implementation of a net defined in \refsec{implementation}. However, the actual application of these
results to show the correctness of that implementation is presented in \refsec{correctness}.

\begin{defi}\label{df-deterministic}
A labelled transition system $(\st,\tr,\inist)$ is called \emph{deterministic}
if for all reachable states $\mathfrak{M}\in [\inist\rangle$ we have
$\mathfrak{M}\arrownot\goesto[\tau]$ and if $\mathfrak{M}\goesto[a]\mathfrak{M}'$
and $\mathfrak{M}\goesto[a]\mathfrak{M}''$ for some $a\in\act$ then $\mathfrak{M}'=\mathfrak{M}''$.
\end{defi}
\noindent
Deterministic systems may not have reachable $\tau$-transitions at all;
this way, if $\mathfrak{M}\Goesto[\sigma]\mathfrak{M}'$ and
$\mathfrak{M}\Goesto[\sigma]\mathfrak{M}''$ for some $\sigma\in\act^*$ then
$\mathfrak{M}'=\mathfrak{M}''$.
Note that the labelled transition system associated to a
\hyperlink{plain}{plain} Petri net is deterministic; the same applies
to the ST-LTS, the split LTS or the step LTS associated to such a net.

\begin{lem}\label{lem-plain branching bisimilarity}
Let $(\st_1,\tr_1,\inist_1)$ and $(\st_2,\tr_2,\inist_2)$ be two
labelled transition systems, the latter being deterministic.
Suppose there is a relation $\Rel \subseteq \st_1\times\st_2$ such that
\begin{enumerate}[\em(a)]
\item $\inist_1\Rel \inist_2$,
\item if $\mathfrak{M}_1\Rel \mathfrak{M}_2$ and
  $\mathfrak{M}_1\goesto[\tau]\mathfrak{M}'_1$ then $\mathfrak{M}'_1\Rel \mathfrak{M}_2$,
\item if $\mathfrak{M}_1\Rel \mathfrak{M}_2$ and
  $\mathfrak{M}_1\goesto[a]\mathfrak{M}'_1$ for some $a\in\act$ then
  $\exists \mathfrak{M}'_2.~\mathfrak{M}_2\goesto[a]\mathfrak{M}'_2 \wedge \mathfrak{M}'_1\Rel \mathfrak{M}'_2$,
\item if $\mathfrak{M}_1\Rel \mathfrak{M}_2$ and
  $\mathfrak{M}_2\goesto[a]$ for some $a\in\act$ then either
  $\mathfrak{M}_1 \goesto[a]$ or $\mathfrak{M}_1 \goesto[\tau]$
\item and there is no infinite sequence $\mathfrak{M}_1\!\goesto[\tau]\! \mathfrak{M}'_1\!\goesto[\tau]\! \mathfrak{M}''_1\goesto[\tau] \cdots$
  with $\mathfrak{M}_1\Rel \mathfrak{M}_2$ for some $\mathfrak{M}_2$.
\end{enumerate}
Then $\Rel$ is a branching bisimulation with explicit divergence, and the two LTSs are
branching bisimilar with explicit divergence.
\end{lem}

\begin{proofNobox}
It suffices to show that $\Rel $ satisfies Conditions 1--3 of \refdf{branching LTS};
the condition on explicit divergence follows immediately from (e),
using that a deterministic LTS admits no divergence at all.
\begin{enumerate}[(1)]
\item By (a).
\item In case $\alpha=\tau$ this follows directly from (b), and otherwise from (c).
  In both cases $\mathfrak{M}^\dagger_2:=\mathfrak{M}_2$ and when $\alpha=\tau$ also $\mathfrak{M}'_2:=\mathfrak{M}_2$.
\item Suppose $\mathfrak{M}_1\Rel \mathfrak{M}_2$ and $\mathfrak{M}_2\goesto[\alpha]\mathfrak{M}'_2$.
  Since $(\st_2,\tr_2,\inist_2)$ is deterministic, $\alpha=a\in\Act$. By (d) we have either 
  $\mathfrak{M}_1 \goesto[a] \mathfrak{M}^1_1$ or $\mathfrak{M}_1 \goesto[\tau] \mathfrak{M}^1_1$ for some $\mathfrak{M}^1_1\in\st_1$.
  In the latter case (b) yields $\mathfrak{M}^1_1\Rel \mathfrak{M}_2$, and using (d) again,
  either $\mathfrak{M}^1_1 \goesto[a] \mathfrak{M}^2_1$ or $\mathfrak{M}^1_1 \goesto[\tau] \mathfrak{M}^2_1$ for some $\mathfrak{M}^2_1\in\st_1$.
  Repeating this argument, if the choice between $a$ and $\tau$ is made $k$ times in
  favour of $\tau$ (with $k\geq 0$), we obtain $\mathfrak{M}^{k}_1\Rel
  \mathfrak{M}_2$ (where $\mathfrak{M}^0_1:=\mathfrak{M}_1$) and
  either $\mathfrak{M}^{k}_1 \goesto[a] \mathfrak{M}^{k+1}_1$ or $\mathfrak{M}^k_1 \goesto[\tau] \mathfrak{M}^{k+1}_1$.
  By (e), at some point the choice must be made in favour of $a$, say at $\mathfrak{M}^k_1$.
  Thus $\mathfrak{M}_1\Goesto \mathfrak{M}^k_1 \goesto[a] \mathfrak{M}^{k+1}_1$, with $\mathfrak{M}^k_1\Rel \mathfrak{M}_2$.
  We take \plat{$\mathfrak{M}^\dagger_1$} and $\mathfrak{M}'_1$ from
  \refdf{branching LTS} to be $\mathfrak{M}^k_1$ and $\mathfrak{M}^{k+1}_1$.
  It remains to show that $\mathfrak{M}^{k+1}_1\Rel \mathfrak{M}'_2$.
  By (c) there is an $\mathfrak{M}''_2\in\st_2$ with $\mathfrak{M}_2\goesto[a]\mathfrak{M}''_2$ and $\mathfrak{M}^{k+1}_1\Rel \mathfrak{M}''_2$.
  Since $(\st_2,\tr_2,\inist_2)$ is deterministic, $\mathfrak{M}'_2=\mathfrak{M}''_2$.
  \qed
\end{enumerate}
\end{proofNobox}

\begin{lem}\label{lem-1ST}
Let $N=(S,T,F,M_0,\ell)$ and $N'=(S',T',F',M'_0,\ell')$ be two nets, $N'$ being plain.
Suppose there is a relation
$\Rel  \subseteq (\nat^{S}\times\nat^{T})\times (\nat^{S'}\times\nat^{T'})$
such that
\begin{enumerate}[\em(a)]
\item $(M_0,\emptyset)\Rel (M'_0,\emptyset)$,
\item if $(M_1,U_1)\Rel (M_1',U'_1)$ and
  $(M_1,U_1)\goesto[\tau](M_2,U_2)$ then $(M_2,U_2)\Rel (M_1',U'_1)$,
\item if $(M_1,U_1)\Rel (M_1',U'_1)$ and
  $(M_1,U_1)\goesto[\eta](M_2,U_2)$ for some $\eta\in\Act^\pm$\\\mbox{}\qquad then
  $\exists (M'_2,U'_2).~(M'_1,U'_1)\goesto[\eta](M'_2,U'_2) \wedge (M_2,U_2)\Rel (M_2',U'_2)$,
\item if $(M_1,U_1)\Rel (M_1',U'_1)$ and
  $(M'_1,U'_1)\goesto[\eta]$ with $\eta\in\Act^\pm$\\\mbox{}\qquad then either
  $\mathord{(M_1,U_1) \goesto[\eta]}$ or $\mathord{(M_1,U_1) \goesto[\tau]}$
\item and there is no infinite sequence $(M,U)\goesto[\tau] (M_1,U_1)\goesto[\tau] (M_2,U_2)
  \goesto[\tau] \cdots$\\\mbox{}\qquad with $(M,U)\Rel (M',U')$ for some $(M',U')$.
\end{enumerate}
Then $\Rel $ is a branching split bisimulation with explicit divergence, and $N \approx^\Delta_{bSTb} N'$.
\end{lem}

\begin{proof}
That $N$ and $N'$ are branching split bisimilar with explicit
divergence follows directly from  Lemma \ref{lem-plain branching bisimilarity}
by taking $(\st_1,\tr_1,\inist_1)$ and $(\st_2,\tr_2,\inist_2)$ to be
the split LTSs associated to $N$ and $N'$ respectively.
Here we use that the split LTS associated to a plain net is deterministic.
The final conclusion follows by \refpr{split}.
\end{proof}
\noindent
\reflem{1ST} provides a method for proving \plat{$N \approx^\Delta_{bSTb} N'$} that can
be more efficient than directly checking the definition. In particular, the
intermediate states $\mathfrak{M}^\dagger$ and the sequence of
$\tau$-transitions $\Goesto$ from \refdf{branching LTS} do not occur in
\reflem{plain branching bisimilarity}, and hence not in \reflem{1ST}. Moreover,
in Condition (d) one no longer has to match the targets of corresponding transitions.
\reflem{2ST} below, when applicable, provides an even more efficient method:
it is no longer necessary to specify the branching split bisimulation $\Rel$,
and the targets have disappeared from the transitions in Condition~\ref{2cST} as well.
Instead, we have acquired Condition~\ref{clause1ST}, but this is a structural
property, which is relatively easy to check.

\begin{lem}\label{lem-2ST}
Let $N=(S,T,F,M_0,\ell)$ be a net and $N'=(S',T',F',M'_0,\ell')$ be a plain
net with $S'\subseteq S$ and $M'_0=M_0\upharpoonright S'$.
Suppose:
\begin{enumerate}[\em(1)]
\item $\forall t\inp T,~\ell(t)\neq\tau.~ \exists t'\inp T',~\ell(t')=\ell(t).~
       \exists G\fin \nat^T,~\ell(G)\equiv\emptyset.~ \marking{t'}=\marking{t+G}$.
      \label{clause1ST}
\item For any $G\fin \Int^T$ with $\ell(G)\equiv\emptyset$, ~\plat{$M'\inp\nat^{S'}$},
      ~\plat{$U'\inp\nat^{T'}$} and ~$U\inp\nat^T$ with ~$\ell'(U')\mathbin=\ell(U)$,
      ~$M'+\precond{U'}\in [M'_0\rangle_{N'}$
      and ~$M:=M'+\precond{U'}+(M_0-M'_0)+\marking{G}-\precond{U}\in\nat^S$ with $M+\precond{U}\in[M_0\rangle_N$,
      it holds that:\label{clause2ST}
\begin{enumerate}[\em(a)]
\item there is no infinite sequence $M\goesto[\tau] M_1\goesto[\tau] M_2\goesto[\tau] \cdots$\label{2aST}
\item if $M' \goesto[a]$ with $a\in \Act$ then $M \goesto[a]$ or $M \goesto[\tau]$\label{2bST}
\item and if $M\goesto[a]$ with $a\inp \Act$ then $M'\goesto[a]$.\label{2cST}
\end{enumerate}
\end{enumerate}
Then $N \approx^\Delta_{bSTb} N'$.
\end{lem}

\begin{proofNobox}\hspace{-5pt}\footnote{%
For didactic reason it may be preferable to skip ahead and read the (simpler) proof of \reflem{2} first.}\,
Define $\Rel \subseteq (\nat^{S}\times\nat^{T})\times (\nat^{S'}\times\nat^{T'})$ by
$(M,U) \Rel  (M',U') :\Leftrightarrow
\ell'(U')=\ell(U) \wedge M'+\!\precond{U'}\inp [M'_0\rangle_{N'}
\wedge \exists G\fin\Int^T.~ \ell(G)\equiv\emptyset \wedge
M+\precond{U}=M'+\precond{U'}+(M_0\mathord-M'_0)+\marking{G}\in[M_0\rangle_N$.
It suffices to show that $\Rel $ satisfies Conditions (a)--(e) of \reflem{1ST}.
\begin{enumerate}[(a)]
\item Take $G=\emptyset$.
\item Suppose $(M_1,U_1)\Rel (M_1',U'_1)$ and $(M_1,U_1)\goesto[\tau](M_2,U_2)$.
  Then $\ell'(U'_1)\mathbin=\ell(U_1) \wedge M'_1+\!\precond{U'_1}\inp [M'_0\rangle_{N'}\linebreak[2] \wedge
  \exists G\fin\Int^T.~ \ell(G)\mathbin\equiv\emptyset \wedge
  M_1=M'_1+\!\precond{U'_1}+(M_0-M'_0)+\marking{G}-\!\precond{U_1} \wedge M_1+\precond{U}\in[M_0\rangle_N$
  and moreover $M_1\goesto[\tau]M_2 \wedge U_2=U_1$.
  So $M_1 [t\rangle M_2$ for some $t\mathbin\in T$ with $\ell(t)\mathbin=\tau$. Hence
  $M_2=M_1+\marking{t}=M'_1+\!\precond{U'_1}+(M_0\mathord-M'_0)+\marking{G+t}\linebreak[2]-\!\precond{U_1}$.
  Since $(M_1+\precond{U_1}) [t\rangle (M_2+\precond{U_1})$, we have $M_2+\precond{U_1}\in[M_0\rangle_N$.
  Since also $\ell(G+t)\equiv\emptyset$ it follows that $(M_2,U_1)\Rel (M_1',U'_1)$.
\item Suppose $(M_1,U_1)\Rel (M_1',U'_1)$ and $(M_1,U_1)\goesto[\eta](M_2,U_2)$,
  with $\eta\in\Act^\pm$.
  Then $\ell'(U'_1)\mathbin=\ell(U_1)$, ~$M'_1+\!\precond{U'_1}\inp [M'_0\rangle_{N'}$ and
  \begin{equation}\label{G}
  \exists G\fin\Int^T.~ \ell(G)\mathbin\equiv\emptyset \wedge
  M_1+\!\precond{U_1}=M'_1+\!\precond{U'_1}+(M_0-M'_0)+\marking{G}\in[M_0\rangle_N.
  \end{equation}
  First suppose $\eta=a^+$. Then $\exists t\inp T.~ \ell(t)\mathbin=a \wedge M_1[t\rangle
  \wedge M_2=M_1-\precond{t} \wedge U_2=U_1+\{t\}$.
  Using that $M_1\goesto[a]$ with $a\in \Act$, by Condition~\ref{2cST}
  we have $M'_1\goesto[a]$, \ie $M'_1[t'\rangle$ for some
  $t'\in T$ with $\ell'(t')=a$.  Let $M'_2:=M'_1-\precond{t}$ and $U'_2:=U'_1+\{t'\}$.
  Then $(M'_1,U'_1)\goesto[a^+](M'_2,U'_2)$.
  Moreover, $\ell(U_2)=\ell(U'_2)$,
  ~$M'_2+\precond{U'_2} = M'_1+\precond{U'_1}\inp [M'_0\rangle_{N'}$ and
  $M_2+\precond{U_2} = M_1+\precond{U_1}$. In combination with (\ref{G}) this yields
  $$~~~M_2+\!\precond{U_2}=M_1+\precond{U_1} =M'_1+\!\precond{U'_1}+(M_0\mathord-M'_0)+\marking{G}
   =M'_2+\!\precond{U'_2}+(M_0\mathord-M'_0)+\marking{G},\!\!\!\!\!$$
  so $(M_2,U_2)\Rel (M_2',U'_2)$.

  Now suppose $\eta=a^-$. Then $\exists t\inp U_1.\ \ell(t)\mathbin=a \wedge
  U_2\mathbin=U_1\mathord-\{t\} \wedge M_2=M_1+\postcond{t}$.  Since
  $\ell'(U'_1)\mathbin=\ell(U_1)$ there is a $t'\inp U'_1$ with
  $\ell(t')\mathbin=a$.  Let $M'_2:=M'_1+\postcond{t'}$ and
  $U'_2:=U'_1-\{t'\}$. Then $(M'_1,U'_1)\goesto[a^-](M'_2,U'_2)$.
  By construction, $\ell(U_2)=\ell(U'_2)$.
  Moreover, $M_2+\precond{U_2} = M_1+\postcond{t}+\precond{U_1}-\precond{t}=(
  M_1+\precond{U_1})+\marking{t}$, and likewise
  \begin{equation}\label{M'}
  M'_2+\precond{U'_2} = (M'_1+\precond{U'_1})+\marking{t'}
  \end{equation}
  so $(M'_1+\precond{U'_1})[t'\rangle (M'_2+\precond{U'_2})$.
  Since $M'_1+\precond{U'_1}\inp [M'_0\rangle_{N'}$, this yields 
  $M'_2+\precond{U'_2}\inp [M'_0\rangle_{N'}$.
  Moreover, $M_2+\!\precond{U_2} = M_1+\postcond{t}+\!\precond{U_1}-\!\precond{t} =
  M_1+\!\precond{U_1}+\marking{t} \in[M_0\rangle_N$.
  Furthermore, combining (\ref{G}) and (\ref{M'}) gives
  \begin{equation}\label{G2}
  \exists G\fin\Int^T.~ \ell(G)\mathbin\equiv\emptyset \wedge
  M_2+\!\precond{U_2}-\marking{t}=M'_2+\!\precond{U'_2}-\marking{t'}+(M_0-M'_0)+\marking{G}.
  \end{equation}
  By Condition~\ref{clause1ST} of \reflem{2ST}, $\exists t''\inp T',~\ell(t'')=\ell(t).~
       \exists G_t\fin \nat^T,~\ell(G_t)\equiv\emptyset.~ \marking{t}=\marking{t''-G_t}$.
  Since $N'$ is a \hyperlink{plain}{plain} net, it has only one
  transition $t^\dagger$ with $\ell(t^\dagger)\mathbin=a$, so $t''\mathbin=t'$.
  Substitution of $\marking{t'-G_t}$ for $\marking{t}$ in (\ref{G2}) yields
  \[\qquad
  \exists G\fin\Int^T.~ \ell(G)\mathbin\equiv\emptyset \wedge
  M_2+\!\precond{U_2}=M'_2+\!\precond{U'_2}+(M_0-M'_0)+\marking{G-G_t}.
  \]
  Since $\ell(G-G_t)\equiv\emptyset$ we obtain $(M_2,U_2)\Rel (M_2',U'_2)$.
\item Follows directly from Condition~\ref{2bST} and \refdf{split marking}.
\item Follows directly from Condition~\ref{2aST} and \refdf{split marking}.
  \qed
\end{enumerate}
\end{proofNobox}
\noindent
To illustrate the use of Lemmas~\ref{lem-2} and~\ref{lem-2ST}, let $N'$ be a plain net and $N$
be its conflict replicating implementation, depicted in \reffig{conflictrepl-expanded}.
Condition (1) says that for any visible transition $t$ in the implementation---this must be
\plat{$\exec{j}$} for some $i$ and $j$---there must be a transition $t'$ in $N'$ with the same label---this
must be $i$---such that the same token replacement $\marking{t'}$ that results from firing $t'$ in
the net $N'$ can also achieved by $t$ in $N$ together with a multiset $G$ of internal transitions of $N$.
For this to even make sense it is necessary that $S'\subseteq S$, so that $\marking{t'}$ can just as
well be seen as a token replacement of $N$. This condition can be fulfilled by taking $G$ to
contain $\dist$ for every preplace $p$ of $i$, \plat{$\fetch$} for every preplace $p$ of $i$ and
every $c\in\postcond{p}$, $\fetched{j}$, $u\cdot\elide$ for $u\in\UIij$, and $\comp{j}$.

In the proof of \reflem{2}/\ref{lem-2ST}, a branching bisimulation is constructed between the markings of $N'$ and $N$,
by relating any reachable marking $M'$ of $N'$ with the corresponding marking
$M'+(M_0\mathord-M'_0)$ of $N$; the latter is the marking $M'$ seen as a marking of $N$, together with
those places in $S\setminus S'$ that are marked initially (or by default).
In addition, $M'$ is also related to markings obtained from $M'+(M_0\mathord-M'_0)$ by adding or
subtracting the token replacement due to firing some internal transitions of $N$.
For instance, compared to the state of $N$ given by the marking $M'+(M_0\mathord-M'_0)$ it could be
that $\comp{j}$ has not yet fired---so that $\ack(t)$ is marked for all $t\in \UIij$ instead of the
postplaces $r$ of $i$---and that $\dist$ has already fired for some place $p$.
This gives rise to the marking $M'+(M_0\mathord-M'_0)+\marking{G}$ being related to $M'$, with
$G=-\{\comp{j}\}+\{\dist\}$.
To show that the relation really is a branching bisimulation with explicit divergence it
suffices to check the conditions (a)--(c). That these are enough to obtain the stronger conditions
(a)--(e) of \reflem{1}/\ref{lem-1ST} follows with help of the new condition (1).

In the proof of \reflem{2ST} the bisimulation constructed in the proof of \reflem{2} is strengthened to
a split bisimulation by taking account of the sets $U'$ and $U$ of transitions currently firing in
$N'$ and $N$. Here we need to require that $U'$ and $U$ carry the same multiset of labels.
Moreover, the preplaces of $U'$ and $U$ need to be added to $M'$ and $M$ when determining that they
are reachable markings, and in relating these markings to each other; for these purposes we thus
use the markings we would have had before starting the transitions that are currently firing. On the other
hand, $M'$ and $M$ themselves need to be markings (i.e.\ put a nonnegative number of tokens in
each place), and in conditions (a)--(c) only those transitions matter that can be fired from $M'$ and $M$
themselves---without the preplaces of $U'$ and $U$.

In \reflem{2ST} a relation is explored between markings $\bar{M}$ and $\bar{M}+\marking{H}$
(where $\bar{M}$ is $M'+\!\precond{U'}+(M_0-M'_0)$ of \reflem{2ST}, $H:=G$, and
$\bar{M}+\marking{H}$ is $M+\!\precond{U}$ of \reflem{2ST}).
In such a case, we can think of $\bar{M}$ as an ``original marking'', and of
$\bar{M}+\marking{H}$ as a modification of this marking by the token replacement
$\marking{H}$. The next lemma provides a method to trace certain places $s$
marked by $\bar{M}+\marking{H}$ (or transitions $t$ that are enabled under $\bar{M}+\marking{H}$)
back to places that must have been marked by $\bar{M}$ before taking into account the
token replacement $\marking{H}$. Such places are called \emph{faithful origins}
of $s$ (or $t$). In tracking the faithful origins of places and transitions, we
assume that the places marked by $\bar{M}$ are taken from a set $S_+$ and the
transitions in $H$ from a set $T_+$. In \reflem{origin} we furthermore assume that the flow
relation restricted to $S\cup T_+$ is acyclic. We will need this lemma in proving the
correctness of our final method of proving $N \approx^\Delta_{bSTb} N'$.

\begin{defi}\label{df-faithful}
Let $N=(S,T,F,M_0,\ell)$ be a Petri net, $T_+\subseteq T$ a set of
transitions and $S_+\subseteq S$ a set of places.
\begin{iteMize}{$\bullet$}
\item
A \emph{path} in $N$ is an alternating sequence $\pi=x_0 x_1 x_2 \cdots x_n \in (S\cup T)^*$ of
places and transitions, such that $F(x_i,x_{i+1})>0$ for $0\mathbin\leq i \mathbin< n$.
The \emph{arc weight} $F(\pi)$ of such a path is the product $\Pi_0^{n-1}F(x_i,x_{i+1})$.
\item
  A place $s\in S$ is called \emph{faithful} w.r.t.\ $T_+$ and $S_+$
  iff $|\{s\}\cap S_+| + \sum_{t\in T_+}F(t,s)=1$.
\item
A path $x_0 x_1 x_2 \cdots x_n \in (S\cup T)^*$ from $x_0$ to $x_n$ is \emph{faithful}
w.r.t.\ $T_+$ and $S_+$ iff all intermediate nodes $x_i$ for $0\leq i < n$ are either
transitions in $T_+$ or faithful places w.r.t.\ $T_+$ and $S_+$.
\item
  For $x\in S\cup T$, the \emph{infinitary multiset} $^*x\in(\nat\cup\{\infty\})^{S_+}$
of \emph{faithful origins} of $x$ is given by
$^*x(s)=\sup\{F(\pi)\mid \pi \mbox{ is a faithful path from $s\in S_+$ to $x$}\}$. (So
$^*x(s)=0$ if no such path exists.)
\end{iteMize}
\end{defi}
\noindent
Suppose a marking $M$ is reachable from a marking $\bar M\in \nat^{S_+}$
by firing transitions from $T_+$ only. So $M=\bar M+\marking{H}$ for
some $H\fin \nat^{T_+}$. Then, if a faithful place $s$ bears a token
under $M$---i.e.\ $M(s)>0$---this token has a unique source:
if $s\in S_+$ it must stem from $\bar M$ and otherwise it must be
produced by the unique transition $t\mathbin\in T_+$ with $F(t,s)\mathbin=1$.

Now consider a period in the evolution of the net $N$ that starts with
the marking $\bar M$, and during which only transitions from $T_+$ fire.
Suppose $\pi=x_0 x_1 x_2 \cdots x_n$ is a faithful path from a place
$x_0\in S_+$ to a either a faithful place $x_n$ that gets marked
at some point during this period or a transition $x_n$ that fires
during (or right after) this period. In that case a token, left on
$x_0$ by the marking $\bar M$, must have travelled along that path from
$x_0$ to $x_n$---where a token is understood to visit a transition
when that transition fires. Namely, if $x_{i+1}$ is a transition that
fired at some point, then its (faithful) preplace $x_i$ must have been
marked right beforehand; and if a faithful place $x_{j+i}$ was marked at
some point, then $x_{j+i}\notin S_+$ and the token in $x_{j+i}$ must
have been produced by the transition $x_i\in T_+$.

Note that $F(\pi)$ is the product of all arc weights in the path on
arcs from places to transitions; for all the weights on arcs from
transitions in $T_+$ to faithful places are 1.  Taking arc weights
into account, for every token in $x_n$ as many as $F(\pi)$ token must
have started in $x_0$. Namely, for a transition $x_{i+1}$ to fire
once, $F(x_i,x_{i+1})$ tokens must have come from place $x_i$, and for
each token in a faithful place $x_{j+1}$, the transition $x_j$ must
have fired once.

In a net without arc weights, $^*x$ is always a set, namely the set of 
places $s$ in $S_+$ from which the flow relation of the net admits a path to $x$ that passes only
through faithful places and transitions from $T_+$ (with the possible exception of $x$ itself).
For nets with arc weights, the underlying set of $^*x$ is
the same, and the multiplicity of $s \in \mbox{}^*x$ is obtained by multiplying all
arc weights on the qualifying path from $s$ to $x$; in case of multiple such paths, we
take the upper bound over all such paths (which could yield the value $\infty$).
It follows from the analysis above that if a faithful place $x$ gets
marked, or a transition $x$ enabled, during a period as described above,
then at least $\mbox{}^*x(s)$ tokens must have been present in
$s$ at the beginning of this period. \reflem{origin} formalises this analysis by comparing a marking
$\bar{M}+\marking{H}$ that marks or enables $x$ (possibly multiple times) with the marking $\bar{M}$
that marks the faithful origins $\mbox{}^*x$ of $x$. Here $H \fin \nat^{T_+}$ is the multiset of
transitions whose firing converts $\bar{M}$ into $\bar{M}+\marking{H}$. However, \reflem{origin}
does not require that this multiset actually can be fired in any particular order.
To enable that generalisation, it must assume that $F\upharpoonright(S\cup T_+)$ is acyclic.

For $k \ne 0$, we have
$k\cdot{}^*x(s)=\sup\{k \cdot F(\pi)\mid \pi \mbox{ is a faithful path from $s\in S_+$ to $x$}\}$.
In order to also have this equality for $k=0$ and ${}^*x(s) = \infty$ we define
$0\cdot\infty := 0$ in this context.

\begin{obs}\label{obs-origin}
  Let $(S,T,F,M_0,\ell)$ be a Petri net, $T_+\subseteq T$ a set of
  transitions and ${S_+}\subseteq S$ a set of places.
For faithful places $s$ and transitions $t\in T$ we have\vspace{-1.5pt}
$$^*s=\left\{\begin{array}{@{}ll@{}}\{s\} & \mbox{if}~s\in {S_+}\\
  ^*t   & \mbox{if}~t\in T_+ \wedge F(t,s)=1
            \end{array}\right.
\qquad\qquad\vspace{-1.5pt}
^*t=\bigcup\{F(s,t)\cdot\mbox{}^*s \mid s\in \precond{t} \wedge s {\rm ~faithful}\}.$$
\end{obs}

\begin{lem}\label{lem-origin}\label{lem-faithful origins}
  Let $(S,T,F,M_0,\ell)$ be a Petri net, $T_+\subseteq T$ a set of transitions
  such that $F\upharpoonright(S\cup T_+)$ is acyclic, and ${S_+}\subseteq S$ a set of places. 
  Let $\bar M\in \nat^{S_+}$ and $H\fin \nat^{T_+}$, such that $\bar M+\marking{H}\in\nat^S$ (\ie
  places occur only non-negatively in $\bar M+\marking{H}$). Then
\begin{enumerate}[\em(a)]
\item
  for any faithful place $s$ w.r.t.\ $T_+$ and ${S_+}$ we have
$(\bar M+\marking{H})(s)\cdot\mbox{}^*s \leq \bar M$;
\item
for any $k\in\nat$, and any transition $t$ with $(\bar M+\marking{H}) [k\cdot\{t\}\rangle$,
we have $k\cdot\mbox{}^*t\leq \bar M$.
\end{enumerate}
\end{lem}

\begin{proof}
We apply induction on $|H|$.
In the base case, $H=\emptyset$, which formally is included in the induction step,
(a) follows directly from the assumption that $\bar M\in \nat^{S_+}$
and the observation that $^*s = \{s\}$.
\\[.5ex]
(a).
When $(\bar M+\marking{H})(s)=0$ it trivially follows that $(\bar M+\marking{H})(s)\cdot\mbox{}^*s \leq \bar M$.
So suppose $(\bar M+\marking{H})(s)>0$. Then either $s\in {S_+}$ or there is a
unique $t\in T_+$ with $H(t)>0$ and $F(t,s)=1$.
In the first case, using that $s\in\postcond{u}$ for no $u\in T_+$, we
have $(\bar M+\marking{H})(s)\leq \bar M(s)$, so $(\bar M+\marking{H})(s)\cdot\mbox{}^*s
  \leq \bar M(s)\cdot \{s\} \leq \bar M$.
  In the latter case, we have $(\bar M+\marking{H})(s) \leq \bar M(s)+\sum_{u\in T_+}H(u)\cdot F(u,s) =
\bar M(s) + H(t) = H(t)$ and $^*s = \mbox{}^*t$.
Thus:
\begin{equation}\label{Referee3}
 (\bar M + \marking{H})(s) \cdot {}^*s \leq H(t)\cdot {}^*t\;.
\end{equation}

Let $U:=\{u\in T_+\mid H(u)>0 \wedge u F^+ t\}$ be the set of
transitions occurring in $H$ from which the flow relation of the net
offers a non-empty path to $t$. As $F\upharpoonright(S\cup T_+)$ is acyclic,
$t\notin U$, so $H\!\upharpoonright\! U < H$.
Let $s'$ be any place with $s'\in\precond{u}$ for some transition
$u\in U$. Then, by construction of $U$, it cannot happen
that $s'\in \postcond{v}$ for some transition $v\notin U$ with $H(v)>0$.
Hence $(\bar M+\marking{H\!\upharpoonright\! U})(s')\geq (\bar M+\marking{H})(s')\geq 0$.
Moreover, for any other place $s''$ we have
$\precond{(H\!\upharpoonright\! U)}(s'')=0$ and thus
$(\bar M+\marking{H\!\upharpoonright\! U})(s'')\geq \bar M(s'')\geq 0$.
It follows that $\bar M+\marking{H\!\upharpoonright\! U}\in\nat^S$.

For each $s'''\in \precond{t}$ we have
$\postcond{(H-H\!\upharpoonright\! U)}(s''')=0$ and
$\precond{(H-H\!\upharpoonright\! U)}(s''') \ge H(t)\cdot \precond{t} (s''')$ and therefore
$0 \leq (\bar M+\marking{H})(s''') \leq (\bar M+\marking{H\!\upharpoonright\! U})(s''') - H(t)\cdot \precond{t}(s''')$.
For this reason, $H(t)\cdot\precond{t} \leq \bar M+\marking{H\!\upharpoonright\! U}$.
It follows that $(\bar M+\marking{H\!\upharpoonright\! U})[H(t)\cdot \{t\}\rangle$.
Thus, by (\ref{Referee3}) and induction,
$(\bar M+\marking{H})(s)\cdot\mbox{}^*s \leq H(t)\cdot \mbox{}^*t \leq \bar M$.\vspace{1ex}

\noindent
(b).
Let $(\bar M+\marking{H}) [k\cdot\{t\}\rangle$.
For any faithful $s\in \precond{t}$ we have $(\bar M+\marking{H})(s)\geq k\cdot F(s,t)$,
and thus, using (a),\vspace{-1ex}
 $$ k\cdot F(s,t)\cdot\mbox{}^*s
        \leq  (\bar M+\marking{H})(s)\cdot\mbox{}^*s
        \leq \bar M\;.\vspace{1ex}$$
Therefore, by \refobs{origin}, $k\cdot\mbox{}^*t =
\bigcup\{k\cdot F(s,t)\cdot\mbox{}^*s \mid s\in \precond{t} \wedge s {\rm ~faithful}\}\leq \bar M$.
\end{proof}
\noindent
As a (forthcoming) application of \reflem{faithful origins}---in fact the only one we'll need in this
paper---consider the branching split bisimulation with explicit divergence between a net $N'$ and
its conflict replicating implementation $N$ that is constructed according to the proof of \reflem{2ST}.
When a split marking $(M',U')$ is related to $(M,U)$, then
$M+\precond{U}=M'+\precond{U'}+(M_0\mathord-M'_0)+\marking{G}$ for a signed multiset $G$ of internal
transitions of $N$. Furthermore suppose that $G$ is a true multiset over the set of transitions
$T_+$, consisting of $\dist$, $\ini\cdot\fire$ and $\trans{j}\cdot\fire$ only (for arbitrary $p$,
$j$ and $h$). Take $\bar M:=M'+\precond{U'}+(M_0\mathord-M'_0)$, $H:=G$ and thus
$M+\precond{U}=\bar M +\marking{H}$. Let $S_+:=S' \cup \{s \in S \mid (M_0\mathord-M'_0)(s) >0\}$.
Then $~ p ~~ \dist ~~ p_i ~~ \ini[i]\cdot\fire ~~ \Pre^i_j ~~ \exec{j} ~$
is a faithful path from $p$ to $\exec{j}$. The arc weight of this path is
$F'(p,i)$. So \plat{$\mbox{}^*\exec{j} \geq F'(p,i)$.} Thus if \plat{$\exec{j}$} is enabled under $M+\precond{U}$
then $\bar{M}$ must place at least $F'(p,i)$ tokens in the place $p$. As this reasoning applies to
every preplace $p$ of $i$, it follows that $i$ is enabled under $M'+\precond{U'}$.

The following theorem is the main result of this section. It presents a method for proving
$N \approx^\Delta_{bSTb} N'$ for $N$ a net and $N'$ a plain net. Its main advantage w.r.t.\ 
directly using the definition, or w.r.t.\ application of \reflem{1ST} or \ref{lem-2ST}, is
the replacement of requirements on the dynamic behaviour of nets by structural requirements.
Such requirements are typically easier to check.
Replacing the requirement ``$M+\precond{U}\in[M_0\rangle_N$'' in Condition~\ref{lastST} by
``$M+\precond{U}\in\nat^S$'' would have yielded an even more structural version of this
theorem; however, that version turned out not to be strong enough for the verification
task performed in \refsec{correctness}.

\begin{thm}\label{thm-3ST}
Let $N=(S,T,F,M_0,\ell)$ be a net and $N'=(S',T',F',M'_0,\ell')$ be a plain
net with $S'\subseteq S$ and $M'_0=M_0\upharpoonright S'$.
Suppose there exist sets $T_+ \subseteq T$ and $T_-\subseteq T$
and a class $\NF\subseteq \Int^T$, such that
\begin{enumerate}[\em(1)]
\item $F\upharpoonright(S\cup T_+)$ is acyclic.\label{acyclic+ST}
\item $F\upharpoonright(S\cup T_-)$ is acyclic.\label{acyclic-ST}
\item $\forall t\inp T,~\ell(t)\mathbin{\neq}\tau.~ \exists t'\inp T',~\ell(t')\mathbin=\ell(t).~
       \left(\precond{t'} \leq \mbox{}^*t \wedge
       \exists G\fin \nat^T,~\ell(G)\equiv\emptyset.~ \marking{t'}=\marking{t+G}\right)$.\\
  Here $\mbox{}^*t$ is the multiset of faithful origins of $t$ w.r.t.\ $T_+$ and
  $S'\cup\{s\in S \mid M_0(s)>0\}$.\label{matchingST}
\item There exists a function $f:T\rightarrow\nat$ with $f(t)>0$ for all $t\inp T$,
  extended to $\Int^T$ as in \refdf{multiset}, such that
  for each $G\fin \Int^T$ with $\ell(G)\equiv\emptyset$ there is an $H\fin \NF$
  with $\ell(H)\equiv\emptyset$, $\marking{H}=\marking{G}$ and $f(H)=f(G)$.
  \label{normalformST}
\item For every $M'\in\nat^{S'}$, $U'\in\nat^{T'}$ and $U\in\nat^{T}$ with
  $\ell(U)=\ell'(U')$ and $M'+\precond{U'}\in[M'_0\rangle_{N'}$, there is an \plat{$H_{M',U}\fin\nat^{T_+}$}
  with $\ell(H_{M',U})\equiv\emptyset$, such that\label{lastST}
  for each $H\mathbin{\fin} \NF$ with $M:=M'+\precond{U'}+(M_0-M'_0)+\marking{H}-\precond{U}\in\nat^S$
  and $M+\precond{U}\in[M_0\rangle_N$:
  \begin{enumerate}[\em(a)]
  \item $M_{M',U}:=M'+\precond{U'}+(M_0-M'_0)+\marking{H_{M',U}}-\precond{U}\in\nat^S$,\label{markingST}
  \item if $M'\goesto[a]$ with $a\in\Act$ then \plat{$M_{M',U}\goesto[a]$},\label{matchST}
    \item $H\leq H_{M',U}$.\label{upperboundST}
    \item if $H(u)<0$ then $u\in T_-$,\label{T-ST}
    \item if $H(u)<0$ and $H(t)>0$ then $\precond{u} \cap \precond{t} = \emptyset$,
          \label{disjoint preplacesST}
    \item if $H(u)<0$ and $(M+\!\precond{U})[t\rangle$ with $\ell(t)\neq\tau$
          then $\precond{u} \cap \precond{t} = \emptyset$,
          \label{disjoint preplaces 2ST}
    \item if $(M+\!\precond{U})[\{t\}\mathord+\{u\}\rangle$ and
          and $t',u'\in T'$ with $\ell'(t')=\ell(t)$ and $\ell'(u')=\ell(u)$, then
          $\precond{t'}\cap\precond{u'}=\emptyset$.\label{concurrent}
  \end{enumerate}
\end{enumerate}
Then $N \approx^\Delta_{bSTb} N'$.
\end{thm}

\begin{proofNobox}
It suffices to show that Condition~\ref{clause2ST} of \reflem{2ST} holds
(as Condition~\ref{clause1ST} of \reflem{2ST} is part of Condition~\ref{matchingST} above).
So let $G\fin \Int^T$ with $\ell(G)\equiv\emptyset$, ~\plat{$M'\inp\nat^{S'}$},
\plat{$U'\inp\nat^{T'}\!$} and $U\inp\nat^T\!$ with $\ell'(U')\mathbin=\ell(U)$,
      ~$M'\mathord+\!\precond{U'}\in [M'_0\rangle_{N'}$,
      ~$M:=M'\mathord+\!\precond{U'}\mathord+(M_0\mathord-M'_0)\mathord+\marking{G}
      \mathord-\!\precond{U}\inp\nat^S$ and $M+\precond{U}\in[M_0\rangle_N$.
\begin{enumerate}[(a)]
\item
Suppose $M\goesto[\tau]M_1\goesto[\tau]M_2\goesto[\tau] \cdots$.
Then there are transitions $t_i\in T$ with $\ell(t_i)=\tau$, for all $i\mathbin\geq 1$, such that
$M[t_1\rangle M_1[t_2\rangle M_2[t_3\rangle \cdots$.
As also $(M+\!\precond{U})[t_1\rangle (M_1+\!\precond{U})[t_2\rangle (M_2+\!\precond{U})[t_3\rangle \cdots$,
it follows that $(M_i+\!\precond{U})\inp[M_0\rangle_N$ for all $i\geq 1$.
Let $G_0:=G$ and for all $i\geq 1$ let $G_{i}:=G_{i-1}+\{t_i\}$.
Then $\ell(G_i)\equiv\emptyset$ and $M_i=M'+\!\precond{U'}+(M_0-M'_0)+\marking{G_i}-\!\precond{U}$.
Moreover, $f(G_{i})=f(G_{i-1})+f(t_i) > f(G_{i-1})$.
For all $i\geq 0$, using Condition~\ref{normalformST},
let $H_i\mathbin{\fin}\NF$ be so that $\marking{H_i}\mathbin=\marking{G_i}$ and $f(H_i)\mathbin=f(G_i)$.
Then $M_i=M'+\!\precond{U'}+(M_0-M'_0)+\marking{H_i}-\!\precond{U}$ and $f(H_0)<f(H_1)<f(H_2)<\cdots$.
However, from Condition~\ref{upperboundST} we get $f(H_i)\leq f(H_{M'})$ for all $i\geq 0$.
The sequence $M\goesto[\tau]M_1\goesto[\tau]M_2\goesto[\tau] \cdots$ therefore must be finite.
\item
Now suppose $M' \goesto[a]$ with $a\in\Act$.
By Condition~\ref{normalformST} above there exists an $H\fin\NF$
such that $\ell(H)\equiv\emptyset$ and $\marking{H}=\marking{G}$, and hence
$M=M'+\!\precond{U'}+(M_0-M'_0)+\marking{H}-\!\precond{U}$.
Let $H_-:=\{u\in T\mid H(u)<0\}$.
\begin{iteMize}{$\bullet$}
\item First suppose $H_-\neq\emptyset$.
By Condition~\ref{T-ST}, $H_-\subseteq T_-$.
  By Condition~\ref{acyclic-ST}, the relation $<_-:=(F\upharpoonright(S\cup T_-))^+$
  is a partial order on $S\cup T_-$, and hence on $H_-$.
  Let $u$ be a minimal transition in $H_-$ w.r.t.\ $<_-$.
  By definition, for all $s\in S$,
  \begin{equation*}
  M(s)=M'(s)+\!\precond{U'}(s)+(M_0-M'_0)(s)+\mbox{}
  \end{equation*}
  \begin{equation}\label{token countST}
  \!\sum_{t\in T}H(t)\cdot F(t,s)+\!\sum_{t\in T}\!-H(t)\cdot F(s,t)
  +\!\sum_{t\in U}\!-U(t)\cdot F(t,s).
  \end{equation}%
  As $M'_0=M_0\upharpoonright S'$, we have $M'_0\leq M_0$.
  Hence the first three summands in this equation are always nonnegative.
  Now assume $s\in\precond{u}$. Since $u$ is minimal w.r.t.\ $<_-$,
  there is no $t\in T$ with $H(t)<0$ and $F(t,s)\neq 0$.
  Hence also all summands $H(t)\cdot F(t,s)$ are nonnegative.
  By Condition~\ref{disjoint preplacesST}, there is no
  $t\in T$ with $H(t)>0$ and $F(s,t)\neq 0$,
  so all summands $-H(t)\cdot F(s,t)$ are nonnegative as well.
  By Condition~\ref{disjoint preplaces 2ST}, there is no
  $t\in T$ with $U(t)>0$ and $F(s,t)\neq 0$,
  for this would imply that $\ell(t)\neq\tau$ and $(M+\!\precond{U})[t\rangle$, so
  no summands in (\ref{token countST}) are negative.
  Thus $0\leq -H(u)\cdot F(s,u) \leq M(s)$.
  Since $H(u)\leq -1$, this implies $M(s)\geq F(s,u)$.
  Hence $u$ is enabled in $M$. As $\ell(u)=\tau$, we have $M\goesto[\tau]$.
\item Next suppose $H_-\!=\emptyset$ but $H\neq H_{M',U}$.
  Let $H^\smile:=\{u\in T\mid H_{M',U}(u)-H(u)>0\}$.
  Then $H^\smile\neq\emptyset$ by Condition~\ref{upperboundST}.
  Since \plat{$H_{M',U}\fin\nat^{T_+}\!\!$}, $H^\smile\subseteq T_+$.
  By Condition~\ref{acyclic+ST}, $<_+:=(F\upharpoonright(S\cup T_+))^+$
  is a partial order on $S\cup T_+$, and hence on $H^\smile$.
  Let $u$ be a minimal transition in $H^\smile$ w.r.t.\ $<_+$.
  We have $M=M'+\!\precond{U'}+(M_0-M'_0)+\marking{H_{M',U}+(H-H_{M',U})}-\!\precond{U}=M_{M',U}+\marking{H-H_{M',U}}$.
  Hence, for all $s\in S$,
  \vspace{-1ex}

  {\small
  \begin{equation}\qquad\qquad\label{token count 2ST}
  M(s)=M_{M',U}(s)+\sum_{t\in T}(H-H_{M',U})(t)\cdot F(t,s)+\sum_{t\in T}-(H-H_{M',U})(t)\cdot F(s,t)\;.
  \end{equation}}%
  By Condition~\ref{markingST}, $M_{M',U}\in\nat^S$.
  By Condition~\ref{upperboundST}, $H- H_{M',U}\leq 0$.
  For $s\in \precond{u}$ there is moreover no $t\in H^\smile$ with $s\in\postcond{t}$,
  so no $t\in T$ with $(H-H_{M',U})(t)<0$ and $F(t,s)\neq 0$.
  Hence no summands in (\ref{token count 2ST}) are negative.
  It thereby follows that $0\leq -(H\mathord-M_{M',U})(u)\cdot F(s,t) \leq M(s)$.
  Since $(H\mathord-H_{M',U})(u)\leq -1$, this implies $M(s)\geq F(s,u)$.
  Hence $u$ is enabled in $M$. As $\ell(u)=\tau$, we have $M\goesto[\tau]$.
\item Finally suppose $H= H_{M',U}$. Then $M=M_{M',U}$ and
  $M\goesto[a]$ follows by Condition~\ref{matchST}.
\end{iteMize}
\item
  Next suppose $M\goesto[a]$ with $a\in\Act$.
  Then there is a $t\in T$ with $\ell(t)=a\neq\tau$ and $M[t\rangle$.
  So $(M+\!\precond{U})[t\rangle$.
  We will first show that $(M'+\!\precond{U'})\goesto[a]$.
  By Condition~\ref{normalformST} there exists an $H_0\fin\NF\subseteq \Int^T$
  such that $\ell(H_0)\equiv\emptyset$ and $\marking{H_0}=\marking{G}$, and hence
  $M+\!\precond{U}=M'+\!\precond{U'}+(M_0-M'_0)+\marking{H_0}\in[M_0\rangle_N$.
  For our first step, it suffices to show that whenever $H \mathbin{\fin} \NF$ with
  $M_H := M'+\!\precond{U'}+(M_0-M'_0)+\marking{H}\inp[M_0\rangle$
  and $M_H [t\rangle$, then $(M'+\!\precond{U'})\goesto[a]$.
  We show this by induction on $f(H_{M',U}-H)$, observing that
  $f(H_{M',U}-H)\in\nat$ by Conditions~\ref{upperboundST} (with empty $U$) and~\ref{normalformST}.

  We consider two cases, depending on the emptiness of $H_-:=\{u\in T\mid H(u)<0\}$.

  First assume $H_-\!\mathbin=\emptyset$. Then $H\mathbin{\fin}\nat^T\!$.
  By Condition~\ref{upperboundST} (with empty $U$) we even have $H\mathbin{\fin}\nat^{T_+}\!\!$. 
  Let $\mbox{}^*t$ denote the multiset of faithful origins of $t$ w.r.t.\ $T_+$ and
  $S_+ := S'\cup\linebreak[3]\{s\in S \mid M_0(s)>0\}$.
  By \reflem{faithful origins}(b), taking $k\mathbin=1$ and $\bar M := M'+\!\precond{U'}+(M_0-M'_0)$,
  and using Condition~\ref{acyclic+ST} of \refthm{3ST},
  $^*t \leq M'+\!\precond{U'}+(M_0-M'_0)$. So by Condition~\ref{matchingST} of \refthm{3ST}
  there is a $t'\in T'$ with $\ell(t')=\ell(t)$ and $\precond{t'} \leq M'+\!\precond{U'}+(M_0-M'_0)$.
  Since $\precond{t'} \in \nat^{S'}$ and $M'_0=M_0\!\upharpoonright\! S'$, this implies
  $\precond{t'} \leq M'+\!\precond{U'}$.
  It follows that $(M'+\!\precond{U'})[t'\rangle_{N'}$ and hence $(M'+\!\precond{U'})\goesto[a]$.

  Now assume $H_- \neq \emptyset$.
  By the same proof as for (b) above, case $H_- \neq \emptyset$,
  there is a transition $u\in H_-$ that is enabled in $M_H$.
  So $M_H[u\rangle M_1$ for some $M_1\in[M_0\rangle_N$, and $M_1=M'+\!\precond{U'}+(M_0-M'_0)+\marking{H+u}$.
  By Condition~\ref{disjoint preplaces 2ST} of \refthm{3ST} (still with empty $U$),
  $\precond{u}\cap\precond{t}=\emptyset$, and thus $M_1[t\rangle$.
  By Condition~\ref{normalformST} of \refthm{3ST} there exists an $H_1\mathbin{\fin}\NF$
  such that $\ell(H_1)\mathbin{\equiv}\emptyset$, $\marking{H_1}\mathbin=\marking{H+u}$, and
  $f(H_1)\mathbin=f(H+u)\mathbin>f(H)$. Thus $M_1=M_{H_1}$ and $f(H_{M',U}-H_1)<f(H_{M',U}-H)$.
  By induction we obtain $(M'+\!\precond{U'})\goesto[a]$.

  By the above reasoning, there is a $t'\in T'$ such that
  $\ell'(t')=\ell(t)$ and $(M'+\!\precond{U'})[t'\rangle$.
  Now take any $u'\in U'$. Then there must be an $u\in U$ with
  $\ell'(u')=\ell(u)$. Since $M[t\rangle$, we have $(M+\!\precond{U})[\{t\}\mathord+\{u\}\rangle$ and
  by Condition~\ref{concurrent} we obtain $\precond{t'}\cap\precond{u'}=\emptyset$.
  It follows that $M'[t'\rangle$, and hence $M'\goesto[a]$.
  \qed
\end{enumerate}
\end{proofNobox}

\noindent
\refthm{3ST} will be applied in \refsec{correctness} to show the correctness of our conflict
replicating implementation $N$ of a given net $N'$. A crucial observation about $N$ is that
its internal transitions can be partitioned into a set $T_+$ of transitions (3 boxes in
\reffig{conflictrepl-expanded}) that have to occur before firing $\exec{j}$ (for some $i$ and $j$)
and a set $T_-$ of transitions (14 boxes) that can only occur afterwards. In the construction of our
bisimulation we consider markings of the form $M'+\precond{U'}+(M_0\mathord-M'_0)+\marking{H}$,
where $H$ is a signed multiset of internal transitions that tells how much the marking deviates from
the marking $M'+\precond{U'}+(M_0\mathord-M'_0)$ of $N$. The bisimulation relates both markings of
$N$ to the marking $M'+\precond{U'}$ of $N'$. When an internal transition of $N$ fires, the related
marking of $N'$ remains the same. However, when $N$ fires a visible transition $\exec{j}$ then the
related marking of $N'$ becomes $M'+\precond{U'} + \marking{i}$, so in view of the structural
property in \reflem{2ST}(1), a new set $H'$ can be calculated as $H':=H-G$, where $G$ is the signed
multiset for which $\marking{i}=\marking{\exec{j}+G}$. A consequence of this is that elements of
$T_+$ only occur with positive multiplicities in $H$, whereas elements of $T_-$ occur only with
negative multiplicities. 

To be precise, it may be that two different sets $H_1$ and $H_2$ yield the same token replacement,
i.e.\ $\marking{H_1}=\marking{H_2}$. As a result of this, there may be multiple ways to write a
marking as $M'+\precond{U'}+(M_0\mathord-M'_0)+\marking{H}$ for given $M'$ and $U'$.
The above applies only when converting the signed multisets $H$ to a normal form $\NF$ that
eliminates this ambiguity.

For given $M'$ and $U'$, the multiset $H_{M',U}$ is an upper bound of the possible choices of $H$
for which $M'+\precond{U'}+(M_0\mathord-M'_0)+\marking{H}$ can be a reachable marking. This is
expressed by Condition~\ref{upperboundST}. If all internal transitions in $H_{M',U}$ have fired, the
next transition must be an external one. Now the conditions of \refthm{3ST} guarantee that as long
as this upper bound is not reached, the net $N$ can perform internal actions, and when it is reached
(and possibly also beforehand) it can perform the same actions as the net $N'$ under marking $M'$.
Condition~\ref{normalformST} moreover guarantees that this upper bound will be reached in finitely
many steps. Due the the need to renormalise the signed multisets $H$ after adding elements to them,
this is not straightforward.

These considerations imply that transitions fired by $N'$ can be simulated by $N$. The other
direction involves similar arguments, together with an application of \reflem{faithful origins}.

\subsection*{Digression: Interleaving semantics}

Above, a method is presented for establishing the equivalence of two Petri nets, one of
which known to be \hyperlink{plain}{plain}, up to branching ST-bisimilarity with explicit divergence.
Here, we simplify this result into a method for establishing the equivalence of the two
nets up interleaving branching bisimilarity with explicit divergence.
This result is not applied in the current paper.

\begin{lem}\label{lem-1}
Let $N=(S,T,F,M_0,\ell)$ and $N'=(S',T',F',M'_0,\ell')$ be two nets, $N'$ being plain.
Suppose there is a relation $\Rel  \subseteq \nat^{S}\times\nat^{S'}$ such that
\begin{enumerate}[\em(a)]
\item $M_0\Rel M'_0$,
\item if $M_1\Rel M_1'$ and
  $M_1\goesto[\tau]M_2$ then $M_2\Rel M_1'$,
\item if $M_1\Rel M_1'$ and
  $M_1\goesto[a]M_2$ for some $a\in\Act$ then
  $\exists M'_2.~M'_1\goesto[a]M'_2 \wedge M_2\Rel  M'_2$,
\item if $M_1\Rel M_1'$ and
  $M'_1\goesto[a]$ for some $a\in\Act$ then either
  $\mathord{M_1 \goesto[a]}$ or $\mathord{M_1 \goesto[\tau]}$
\item and there is no infinite sequence $M\goesto[\tau] M_1\goesto[\tau] M_2\goesto[\tau] \cdots$
  with $M\Rel M'$ for some $M'$.
\end{enumerate}
Then $N$ and $N'$ are interleaving branching bisimilar with explicit divergence. 
\end{lem}

\begin{proof}
This follows directly from \reflem{plain branching bisimilarity}
by taking $(\st_1,\tr_1,\inist_1)$ and $(\st_2,\tr_2,\inist_2)$ to be
the interleaving LTSs associated to $N$ and $N'$ respectively, using the fact
that the LTS associated to a plain net is deterministic.
\end{proof}

\begin{lem}\label{lem-2}
Let $N=(S,T,F,M_0,\ell)$ be a net and $N'=(S',T',F',M'_0,\ell')$ be a plain
net with $S'\subseteq S$ and $M'_0=M_0\upharpoonright S'$.
Suppose:
\begin{enumerate}[\em(1)]
\item $\forall t\inp T,~\ell(t)\neq\tau.~ \exists t'\inp T',~\ell(t')=\ell(t).~
       \exists G\fin \nat^T,~\ell(G)\equiv\emptyset.~ \marking{t'}=\marking{t+G}$.
      \label{clause1}
\item For any $G\mathbin{\fin} \Int^T$ with $\ell(G)\mathbin\equiv\emptyset$, $M'\inp[M'_0\rangle_{N'}$
      and $M:=M'\mathord+(M_0\mathord-M'_0)\mathord+\marking{G}\inp[M_0\rangle_N$, it holds that:\label{clause2}
\begin{enumerate}[\em(a)]
\item there is no infinite sequence $M\goesto[\tau] M_1\goesto[\tau] M_2\goesto[\tau] \cdots$,\label{2a}
\item if $M' \goesto[a]$ with $a\in \Act$
  then $\mathord{M \goesto[a]}$ or $\mathord{M \goesto[\tau]}$\label{2b}
\item and if $M \goesto[a]$ with $a\in \Act$ then $M' \goesto[a]$.\label{2c}
\end{enumerate}
\end{enumerate}
Then $N$ and $N'$ are interleaving branching bisimilar with explicit divergence.
\end{lem}

\begin{proofNobox}
Define $\Rel \subseteq \nat^{S}\times \nat^{S'}$ by
$$M \Rel  M' :\Leftrightarrow M'\inp[M'_0\rangle_{N'} \wedge \exists
G\mathbin{\fin}\Int^T.~M=M'\mathord+(M_0\mathord-M'_0)\mathord+\marking{G}\inp[M_0\rangle_N
\wedge \ell(G)\equiv\emptyset.$$
It suffices to show that $\Rel $ satisfies Conditions (a)--(e) of \reflem{1}.
\begin{enumerate}[(a)]
\item Take $G=\emptyset$.
\item Suppose $M_1\Rel M_1'$ and $M_1\goesto[\tau]M_2$.
  Then $\exists G\fin\Int^T\!\!.~M_1=M'_1+(M_0-M'_0)+\marking{G} \wedge \ell(G)\equiv\emptyset$
  and $\exists t\mathbin\in T.~\ell(t)=\tau \wedge M_2=M_1+\marking{t}=M'_1+(M_0-M'_0)+\marking{G+t}$.
  Moreover, $M_1\in[M_0\rangle_N$ and hence $M_2\in[M_0\rangle_N$.
  Furthermore, $M_1'\in[M'_0\rangle_{N'}$ and $\ell(G+t)\equiv\emptyset$, so $M_2\Rel M'_1$.
\item Suppose $M_1\Rel M_1'$ and $M_1\goesto[a]M_2$.
  Then $\exists G\fin\Int^T\!\!.~M_1=M'_1+(M_0-M'_0)+\marking{G} \wedge \ell(G)\equiv\emptyset$
  and $\exists t\mathbin\in T.~\ell(t)=a\neq\tau \wedge M_2=M_1+\marking{t}=M'_1+(M_0-M'_0)+\marking{G+t}$.
  Moreover, $M_1\in[M_0\rangle_N$ and hence $M_2\in[M_0\rangle_N$.
  Furthermore, $M_1'\in[M'_0\rangle_{N'}$.
  By Condition~\ref{clause1} of \reflem{2}, $\exists t'\inp T',~\ell(t')\mathbin=\ell(t).\linebreak[3]\
       \exists G_t\mathbin{\fin} \nat^T,~\ell(G_t)\equiv\emptyset.~ \marking{t}=\marking{t'-G_t}$.
  Substitution of $\marking{t'-G_t}$ for $t$ yields $M_2=M'_1+\marking{t'}+(M_0\mathord-M'_0)+\marking{G-G_t}$.
  By Condition~\ref{2c}, $M'_1\goesto[a]$, so $M'_1\goesto[a] M'_2$ for some $M'_2\in[M'_0\rangle_{N'}$.
  As $t'$ is the only transition in $T'$ with $\ell'(t')=a$, we must have $M'_1[t'\rangle M'_2$.
  So $M'_1+\marking{t'}=M'_2$. Since $\ell(G-G_t)\equiv\emptyset$ it follows that $M_2\Rel M'_2$.
\item Follows directly from Condition~\ref{2b}.
\item Follows directly from Condition~\ref{2a}.
  \qed
\end{enumerate}
\end{proofNobox}
\noindent
The above is a variant of \reflem{2ST} that requires Condition~\ref{clause2ST} only
for $U=U'=\emptyset$, and allows to conclude that $N$ and $N'$ are interleaving branching
bisimilar (instead of branching ST-bisimilar) with explicit divergence.
Likewise, the below is a variant of \refthm{3ST} that requires Condition~\ref{lastST} only
for $U=U'=\emptyset$, and misses Condition~\ref{concurrent}.

\begin{thm}\label{thm-3}
Let $N=(S,T,F,M_0,\ell)$ be a net and $N'=(S',T',F',M'_0,\ell')$ be a plain
net with $S'\subseteq S$ and $M'_0=M_0\upharpoonright S'$.
Suppose there exist sets $T_+ \subseteq T$ and $T_-\subseteq T$
and a class $\NF\subseteq \Int^T$, such that
\begin{enumerate}[\em(1)-(4)]
\item Conditions~{\rm(\ref{acyclic+ST})--(\ref{normalformST})} from \refthm{3ST} hold, and
\item[\em(5)] For every reachable marking $M'\in [M'_0\rangle_{N'}$ there is an $H_{M'}\fin\nat^{T_+}$
  with $\ell(H_{M'})\equiv\emptyset$, such that for each $H\fin \NF$ with
  $M:=M'+(M_0-M'_0)+\marking{H}\in[M_0\rangle_N$ one has:
  \begin{enumerate}[\em(a)]
  \item $M_{M'}:=M'+(M_0-M'_0)+\marking{H_{M'}}\in\nat^S$,\label{marking}
  \item if $M'\goesto[a]$ with $a\in\Act$ then $M_{M'}\goesto[a]$,
    \item $H\leq H_{M'}$,\label{upperbound}
    \item if $H(u)<0$ then $u\in T_-$,\label{T-}
    \item if $H(u)<0$ and $H(t)>0$ then $\precond{u} \cap \precond{t} = \emptyset$,
          \label{disjoint preplaces}
    \item if $H(u)<0$ and $M[t\rangle$ with $\ell(t)\neq\tau$ then $\precond{u} \cap \precond{t} = \emptyset$.
          \label{disjoint preplaces 2}
  \end{enumerate}
\end{enumerate}
Then $N$ and $N'$ are interleaving branching bisimilar with explicit divergence.
\end{thm}

\begin{proof}
A straightforward simplification of the proof of \refthm{3ST}.
\end{proof}

\section{The Correctness Proof}\label{sec-correctness}

We now apply the preceding theory to prove the correctness of the conflict replicating
implementation.

\begin{thm}\label{thm-correctness}
Let $N'$ be a finitary plain structural conflict net without a
fully reachable pure $\structuralM$. Then $\impl{N'} \approx^\Delta_{bSTb} N'$.
\end{thm}

\begin{proofNobox}
Let $N'=(S',T',F',M'_0,\ell')$ be the given finitary plain structural conflict net without a fully reachable
pure $\structuralM$, and $N=(S,T,F,M_0,\ell)$ be its conflict replicated implementation $\impl{N'}$.
This convention (at the expense of primes in the statement of the theorem) pays off in terms of a
significant reduction in the number of primes in this paper.

For future reference, \reftab{conflictrepl} provides a place-oriented
representation of the conflict replicating implementation of a given net
$N'=(S',T',F',M'_0,\ell')$, with the macros for reversible transitions expanded.
\hypertarget{far}{Here $\mbox{\hyperlink{Tback}{$T^\leftarrow$}}=\{\ini \mid j\inp T'\} \cup
\{\trans{j} \mid h <^\# j\inp T'\}$, \plat{$(\trans{j})^{\,\it far}=\{\transout{j}\}$}
and \plat{$(\ini)^{\,\it far}=\{\Pre^j_k \mid k\geq^\# j\} \cup \{\transin{j}\mid h<^\# j\}$}.}

\[
\begin{array}{@{}l@{~}lll@{}}
~\\[-1.5ex]
\textbf{Place} &
\textrm{Pretransitions}\hfill\scriptstyle\rm{arc~weights} & \textrm{Posttransitions}\hfill\scriptstyle\rm{arc~weights} & \textrm{for all} \\
\hline
p                & \comp{j}\weight{i,p} & \dist ~~~\mbox{\scriptsize (if $\postcond{p}\mathbin{\neq}\emptyset$)}
                     & p\inp S',~ i\in \precond{p} \\
p_c            & \left\{
                          \begin{array}{@{}l@{}}\dist\\\ini[c]\mathop{\cdot}\undone\raisebox{1ex}{\weight{p,c}}\!\!\!\end{array}\right. &
                          \begin{array}{@{}l@{}}\ini[c]\cdot\fire~~~~~~~~~~~\qquad\weight{p,c}\\\fetch\weight{p,i}\end{array} &
                          \begin{array}{@{}l@{}}p\inp S',~c\in\postcond{p}\\j\geq^\# i \in \postcond{p}\end{array} \\
\pi_c ~\hfill\mbox{\scriptsize(marked)} & \ini[c]\cdot\reset & \ini[c]\cdot\fire & i\confeq c \in T'\\
\Pre^i_j           & \left\{
                          \begin{array}{@{}l@{}}\ini[i]\cdot\fire\\\exec{j}\end{array}\right. &
                          \begin{array}{@{}l@{}}\exec{j}\\\ini[i]\cdot\und[\Pre^i_j]\end{array} &
                          \begin{array}{@{}l@{}}j\geq^\# i\in T'\end{array} \\
\transin{j}          & \left\{
                          \begin{array}{@{}l@{}}\ini\cdot\fire\\\trans{j}\cdot\undone\end{array}\right. &
                          \begin{array}{@{}l@{}}\trans{j}\cdot\fire\\\ini\mathop\cdot\und[\transin{j}]\end{array} &
                                    h <^\# j\in T' \\
\transout{j}         & \left\{
                          \begin{array}{@{}l@{}}\trans{j}\cdot\fire\\\exec{j}\end{array}\right. &
                          \begin{array}{@{}l@{}}\exec{j}\\\trans{j}\mathord\cdot\und[\transout{j}]\!\!\!\end{array} &
                                          h<^\# j\in T',~ i\leqc j \\
\pi_{j\#l} ~\hfill\mbox{\scriptsize(marked)} &  \left\{
                          \begin{array}{@{}l@{}}\fetched{j}\\\trans[j]{l}\cdot\reset[c]\end{array}\right. &
                          \begin{array}{@{}l@{}}\exec{j}\\\trans[j]{l}\cdot\fire\end{array} &
                          \begin{array}{@{}l@{}}i\leqc j <^\# l \in T',~ c\confeq l \end{array} \\
\fetchin           & \exec{j} & \fetch & j\mathbin{\geq^\#} i\inp T',~p\inp\precond{i},~c\inp\postcond{p} \\
\fetchout          & \fetch & \fetched{j} & j\mathbin{\geq^\#} i\inp T',~p\inp\precond{i},~c\inp\postcond{p} \\
\hline
\rule{0pt}{12pt}
\undo(t)      & \exec[i]{j}\cdot\fire & t\cdot\undo,\quad t\cdot\elide & j\geq^\# i\in T',~ t\in \UIij \\
\reset(t)      & \fetched[i]{j}      & t\cdot\reset,\quad t\cdot\elide & j\geq^\# i\in T',~ t\in \UIij \\
\ack(t)      & t\cdot\reset,\quad t\cdot\elide & \comp[i]{j} & i\in T',~ t\in \UIij \\
\Fired(t)       & t\cdot\fire & t\cdot\undo & t\in T^\leftarrow,~ \UIij\ni t \\
\keep(t)    & t\cdot\undo & t\cdot\reset & t\in T^\leftarrow,~ \UIij\ni t \\
\take(f,t)  & t\cdot\undo & t\cdot\und  & t\in T^\leftarrow,~ \UIij\ni t,~ f\inp t^{\,\it far} \\
\took(f,t)  & t\cdot\und & t\cdot\undone  & t\in T^\leftarrow,~ f\in t^{\,\it far} \\
\rho(t)       & t\cdot\undone & t\cdot\reset & t\in T^\leftarrow,~ \UIij\ni t \\
\end{array}
\]
\begin{center}
\vspace{1ex}
\refstepcounter{table}
Table \thetable: The conflict replicating implementation.
\label{tab-conflictrepl}
\vspace{3ex}
\end{center}

We will obtain \refthm{correctness} as an application of \refthm{3ST}.
Following the construction of $N$ described in \refsec{implementation},
we indeed have $S'\subseteq S$ and $M'_0=M_0\upharpoonright S'$.
Let $T_+\subseteq T$ be the set of transitions
\begin{equation}\label{positive transitions}
  \dist      \qquad
  \ini\cdot\fire       \qquad
  \trans{j}\cdot\fire  \qquad
\end{equation}
for any applicable values of $p\inp S'$ and $h,j\inp T'\!$.
Furthermore, $T_-:=(T\setminus (T_+ \cup \{\exec{j}\mid i\leqc j \in T'\}))$.
We start with checking Conditions \ref{acyclic+ST}, \ref{acyclic-ST} and \ref{matchingST} of \refthm{3ST}.
\begin{enumerate}[1.]
\item[\ref{acyclic+ST}.] Let $<_+$ be the partial order on $T_+$ given by the order of
  listing in (\ref{positive transitions})---so $\ini[i]\cdot\fire <_+ \trans{j}\!\cdot\fire$, for any
  $i\in T'$ and $h<^\#j\in T'$, but
  the transitions $\trans{j}\cdot\fire$ and $\trans[k]{l}\cdot\fire$ for $(i,j)\neq(k,l)$ are unordered.
  By examining  \reftab{conflictrepl}
  we see that for any place with a pretransition $t$ in $T_+$, all its posttransitions $u$ in $T_+$
  appear higher in the $<_+$-ordering: $t<_+ u$. From this it follows that
  $F\upharpoonright(S\cup T_+)$ is acyclic.
\item[\ref{acyclic-ST}.] Let $<_-\!$ be the partial order on $T_-\!$ given by the column-wise order of
  the following enumeration of $T_-\!$:
\[\begin{array}{l}              t\cdot\undo \\
  \trans{j}\cdot\und  \\ \trans{j}\cdot\undone \\
  \ini\cdot\und        \\  \ini\cdot\undone  \\
\end{array}\qquad\qquad
\begin{array}{l}
  \fetch               \\  \fetched{j}         \\
  t\cdot\reset[i]  \\ t\cdot\elide[i] \\  \comp{j}
\end{array}\]
  for any $t\in\{\ini,~ \trans{j}\}$ and any applicable
  values of $f\inp S$, $p\inp S'$, and $h,i,j,c\inp T'\!$.
  By examining \reftab{conflictrepl}
  we see that for any place with a pretransition $t$ in $T_-$, all its posttransitions $u$ in $T_-$
  appear higher in the $<_-$-ordering: $t<_- u$. From this it follows that
  $F\upharpoonright(S\cup T_-)$ is acyclic.
\item[\ref{matchingST}.]
  The only transitions $t\in T$ with $\ell(t)\neq\tau$ are $\exec{j}$, with
  $i\leqc j\in T'$. So take $i\leqc j\in T'$. Then the only transition $t'\inp T'$
  with $\ell'(t')\mathbin=\ell(\exec{j})$ is $i$. Now two statements regarding $i$ and $\exec{j}$
  need to be proven. For the first, note that, for any $p\in\precond{i}$, the places
  $p$, $p_i$ and $\Pre^i_j$ are faithful w.r.t.\ $T_+$ and $S'\cup\{s\in S \mid M_0(s)>0\}$.
  Hence $~ p ~~ \dist ~~ p_i ~~ \ini[i]\cdot\fire ~~ \Pre^i_j ~~ \exec{j} ~$
  is a faithful path from $p$ to $\exec{j}$. The arc weight of this path is
  $F'(p,i)$. Thus $\precond{i} \leq \mbox{}^*\exec{j}$.

  The second statement holds because, for all $i\leqc j\in T'$,
  \begin{equation}\label{mimic}\qquad
  \marking{i} = \llbracket
     \exec{j} + \!\!\sum_{p\in\precond{i}}\big(F'(p,i)\cdot\dist +
     \!\!\sum_{c\in\postcond{p}}\fetch\big) + \fetched{j} + \comp{j} +
     \sum_{t\in\UIij} t\cdot\elide
     \rrbracket.
  \end{equation}
  To check that these equations hold, note that
  $$\qquad\quad\begin{array}{l@{~}c@{~}l@{}}
    \marking{\dist}&=&-\{p\} + \{p_c \mid c\in \postcond{p}\},\\
    \marking{\exec{j}}&=&-\{\pi_{j\#l}\mid l\geq^\# j\}+\{\fetchin \mid p\inp\precond{i},~
    c\inp\postcond{p}\}+\{\undo(t) \mid t\in\UIij\},\\
    \marking{\fetch}&=&-\{\fetchin\} - F'(p,i)\cdot\{p_c\}+\{\fetchout\},\\
    \marking{\fetched{j}}&=&-\{\fetchout\mid p\inp\precond{i},~ c\inp\postcond{p}\}
    +\{\pi_{j\#l}\mid l\geq^\# j\}+\{\reset(t)\mid t\inp \UIij\},\\
    \marking{t\cdot\elide}&=&-\{\undo(t),~\reset(t)\mid t\in \UIij\}+\{\ack(t) \mid t\in\UIij\},\\
    \marking{\comp{j}} &=&-\{\ack(t) \mid t\in\UIij\}+
                           \plat{$\displaystyle\sum_{r\in\postcond{i}}F'(i,r)\cdot\{r\}$}.\\[1ex]
  \end{array}$$
\end{enumerate}
Before we define the class $\NF\subseteq \Int^T$ of signed multisets of transitions in
normal form, and verify conditions \ref{normalformST} and \ref{lastST}, we derive
some properties of the conflict replicating implementation $N=\impl{N'}$.

\begin{clm}\label{cl-G-properties}
  For any $M'\in\Int^{S'}$ and $G\fin \Int^T$ such that
  $M:=M'+(M_0-M'_0)+\marking{G}\in\nat^S$
  and for each $i\in T'$ and $t\inp \UIij$ we have
  \begin{eqnarray}
  G(t\cdot\elide)+G(t\cdot\undo) &\!\!\!\!\leq\!\!\!\!& \sum_{j\geq^\# i}G(\exec[i]{j})\label{undo}\\
  \hspace{-2em}
  G(\comp[i]{j}) \leq G(t\cdot\elide)+G(t\cdot\reset) &\!\!\!\!\leq\!\!\!\!& \sum_{j\geq^\# i}G(\fetched[i]{j})\label{reset}\\
  G(t\cdot\reset) &\!\!\!\!\leq\!\!\!\!& G(t\cdot\undo)\label{interfacecount}.
  \end{eqnarray}
  Moreover, for each \hyperlink{far}{$t\in T^\leftarrow$ and $f\in t^{\,\it far}$},
  \begin{equation}\label{undocount}
  \sum_{\{\omega\mid t\in \UI_\omega\}}\!\!\!\!\!\!\!\! G(t\cdot\reset[\omega])
  \leq G(t\cdot\undone) \leq G(t\cdot\und)
  \leq \!\!\!\!\!\!\!\!\sum_{\{\omega\mid t\in \UI_\omega\}}\!\!\!\!\!\!\!\! G(t\cdot\undo[\omega]) \leq G(t\cdot\fire)
  \end{equation}
  and for each appropriate $c,h,i,j,l\in T'$ and $p\in S'$:
  \begin{eqnarray}
  G(\fetched[i]{j}) \leq G(\fetch) &\!\!\!\!\leq\!\!\!\!& G(\exec{j})\label{fetch}
  \\[2pt]\label{pi-j}
  G(\ini\cdot\fire) &\!\!\!\!\leq\!\!\!\!& 1+\sum_{\omega }G(\ini\cdot\reset[\omega])
  \\[-6pt]\label{transin}\hspace*{-2em}
  G(\trans{j}\mathord\cdot\fire) \mathord- G(\trans{j}\mathord\cdot\undone) &\!\!\!\!\leq\!\!\!\!&
  G(\ini\mathord\cdot\fire) \mathord- G(\ini\mathord\cdot\und[\transin{j})]\hspace*{2em}
  \end{eqnarray}
  \begin{eqnarray}\label{pi}\hspace*{-2em}
  G(\trans[j]{l}\!\cdot\fire) + \sum_{i\leqc j}G(\exec{j}) &\!\!\!\!\leq\!\!\!\!&
  1+\sum_{\omega}G(\trans[j]{l}\mathord\cdot\reset[\omega])\mathop+ \!\sum_{i\leqc j}G(\fetched{j})
  \\[-5pt]\label{pre}
  \mbox{if ~$M[\exec{j}\rangle$~ then}\quad
  1 &\!\!\!\!\leq\!\!\!\!& G(\ini[i]\mathord\cdot\fire) \mathord- G(\ini[i]\mathord\cdot\und[\Pre^i_j])
  \\[2pt]\label{transout}
  \mbox{if ~$\exists i.~M[\exec{j}\rangle$~ then}\quad
  1 &\!\!\!\!\leq\!\!\!\!& G(\trans{j}\mathord\cdot\fire) \mathord- G(\trans{j}\mathord\cdot\und[\transout{j}]\!)~~~~~
  \end{eqnarray}
    \begin{equation}\label{p_j}
      F'(p,c)\mathord\cdot \big(G(\ini[c]\!\cdot\fire) \mathord- G(\ini[c]\!\cdot\undone)\big)
      + \plat{$\displaystyle\sum_{j\geq^\# i\in \postcond{p}}$}
      F'(p,i) \cdot G(\fetch) \leq G(\dist)
    \end{equation}
  \begin{eqnarray}\label{p}\hspace{-1.67em}
  G(\dist) &\!\!\!\!\leq\!\!\!\!&
  M'(p)+ \hspace{-1em}\sum_{\{i\in T'\mid p\in\postcond{i}\}}\hspace{-1em} G(\comp{j}).
  \end{eqnarray}
\end{clm}

\begin{proofclaim}
  For any $i\in T'$ and $t\in\UIij$, we have
  $$M(\undo(t))=\big(\sum_{j\geq^\# i}G(\exec[i]{j})\big)-G(t\cdot\elide)-G(t\cdot\undo)\geq 0,$$ given that
  $M'(\undo(t))=(M_0-M'_0)(\undo(t))=\emptyset$.
  In this way, the place $\undo(t)$ gives rise to the inequation (\ref{undo}) about $G$.
  Likewise, the places $\ack(t)$, $\reset(t)$ and $\keep(t)$,
  respectively, contribute (\ref{reset}) and (\ref{interfacecount}), whereas
  $\rho(t)$, $\took(t)$, $\take(t)$ and $\Fired(t)$ yield (\ref{undocount}).
  The remaining inequations arise from $\fetchout$, $\fetchin$, $\pi_j$,
  $\transin{j}$, $\pi_{j\#l}$, $\Pre^i_j$, $\transout{j}$, $p_c$ and $p$, respectively.
\end{proofclaim}
  (\ref{pi}) can be rewritten as $T^j_l+\sum_{i\leqc j} E^i_j \leq 1$, where
  $T^j_l:=G(\trans[j]{l}\cdot\fire) - \sum_{\omega}G(\trans[j]{l}\cdot\reset[\omega])$ and
  $E^i_j:=G(\exec[i]{j}) - G(\fetched[i]{j})$.
  By (\ref{undocount})
  $\sum_\omega G(\trans[j]{l}\cdot\reset)\leq G(\trans[j]{l}\cdot\fire)$, so $T^j_l\geq 0$,
  and likewise, by (\ref{fetch}), $E^i_j\geq 0$ for all $i\leqc j$.
  Hence, for all $i\leqc j <^\# l\in T'$,\vspace{-1ex}
  \begin{equation}\label{pi-execute}
  0\leq T^j_l\leq 1 \qquad 0\leq E^i_j\leq 1\qquad T^j_l+\sum_{i\leqc j} E^i_j \leq 1.
  \vspace{-2ex}
  \end{equation}

\newcommand{\follow}{\textit{next}}
\noindent
In our next claim we study triples $(M, M', G)$ with
\begin{enumerate}[(A)]
\item $M\in[M_0\rangle_N$, $M'\in[M'_0\rangle_{N'}$ and $G \fin \Int^T$,\label{r0}
\item $M=M'+(M_0-M'_0)+\marking{G}$,\label{r1}
\item $G(\comp{j}) = 0$ for all $i\in T'$,\label{r2}
\item $G(\dist) \leq M'(p)$ for all $p\in S'$,\label{r3}
\item $G(\fetched[k]{l})\geq 0$ for all $k\leqc l\in T'$,\label{rFp}
\item \plat{$\displaystyle G(\dist) \geq F'(p,i)\cdot G(\exec{j})$}
  for all $i\leqc j\in T'$ and $p\in\precond{i}$,\label{r4}
\item \plat{$0\leq G(\exec{j})\leq 1$} for all $i\leqc j\in T'$,\label{rExec}
\item \plat{$\displaystyle G(\dist) \geq F'(p,j)\cdot G(\exec{j})$}
  for all $i\leqc j\in T'$ and $p\in\precond{j}$,\label{r5}
\item (in the notation of (\ref{pi-execute}))
  if $E^i_j=1$ with $i\leqc j\in T'$ then \plat{$T^h_j=1$} for all $h <^\# j$,\label{rtrans}
\item there are no $j\geq^\# i \confeq k \leq^\# l \inp T'$ with $(i,j)\mathbin{\neq} (k,\ell)$,
  \plat{$G(\exec[i]{j})\mathbin>0$} and \plat{$G(\exec[k]{l})\mathbin>0$},\label{rExec2}
\item there are no $i\leq^\# j \confeq k \leq^\# l \inp T'$ with $(i,j)\mathbin{\neq} (k,\ell)$,
  \plat{$G(\exec[i]{j})\mathbin>0$} and \plat{$G(\exec[k]{l})\mathbin>0$}.\label{rExec3}
\label{rLast}
\end{enumerate}
Given such a triple $(M_1,M'_1,G_1)$ and a transition $t\in T$,
we define $\follow(M_1, M'_1, G_1, t) =: (M, M', G)$ as follows:
Let $G_2:=G_1+\{t\}$.
Take $M:=M_1+\marking{t} = M'_1+(M_0-M'_0)+\marking{G_2}$.
In case $t$ is not of the form $\comp{j}$ we take $M':=M'_1\in[M'_0\rangle_{N'}$ and $G:=G_2\fin\Int^T$.
In case $t\mathbin=\comp{j}$ for some $i\in T'$ then
\plat{$1=G_2(\comp{j}) \leq \sum_{j\geq^\# i} G_2(\exec{j})
  =\sum_{j\geq^\# i} G_1(\exec{j})$} by (\ref{r2}), (\ref{reset}) and (\ref{fetch}),
so by (\ref{rExec}) and (\ref{rExec2}) there is a unique $j\geq^\# i$ with
\plat{$G_1(\exec{j})=1$}. We take $M':=M'_1+\marking{i}$ and
\plat{$G:=G_2-G^i_{\!\!j}$}, where $G^i_{\!\!j}$ is the right-hand side of (\ref{mimic}).

\begin{clm}\label{cl-reachable}~
\begin{enumerate}[(1)]
\item If $M_1[t\rangle$ and $(M_1, M'_1, G_1)$ satisfies (\ref{r0})-(\ref{rExec3}),
  then so does $\follow(M_1, M'_1, G_1, t)$.\label{follow1}
\item For any $M\in[M_0\rangle_{N}$ there exist $M'$ and $G$ such that
  (\ref{r0})-(\ref{rExec3}) hold.\label{follow2}
\end{enumerate}
\end{clm}

\begin{proofclaimNobox}
(\ref{follow2}) follows from (\ref{follow1}) via induction on the reachability of $M$.
In case $M=M_0$ we take $M':=M'_0$ and $G:=\emptyset$. Clearly,
(\ref{r0})--(\ref{rLast}) are satisfied.

Hence we now show (\ref{follow1}). Let $(M, M', G) := \follow(M_1, M'_1, G_1, t)$.
We check that $(M,M',G)$ satisfies the requirements (\ref{r0})--(\ref{rLast}).
\begin{enumerate}[(A)]
\item By construction, $M\in[M_0\rangle_N$ and $G \fin \Int^T$.
  If $t$ is not of the form $\comp{j}$ we have \mbox{$M'\mathbin=M_1\inp[M'_0\rangle_{N'}$}.
  Otherwise, by (\ref{r3}) and (\ref{r4}) we have $M'_1(p)\geq G_1(\dist)\geq F'(p,i)$
  for all $p\inp \precond{i}$, and hence $M'_1[i\rangle$. 
  This in turn implies that $M'=M'_1+\marking{i}\in[M'_0\rangle_{N'}$.
\item  In case $t$ is not of the form $\comp{j}$ we have
  $$M=M_1+\marking{t}=M'_1+(M_0-M'_0)+\marking{G_1+t}= M'+(M_0-M'_0)+\marking{G}.$$
  In case $t=\comp{j}$ we have $M=M'_1+(M_0-M'_0)+\marking{G_2}=
  M'+(M_0-M'_0)+\marking{G}$, using that $\marking{i}=\marking{G^i_{\!\!j}}$.
\item In case $t=\comp{j}$ we have $G(\comp{j}) = G_1(\comp{j})+1-G^i_{\!\!j}(\comp{j}) =
  0+1-1=0$.\\ Otherwise $G(\comp{j}) = G_1(\comp{j})+0=0+0=0$.
\item This follows immediately from (\ref{r2}) and (\ref{p}).
\item  
  The only time that this invariant is in danger is when $t=\comp{j}$.
  Then $G=G_1+\{\comp{j}\}-G^i_{\!\!j}$ for a certain $j\geq^\# i$ with $G_1(\exec{j})=1$.
  By (\ref{rExec2})\footnote{We use (\ref{rExec2}) and (\ref{rFp}) for $G_1$ only, making
  use of the induction hypothesis.}\let\fnote\thefootnote\
  $G_1(\exec{l})\leq 0$ for all $l\geq^\# i$ with $l\neq j$.
  Hence by (\ref{fetch}) $G_1(\fetched{l})\leq 0$ for all such $l$.
  By (\ref{r2}) $G_2(\comp{j}) = G_1(\comp{j})+1 = 1$, so by (\ref{reset})
  $\sum_{l\geq^\# i} G_1(\fetched{l}) \mathbin= \sum_{l\geq^\# i} G_2(\fetched{l})\linebreak[2]>0$;
  hence it must be that $G_1(\fetched{j})\mathord>0$. By (\ref{rFp})$^\fnote$ $G_1(\fetched[k]{l})\mathbin{\geq} 0$
  for all $k\leqc l\inp T'$.
  Given that $G^i_{\!\!j}(\fetched{j})=1$ and $G^i_{\!\!j}(\fetched[k]{l})=0$ for all $(k,l)\neq(i,j)$,
  we obtain $G(\fetched[k]{l})\geq 0$ for all $k\leqc l\in T'$.
\item Take $i\mathbin{\leqc} j \inp T'$ and $p\inp\precond{i}$.
  There are two occasions where the invariant is in danger: when $t=\exec{j}$
  and when $t=\comp[k]{l}$ with $k\in T'$.  First let $t=\exec{j}$.
  Then $M_1[\exec{j}\rangle$.  Thus,\vspace{-1ex}
  \hypertarget{proofr4}{$$\quad\begin{array}[b]{@{}r@{~\geq~}l@{}}
  \multicolumn{2}{@{}l}{G(\dist)}\\
  \mbox{}& \displaystyle
    F'(p,i)\cdot\big(G(\ini[i]\cdot\fire) - G(\ini[i]\cdot\undone)\big)
    + \hspace{-.5em}\sum_{h\geq^\# g\in\postcond{p}}\hspace{-.5em}F'(p,g)\cdot G(\txf{fetch}_{g,h}^{p,i})
\\
  & \displaystyle
    F'(p,i)\cdot\big(G(\ini[i]\cdot\fire) - G(\ini[i]\cdot\undone)\big)
    + \hspace{-.5em}\sum_{h\geq^\# g\in\postcond{p}}\hspace{-.5em}F'(p,g)\cdot G(\fetched[g]{h})
\\
  & F'(p,i)\cdot\big(G(\ini[i]\cdot\fire) - G(\ini[i]\cdot\undone)\big) + F'(p,i)\cdot G(\fetched{j})
\\
  & F'(p,i)\cdot\left(\big(G(\ini[i]\cdot\fire) - G(\ini[i]\cdot\und[\Pre^i_j])\big) + G(\fetched{j})\right)
\\
  & F'(p,i)\cdot\big(1 + G(\fetched{j})\big)
\\
  & F'(p,i)\cdot G(\exec{j})
  \end{array}$$}
  by (\ref{p_j}), (\ref{fetch}), (\ref{rFp}), (\ref{undocount}), (\ref{pre}) and (\ref{pi-execute}), respectively.
  By (\ref{undocount}) $G(\ini[i]\cdot\fire) - G(\ini[i]\cdot\undone)\geq 0$.
  So by (\ref{p_j}), (\ref{rFp}), and (\ref{fetch}) $G(\dist)\geq 0$.
  For this reason we may assume, w.l.o.g., that $G(\exec{j}) \geq 1$.

  We have $G=G_1+\{\comp[k]{l}\}-G^k_l$ for certain $l\geq^\# k$ with $G_1(\exec[k]{l})\mathbin=1$.
  Since \plat{$G^i_{\!\!j}(\exec{j})\mathbin\geq 0$}, we also have $G_1(\exec{j}) \geq 1$.
  By (\ref{rExec2}) this implies that $\neg(i\confeq k)$ or $(i,j)=(k,l)$.
  In the latter case \plat{$G(\exec{j})=\mbox{}$}
  \plat{$G_1(\exec{j})-G^i_{\!\!j}(\exec{j})=1\mathord-1=0$},
  contradicting our assumption.
  In the former case $p\notin\precond{k}$, so $G^k_l(\dist)=0$ and hence
  $G(\dist)=G_1(\dist)\geq F'(p,i)\cdot G_1(\exec{j})=F'(p,i)\cdot G(\exec{j})$.
\item That $G(\exec{j})\geq 0$ follows from (\ref{rFp}) and (\ref{fetch}).
  If $G(\exec{j})\geq 2$ for some $i\leqc j\in T'$ then
  $M'(p)\geq G(\dist) \geq 2\cdot F'(p,i)$ for all $p\in \precond{i}$, using (\ref{r3}) and
  (\ref{r4}), so $M'[2\cdot\{i\}\rangle_{N'}$. Since $N'$ is
  a \hyperlink{finitary}{finitary} \hyperlink{scn}{structural conflict net}, it has no
  self-concurrency, so this is impossible.
\item Take $i\mathbin{\leqc} j \inp T'$ and $p\inp\precond{j}$. The case $i=j$ follows
  from (\ref{r4}), so assume $i<^\# j$.
  By (\ref{undocount}) we have $G(\ini[i]\cdot\fire) - G(\ini[i]\cdot\undone)\geq 0$.
  So by (\ref{p_j}), (\ref{rFp}), and (\ref{fetch}) $G(\dist)\geq 0$.
  Hence, using (\ref{rExec}), we may assume, w.l.o.g., that \plat{$G(\exec{j})=1$}.
  We need to investigate the same two cases as in the proof of (\ref{r4}) above.
  First let $t=\exec{j}$.  Then $M_1[\exec{j}\rangle$.  Thus,
  \hypertarget{proofr5}{$$\qquad\begin{array}[b]{@{}r@{~\geq~}lr@{}}
  \multicolumn{2}{@{}l}{G(\dist)}  & \mbox{(by (\ref{p_j}))} \\
  \mbox{}& \displaystyle
    F'(p,j)\cdot \big(G(\ini\cdot\fire) - G(\ini\cdot\undone)\big)
    + \hspace{-.5em}\sum_{h\geq^\# g\in\postcond{p}}\hspace{-.5em}F'(p,g)\cdot G(\txf{fetch}_{g,h}^{p,j})
\hspace{-4.5em}\\[-10pt]
  & F'(p,j)\cdot \big(G(\ini\cdot\fire) - G(\ini\cdot\undone)\big)
  & \hspace{-4em}\mbox{(by (\ref{rFp}) and (\ref{fetch}))} \\
  & F'(p,j)\cdot \big(G(\ini\cdot\fire) - G(\ini\cdot\und[\mbox{$\transin[i]{j}$}])\big)
  & \mbox{(by (\ref{undocount}))} \\
  & F'(p,j)\cdot \big(G(\trans[i]{j}\cdot\fire) - G(\trans[i]{j}\cdot\undone)
  & \mbox{(by (\ref{transin}))} \\
  & F'(p,j)\cdot \big(G(\trans[i]{j}\cdot\fire) - G(\trans[i]{j}\cdot\und[\mbox{$\transout[i]{j}$}])\big)
  & \mbox{(by (\ref{undocount}))} \\
  & F'(p,j)
  & \mbox{(by (\ref{transout}))}\makebox[0pt]{\,.}
  \end{array}$$}
  Now let $t=\comp[k]{l}$ with $k\in T'$. 
  We have $G=G_1+\{\comp[k]{l}\}-G^k_l$ for certain $l\geq^\# k$ with $G_1(\exec[k]{l})=1$.
  Since $G^i_{\!\!j}(\exec{j})\mathbin\geq 0$, we also have $G_1(\exec{j}) \geq 1$.
  By (\ref{rExec3}) this implies that $\neg(j\confeq k)$ or $(i,j)=(k,l)$.
  In the latter case \plat{$G(\exec{j})=\mbox{}$}
  $G_1(\exec{j})-G^i_{\!\!j}(\exec{j})=1-1=0$, contradicting our assumption.
  In the former case $p\notin\precond{k}$, so \plat{$G^k_l(\dist)=0$} and hence
  $G(\dist)=G_1(\dist)\geq F'(p,j)\cdot G_1(\exec{j})=F'(p,j)\cdot G(\exec{j})$.

\item Let $i\mathbin{\leqc} j \inp T'$ and $h<^\# j$.
  Since, for all $k\mathbin{\leqc} l\inp T'$,
  $G^k_l(\trans{j}\!\cdot\fire)\mathbin=\sum_\omega G^k_l(\trans{j}\!\cdot\reset[\omega])\mathbin=0$ and
  \plat{$G^k_l(\exec{j})=G^k_l(\fetched{j})$},
  the invariant is preserved when $t$ has the form \plat{$\comp[b]{c}\!$}. Using
  (\ref{pi-execute}), it is in danger only when $t=\exec{j}$ or $t=\trans{j}\!\cdot\reset[\omega]$
  for some $\omega$ with \plat{$\trans{j}\inp\UI_\omega$}.

  First assume $M_1[\exec{j}\rangle$ and $T^h_j=G_1(\trans{j}\cdot\fire)-\sum_\omega G_1(\trans{j}\cdot\reset[\omega])=0$.
  Then
  $$\begin{array}[b]{r@{~\leq~}ll}
  \multicolumn{1}{r@{~\leq~}}{1}
  & G_1(\trans{j}\cdot\fire) - G_1(\trans{j}\cdot\und[\mbox{$\transout{j}$}])
  & \mbox{(by (\ref{transout}))} \\
  & G_1(\trans{j}\cdot\fire) - \sum_{\omega} G_1(\trans{j}\cdot\reset[\omega]) =0
  & \mbox{(by (\ref{undocount}))}, \\
  \end{array}$$
  which is a contradiction.

  Next assume \plat{$t=\trans{j}\!\cdot\reset[k]$} with $k \confeq j$, and $E^i_j=1$.
  By (\ref{rFp}) and (\ref{rExec}) the latter implies that \plat{$G_1(\exec{j})=1$} and \plat{$G_1(\fetched{j})=0$}\textsl{}.
  Then  
  $$\begin{array}[b]{r@{~\leq~}ll}
  \multicolumn{1}{r@{~=~}}{0}
  & G_1(\comp[k]{l})
  & \mbox{(by (\ref{r2}))} \\
  & G_1(\trans{j}\cdot\elide[k])+G_1(\trans{j}\cdot\reset[k])
  & \mbox{(by (\ref{reset}))} \\
  \multicolumn{1}{r@{~<~}}{} & G(\trans{j}\cdot\elide[k])+G(\trans{j}\cdot\reset[k]) \\
  & \sum_{l\geq^\#k}G(\fetched[k]{l})
  & \mbox{(by (\ref{reset}))}.
  \end{array}$$
  Hence $G_1(\fetched[k]{l})\mathbin=G(\fetched[k]{l})\mathbin>0$ for some $l\geq^\#k$, and by
  (\ref{fetch}) also $G_1(\exec[k]{l})\mathbin>0$.
  Using (\ref{rExec3}) we obtain $(i,\!j)\mathop=(k,l)$, thereby obtaining a contradiction\vspace{-2pt}
  ({$0\mathbin=G_1(\fetched{j})\mathbin=G_1(\fetched[k]{l})\mathbin>0$}).

\item
Let $j\geq^\# i \confeq k \leq^\# l \in T'$ with $(i,j)\neq (k,\ell)$.
  The invariant is in danger only when \plat{$t\mathbin=\exec{j}$} or $t\mathbin=\exec[k]{l}$.
  W.l.o.g.\ let $t\mathbin=\exec[k]{l}$, with
  $G_1(\exec[k]{l})\mathbin=0$ and $G_1(\exec{j})\mathbin\geq 1$.

  Making a case distinction, first assume \plat{$G(\fetched{j})\mathbin\geq 1$}.
  Using (\ref{r3}), (\ref{r4}) and that $G(\exec[k]{l})=1$, $M'(p) \geq G(\dist)\geq
  F'(p,k)$ for all $p\in\precond{k}$.   Likewise, $M'(p) \geq G(\dist)\geq F'(p,i)$ for all $p\in\precond{i}$.
  Moreover, just as in \hyperlink{proofr4}{the proof of} (\ref{r4}), we derive, for all
  $p\in\precond{i}\cap\precond{k}$,
  $$\quad\begin{array}[b]{@{}r@{~\geq~}l@{}}
  \multicolumn{2}{@{}l}{M'(p)\geq G(\dist)} \\
  \mbox{}& \displaystyle
    F'(p,k)\cdot\big(G(\ini[k]\cdot\fire) - G(\ini[k]\cdot\undone)\big)
    + \hspace{-.7em}\sum_{h\geq^\# g\in\postcond{p}}\hspace{-.5em}F'(p,g)\cdot G(\txf{fetch}_{g,h}^{p,k})
\\
  & \displaystyle
    F'(p,k)\cdot\big(G(\ini[k]\cdot\fire) - G(\ini[k]\cdot\undone)\big)
    + \hspace{-.7em}\sum_{h\geq^\# g\in\postcond{p}}\hspace{-.5em}F'(p,g)\cdot G(\fetched[g]{h})
\\
  & F'(p,k)\cdot\big(G(\ini[k]\cdot\fire) - G(\ini[k]\cdot\undone)\big) + F'(p,i)\cdot G(\fetched{j})
\\
  & F'(p,k)\cdot\big(G(\ini[k]\cdot\fire) - G(\ini[k]\cdot\und[\Pre^k_l])\big) + F'(p,i)\cdot G(\fetched{j})
\\
  & F'(p,k) + F'(p,i)
  \end{array}$$
  by (\ref{r3}), (\ref{p_j}), (\ref{fetch}), (\ref{rFp}), (\ref{undocount}) and (\ref{pre}), respectively.
  It follows that $M'[\{k\}\mathord+\{i\}\rangle$. As $i\confeq k$ and $N'$ is a \hyperlink{finitary}{finitary}
  \hyperlink{scn}{structural conflict net}, this is impossible. (Note that this argument
  holds regardless whether $i=k$.)

  Now assume \plat{$G(\fetched{j})\mathord\leq 0$}.
  Then, in the notation of (\ref{pi-execute}), \plat{$E^i_j\mathord=1$}.
  As $G_1(\exec[k]{l})\linebreak[2]=0$, (\ref{rFp}) and (\ref{fetch}) yield $G_1(\fetched[k]{l})=0$.
  Hence $G(\exec[k]{l})=1$ and $G(\fetched[k]{l})= 0$, so $E^k_l=1$.
  We will conclude the proof by deriving a contradiction from \plat{$E^i_j=E^k_l=1$}.
  In case $j=l$ this contradiction emerges immediately from (\ref{pi-execute}).
  By symmetry it hence suffices to consider the case $j< l$.

  By (\ref{r3}) and (\ref{r5}) we have $M'(p)\geq G(\dist)\geq F'(p,j)$ for all
  $p\in\precond{j}$, so $M'[j\rangle$. Likewise $M'[l\rangle$ and, using (\ref{r4}),
  $M'[i\rangle$ and $M'[k\rangle$. Since $j \confeq i \confeq k$ and
  $N'$ has no \hyperlink{M}{fully reachable \visible pure \structuralM},
  $j \confeq k$. Since $j \confeq k \confeq l$ and
  $N'$ has no \hyperlink{M}{fully reachable \visible pure \structuralM},
  $j \confeq l$. So $j<^\# l$.
  By (\ref{pi-execute}), using that \plat{$E^i_j=1$}, \plat{$T^j_l=0$}. This is in contradiction with
  $E^k_l=1$ and (\ref{rtrans}).

\item Suppose that $G(\exec{j})>0$ and $G(\exec[k]{l})>0$, with $i\leq^\# j \confeq k \leq^\# l \in T'$.
  By (\ref{r3}) and (\ref{r5}) we have $M'(p)\mathbin\geq G(\dist)\mathbin\geq F'(p,j)$ for all
  $p\inp\precond{j}$, so $M'[j\rangle$. Likewise, using (\ref{r4}),
  $M'[i\rangle$ and $M'[k\rangle$. Since $i\confeq j \confeq k$ and
  $N'$ has no \hyperlink{M}{fully reachable \visible pure \structuralM},
  $i \confeq k$. Using this, the result follows from (\ref{rExec2}).
  \hfill\filledbox
\end{enumerate}
\end{proofclaimNobox}

\begin{clm}\label{cl-extra}
For any $M\in[M_0\rangle_N$ there exist $M'\in[M'_0\rangle_{N'}$ and
$G\fin \Int^T$ satisfying (\ref{r0})--(\ref{rLast}) from \refcl{reachable}, and
\begin{enumerate}[(A)]
\setcounter{enumi}{11}
\item there are no $j\geq^\# i \confeq k \leq^\# l \in T'$ with
  \plat{$M[\exec[i]{j}\rangle$} and \plat{$G(\exec[k]{l})>0$},\label{rExec4}
\item there are no $i\leq^\# j \confeq k \leq^\# l \in T'$ with
  \plat{$M[\exec[i]{j}\rangle$} and \plat{$G(\exec[k]{l})>0$},\label{rExec5}
\item if \plat{$M[\exec[i]{j}\rangle$} for $i\leqc j \in T'$ then $M'[j\rangle$.\label{rVeryLast}
\end{enumerate}
\end{clm}
\begin{proofclaimNobox}
Given $M$, by \refcl{reachable}(2) there are $M'$ and $G$ so that
the triple $(M,M',G)$ satisfies (\ref{r0})--(\ref{rLast}).
Assume \plat{$M[\exec[i]{j}\rangle$} for some $i\leqc j \in T'$.
Let $M_1:=M+\marking{\exec[i]{j}}$ and $G_1:=G+\{\exec{j}\}$.
By (\ref{rExec}) $G(\exec{j})\geq 0$, so $G_1(\exec{j})>0$.
By \refcl{reachable}(1) the triple ($M_1,M',G_1$) satisfies (\ref{r0})--(\ref{rLast}).
\begin{enumerate}[(A)]
\setcounter{enumi}{11}
\item
  Suppose \plat{$G(\exec[k]{l})>0$} for certain $l\geq^\# k \confeq i$.\
  In case $(i,j)=(k,\ell)$, $G_1(\exec{j})\geq 2$, contradicting (\ref{rExec}).
  In case $(i,j)\neq (k,\ell)$, $G_1$ fails (\ref{rExec2}), also a contradiction.
\item
  Suppose \plat{$G(\exec[k]{l})>0$} for certain $l\geq^\# k \confeq j$.
  Then $G_1$ fails (\ref{rExec}) or (\ref{rExec3}), a contradiction.
\item
  By (\ref{r3}) and (\ref{r5}) $M'(p)\geq G_1(\dist)\geq F(p,j)$ for all $p\in\precond{j}$, so $M'[j\rangle$.
\hfill\filledbox
\end{enumerate}
\end{proofclaimNobox}

\begin{clm}\label{cl-concurrency}
If \plat{$M[\{\exec[i]{j}\}\mathord+\{\exec[k]{l}\}\rangle$} for some $M\in[M_0\rangle_N$
then $\neg (i \confeq k)$.
\end{clm}
\begin{proofclaim}
Suppose \plat{$M[\{\exec[i]{j}\}\mathord+\{\exec[k]{l}\}\rangle$} for some $M\in[M_0\rangle_N$.
By \refcl{reachable}(2) there exist $M'\in[M'_0\rangle_{N'}$ and
$G\fin \Int^T$ satisfying (\ref{r0})--(\ref{rLast}).
Let $M_1:=M+\marking{\exec[k]{l}}$ and \plat{$G_1:=G\mathbin+\{\exec[k]{l}\}$}.
By \refcl{reachable}(1) the triple $(M_1,M',G_1)$ satisfies (\ref{r0})--(\ref{rLast}).
Let $M_2:=M_1+\marking{\exec{j}}$ and \plat{$G_2:=G_1\mathbin+\{\exec{j}\}$}.
Again by \refcl{reachable}(1), the triple $(M_2,M',G_2)$ also satisfies (\ref{r0})--(\ref{rLast}).
As (\ref{rExec}) implies \plat{$G(\exec{j})\mathbin\geq 0$}, in case $(i,j)\mathbin=(k,l)$ we obtain
\plat{$G_2(\exec{j})\mathbin\geq 2$}, contradicting
(\ref{rExec}). Hence $(i,j)\mathbin{\neq}(k,l)$. Moreover, $G_2(\exec[k]{l})>0$ and $G_2(\exec{j})>0$.
Now (\ref{rExec2}) implies $\neg (i \confeq k)$.
\end{proofclaim}
\noindent
For any $t\in\{\ini,~ \trans{j}\}$ with
$h,j\inp T'$, and any \mbox{$\omega\inp \UI$} with $t\in\UI_\omega$, we write
$$t(\omega) := t\cdot\fire + t\cdot\undo[\omega] +
               \big(\sum_{f\in t^{\,\it far}} t\cdot\und\big) +
               t\cdot\undone+t\cdot\reset[\omega]\;.\vspace{-1ex}$$
The transition $t$ has no preplaces of type {\scriptsize \it in}, nor postplaces of type {\scriptsize \it out}.
By checking in \reftab{reversible} or \reffig{reversible} that each other place
occurs as often in $\precond{u(\omega)}+\postcond{(u\cdot\elide[\omega])}$ as in
$\postcond{u(\omega)}+\precond{(u\cdot\elide[\omega])}$, one verifies, for any $\omega\in
\UI$ with $t\in \UI_\omega$, that\vspace{-1ex}
\begin{equation}\label{elide}
\marking{t(\omega)} = \marking{t\cdot\elide[\omega]}.
\end{equation}
Let $\equiv$ be the congruence relation on finite signed multisets of transitions
generated by
\begin{eqnarray}\label{elideNF}
t(\omega) &\equiv& t\cdot\elide[\omega]
\end{eqnarray}
for all $t \in \{\ini,~ \trans{j} \mid h,j\inp T'\}$ and $\omega\in \UI$ with
$\UI_\omega\ni t$.
Here \emph{congruence} means that $G_1\mathbin\equiv G_2$ implies $k\cdot G_1 \mathbin\equiv k\cdot G_2$ and
$G_1+H \mathbin\equiv G_2+H$ for all $k\inp\Int$ and $H\fin \Int^T$.
Using (\ref{elide}) $G_1\equiv G_2$ implies $\marking{G_1}=\marking{G_2}$.

\begin{clm}\label{cl-0}
If $M'=\marking{G}$ for $M'\in \Int^{S'}$ and $G\fin \Int^T$ such that for all $i\in T'$
we have $G(\comp{j})\mathbin=0$ and either $\forall j\geq^\# i.~G(\exec{j})\mathbin\geq 0$ or $\forall
j\geq^\# i.~G(\exec{j})\mathbin\leq 0$, then $G \mathbin\equiv \emptyset$.
\end{clm}

\begin{proofclaim}
  Let $M'$ and $G$ be as above. 
  W.l.o.g.\ we assume $G(t\cdot\elide[\omega])=0$ for all $t\in\{\ini,~
  \trans{j}\}$ and all $\omega\in\UI$ with $t\in \UI_\omega$, for any $G$ can be brought
  into that form by applying (\ref{elideNF}).
  For each $s\in S\setminus S'$ we have $M'(s)=0$, and using this the inequations
  (\ref{undo})--(\ref{fetch}) and (\ref{p_j}) of \refcl{G-properties} turn into equations.
  For each $i\in T'$ we have $G(\sum_{j\geq^\# i}\exec{j})=0$, using (the equational form
  of) (\ref{undo})--(\ref{interfacecount}), and that $G(\comp{j})=0$. Since
  \plat{$G(\exec{j})\geq 0$} (or $\mbox{}\leq 0$) for all $j\geq^\# i$,
  this implies that \plat{$G(\exec{j})= 0$} for each $i\leqc j\in T'$.
  With (\ref{fetch}) we obtain $G(\fetched{j})=G(\fetch)=0$ for each applicable $p,c,i,j$.
  Using that $G(t\cdot\elide[\omega])=0$ for each applicable $t$ and $\omega$, with
  (\ref{reset})--(\ref{undocount}) and (\ref{p_j}) we find $G(t)=0$ for all $t\in T$.
\end{proofclaim}

\begin{clm}\label{cl-D}
Let $M:=M'+(M_0\mathord-M'_0)+\marking{H}\in[M_0\rangle_N$ for $M'\inp[M'_0\rangle_{N'}$ and
$H\fin \Int^T$ with \plat{$H(\exec{j})\mathbin= 0$} for all $i\leqc j\in T'$.
\begin{enumerate}[(a)]
\item If \plat{$H(\comp[i]{j})<0$} and \plat{$H(\comp[k]{l})<0$} for certain
  $i,k \in T'$ then $\neg(i \mathrel{\#} k)$.\label{Hcomp}
\item If \plat{$M[\exec{j}\rangle$} and \plat{$H(\comp[k]{l})<0$} for certain
  $i,k \in T'$ then $\neg(i \confeq k)$ and $\neg(j \confeq k)$.\label{Hexec}
\item $H(\dist)\geq 0$ for all $p\in S'$ (with $\postcond{p}\neq\emptyset$).\label{dist-positive-H}
\item Let $c\confeq i \in T'$.
  If $H(\dist)\geq F'(p,c)$ for all $p\in\precond{c}$, then $H(\comp{j})=0$.\label{dist-final}
\item If $M[\exec{j}\rangle$ with $i\leqc j\in T'$ then $M'[j\rangle$.\label{Hexecj}
\end{enumerate}
\end{clm}

\begin{proofclaimNobox}
By \refcl{extra} there exist $M'_1\in[M'_0\rangle_{N'}$ and $G_1\fin \Int^T$
satisfying (\ref{r1})--(\ref{rVeryLast}) (with $M$, $M'_1$ and $G_1$ playing the r\^oles of $M$, $M'$ and $G$).
In particular, $M=M'_1+(M_0-M'_0)+\marking{G_1}$, $G_1(\comp{j}) = 0$ for all $i\in T'$,
and $G_1(\exec{j})\geq 0$ for all $i\leqc j\in T'$. Using (\ref{rExec2}), for each $i\in
T'$ there is at most one $j\geq^\#i$ with \plat{$G_1(\exec{j})>0$}; we denote this $j$ by $f(i)$,
and let $f(i):=i$ when there is no such $j$. This makes $f:T'\rightarrow T'$ a function,
satisfying $G_1(\exec{j})=0$ for all $j\geq^\# i$ with $j\neq f(i)$.

Given that \plat{$H(\exec{j})\mathbin=0$} for all $i\leqc j\in T'$,
(\ref{undo})--(\ref{interfacecount}) (or (\ref{reset}) and (\ref{fetch})) imply
$H(\comp{j})\leq 0$ for all $i\in T'$.
Let $M'_2:=M'+\sum_{i \inp T'} H(\comp{j})\cdot \marking{i}$
and $G_2:=H-\sum_{i \inp T'}H(\comp{j})\cdot G^i_{\!f(i)}$, where $G^i_{\!\!j}$ is the
right-hand side of (\ref{mimic}).
Then $M = M'+(M_0-M'_0)+\marking{H} = M'_2+(M_0-M'_0)+\marking{G_2}$,
using that \plat{$\marking{i}=\marking{G^i_{\!f(i)}}$}.
Moreover, $G_2(\comp{j})=0$ for all $i\inp T'$, using that \plat{$G^i_{\!f(i)}(\comp{j})=1$}.

It follows that $M'_1-M'_2 = \marking {G_2-G_1}$. Moreover, we have $(G_2-G_1)(\comp{j})=0$ for
all $i\in T'$. We proceed to show that $G_2-G_1$ satisfies the remaining precondition of \refcl{0}.
So let $i \in T'$. In case \plat{$H(\comp{j})=0$}, for all $j\geq^\# i$ we have
\plat{$G_2(\exec{j})= 0$}, and $G_1(\exec{j})\geq 0$ by (\ref{rExec}). Hence \plat{$(G_2-G_1)(\exec{j})\leq 0$}.
In case \plat{$H(\comp{j})<0$}, we have \plat{$G_2(\exec{{f(i)}})\geq 1$}, and hence, using (\ref{rExec}),
\plat{$(G_2-G_1)(\exec{{f(i)}})\geq 0$}. Furthermore, for all $j\neq f(i)$,
\plat{$G_2(\exec{j})\geq 0$} and \plat{$G_1(\exec{j})=0$}, so again $(G_2-G_1)(\exec{j})\geq 0$.

Thus we may apply \refcl{0}, which yields $G_2 \mathbin\equiv G_1$. It follows that $M'_2\mathbin=M'_1\inp[M'_0\rangle_{N'}$.
\begin{enumerate}[(a)]
\item
Suppose that \plat{$H(\comp[i]{j})<0$} and \plat{$H(\comp[k]{l})<0$} for certain $i\mathrel\#k \in T'$.
Then \plat{$G_2(\exec{{f(i)}})\mathbin>0$} and $G_2(\exec[k]{{f(k)}})\mathbin>0$, so
$G_1(\exec{{f(i)}})\mathbin>0$ and $G_1(\exec[k]{{f(k)}})\mathbin>0$, contradicting (\ref{rExec2}).
\item
Suppose that \plat{$M[\exec{j}\rangle$} and \plat{$H(\comp[k]{l})<0$} for certain
  $k\confeq i$ or $k\confeq j$.\\
Then \plat{$G_1(\exec[k]{{f(k)}})=G_2(\exec[k]{{f(k)}})>0$},
contradicting (\ref{rExec4}) or (\ref{rExec5}).
\item
By (\ref{Hcomp}), for any given $p\in S'$ there is at most one $i\in\postcond{p}$ with $H(\comp{j})<0$.
For all $i\in T'$ with $i\notin\postcond{p}$ we have $G^i_{\!f(i)}(\dist)=0$.
First suppose $k\in\postcond{p}$ satisfies $H(\comp[k]{{f(k)}})<0$.
Then 
$$G_1(\exec[k]{{f(k)}})\begin{array}[t]{@{~=~}l}G_2(\exec[k]{{f(k)}})\\
   H(\exec[k]{{f(k)}})-\sum_{i \in T'}H(\comp{j})\cdot G^i_{\!f(i)}(\exec[k]{{f(k)}})\\
   0-H(\comp[k]{{f(k)}}),\end{array}$$
so by (\ref{r4}) $G_1(\dist)\geq -F'(p,k)\cdot H(\comp[k]{{f(k)}})$.
Hence
$$H(\dist)~\begin{array}[t]{@{}l}=~G_2(\dist)+\sum_{i\in T'}H(\comp{j})\cdot G^i_{\!f(i)}(\dist)\\
   =~ G_1(\dist)+H(\comp[k]{{f(k)}})\cdot G^{k}_{f(k)}(\dist)\\
   \geq~ -F'(p,k)\cdot H(\comp[k]{{f(k)}}) + H(\comp[k]{{f(k)}})\cdot F'(p,k)=0.
  \end{array}$$
In case there is no $i\in\postcond{p}$ with $H(\comp{j})<0$ we have
$$\qquad\! H(\dist)=G_2(\dist)+\!\sum_{i\in T'}H(\comp{j})\cdot G^i_{\!f(i)}(\dist)=G_1(\dist)\mathbin\geq 0\vspace{-2ex}$$
by (\ref{r4}) and (\ref{rExec}).

\item
Since \plat{$H(\comp{j})\leq 0$} and \plat{$G^i_{\!f(i)} (\dist)\geq 0$} for all $i \inp T'$,
also using (\ref{dist-positive-H}),
all summands in \plat{$H(\dist)+\sum_{i \in T'}-H(\comp{{f{i}}})\cdot G^i_{\!f(i)}(\dist)$} are positive.
Now suppose \plat{$H(\comp{j})<0$} for certain $i\inp T'$.
Then, using (\ref{r3}), for all $p\in\precond{i}$,
$$M'_1(p)\geq G_1(\dist) = G_2(\dist) \geq G^i_{\!f(i)}(\dist)=F'(p,i).$$
Furthermore, let $c \confeq i$ and suppose $H(\dist)\geq F'(p,c)$ for all $p\in\precond{c}$.
Then, using (\ref{r3}),\vspace{-2ex}
$$M'_1(p)\geq G_1(\dist) = G_2(\dist) \geq H(\dist)\geq F'(p,c)$$
for all $p\in\precond{c}$.
Moreover, if $p\in \precond{c}\cap\precond{i}$ then
$$M'_1(p)\geq G_2(\dist) \geq H(\dist)+G^i_{\!f(i)}(\dist) \geq F'(p,c)+F'(p,i).$$
Hence $M'_2 [\{c\}\mathord+\{i\}\rangle$. However, since $c \confeq i$ and $N'$ is a
\hyperlink{scn}{structural conflict net}, this is impossible.

\item Suppose $M[\exec{j}\rangle$ with $i\leqc j\in T'$.
Then $M'_1[j\rangle$ by (\ref{rVeryLast}).\\
Now $M'=M'_1+\sum_{k \inp T'} -H(\comp[k]{{f(k)}})\cdot \marking{k}$, with $-H(\comp[k]{j})\geq 0$
for all $k\in T'$. Whenever $-H(\comp[k]{j})>0$ then $\neg(j \confeq k)$ by (\ref{Hexec}).
Hence $M'[j\rangle$.
\hfill\filledbox
\end{enumerate}
\end{proofclaimNobox}
\bigskip
\noindent
We now define the class $\NF\subseteq \Int^T$ of signed multisets of transitions in
\emph{normal form} by $H\in\NF$ iff $\ell(H)\equiv\emptyset$ and, for all $t\in\{\ini,~
\trans{j} \mid h,j\inp T'\}$:
\begin{enumerate}[(NF-1)]
\item $H (t\cdot\elide[\omega]) \leq 0$ for each  $\omega\inp \UI$,\label{NF1}
\item $H (t\cdot\undo[\omega]) \geq 0$ for each  $\omega\inp \UI$, or $H (t\cdot\fire) \geq 0$,\label{NF2}
\item and if $H (t\cdot\elide[\omega]) < 0$ for any  $\omega\inp \UI$,
  then $H (t\cdot\undo[\omega]) \leq 0$ and $H (t\cdot\fire) \leq 0$.\label{NF3}
\end{enumerate}
We proceed verifying the remaining conditions of \refthm{3ST}.
\begin{enumerate}[1.]
\item[\ref{normalformST}.]
By applying (\ref{elideNF}), each signed multiset $G\fin\Int^T$ with $\ell(G)\equiv\emptyset$
can be converted into a signed multiset \mbox{$H\fin\NF$} with $\ell(H)\equiv\emptyset$, such that
$\marking{H}=\marking{G}$. Namely, for any $t\in\{\ini,~
\trans{j} \mid h,j\inp T'\}$, first of all perform the following
three transformations, until none is applicable:
\begin{enumerate}[(i)]
\item correct a positive count of a transition $t\cdot\elide[\omega]$ in $G$ by adding
  $t(\omega)-t\cdot\elide[\omega]$ to $G$;
\item if both $H(t\cdot\undo[\omega])<0$ for some $\omega$ and $H(t\cdot\fire)<0$,
  correct this in the same way;
\item and if, for some $\omega$, $t\mathord\cdot\elide[\omega]$ has a negative and
  $t\mathord\cdot\undo[\omega]$ a positive count, add $t\cdot\elide[\omega]-t(\omega)$.
\end{enumerate}
Note that transformation (iii) will never be applied to the same $\omega$ as (i) or (ii),
so termination is ensured. Properties (NF-\ref{NF1}) and (NF-\ref{NF2}) then hold for $t$.
After termination of (i)--(iii), perform
\begin{enumerate}[(i)]
\item[(iv)] if, for some $\omega$, $H(t\cdot\elide[\omega])<0$ and
  $H(t\cdot\fire)>0$, add $t\cdot\elide[\omega]-t(\omega)$.
\end{enumerate}
This will ensure that also (NF-\ref{NF3}) is satisfied, while preserving (NF-\ref{NF1}) and (NF-\ref{NF2}).

Define the function $f:T\rightarrow\nat$ by $f(u):=1$ for all $u\in T$ not of the form
$u=t\cdot\elide[\omega]$, and $f(t\cdot\elide[\omega]):=f(t(\omega))$ (applying the last
item of \refdf{multiset}). Then surely $f(G)=f(H)$.

\item[\ref{lastST}.] 
  Let $M'\in\nat^{S'}$, $U'\in\nat^{T'}$ and $U\in\nat^{T}$ with $\ell(U)=\ell'(U')$
  and $M'+\!\precond{U'}\in[M'_0\rangle_{N'}$.
  Since $N'$ is a \hyperlink{finitary}{finitary} \hyperlink{scn}{structural conflict net}, it admits no
  self-concurrency, so, as $\precond{U'}\leq M'+\!\precond{U'}\in[M'_0\rangle_{N'}$,
  the multiset $U'$ must be a set. As $N'$ is \hyperlink{plain}{plain}, this implies that the multiset $\ell'(U')$ is a set.
  Since $\ell(U)=\ell'(U')$, also $\ell(U)$, and hence $U$, must be a set.
  All its elements have the form $\exec{j}$ for  $i\leqc j\in T'$,
  since these are the only transitions in $T$ with visible labels.
  Note that $U'$ is completely determined by $U$, namely by
  \plat{$U'=\{i\mid \exists j.~\exec{j}\in U\}$}.
  We take $H_{M',U}:=$\vspace{-1ex}
  $$\qquad\quad\sum_{p\in S'} (M'\mathord+\!\precond{U'})(p)\cdot\{\dist\} +
  \!\!\!\!\!\!\!\! \sum_{(M'+\!\precond{U'})[j\rangle}\!\!\!\!\left(\{\ini\cdot\fire\} + 
  \hspace{-2.7em} \sum_{h<^\#j,~\nexists\exec[g]{h}\in U}\hspace{-2.6em} \{\trans[h]{j}\cdot\fire\}\right)
  $$
  Since $N'$ is \hyperlink{finitary}{finitary}, \plat{$H_{M',U}\fin\nat^{T_+}$}. Moreover,
  $\ell(H_{M',U})\equiv\emptyset$.

  Let $H\mathbin{\fin} \NF$ with $M:=M'+\!\precond{U'}+(M_0\mathord-M'_0)+\marking{H}-\precond{U}\in\nat^S$
  and $M+\precond{U}\in[M_0\rangle_N$.
  Since $H\inp\NF$, and thus $\ell(H)\equiv\emptyset$, $H(\exec[i]{j})=0$.
  From here on we apply \refcl{G-properties} and \refcl{D} with $M+\precond{U}$ and
  $M'+\precond{U'}$ playing the r\^oles of $M$ and $M'$.
  Note that the preconditions of these claims are met.

  That $H(\exec[i]{j})=0$ for all $i\leqc j\inp T'$, together with (\ref{undo}) and the
  requirements (NF-\ref{NF1}) and \mbox{(NF-\ref{NF3})} for normal forms, yields
  $H(t\cdot\elide)\leq 0$ as well as $H(t\cdot\undo)\leq 0$.
  Using this, (\ref{reset})--(\ref{fetch}) imply that
  \begin{equation}\label{T-negative}
  H(u)\leq 0 ~\mbox{ for each }~ u\in T_-.
  \end{equation}
  \begin{clm}\label{cl-C} Let $c \inp T'$ and $p\in\precond{c}$. Then
  \begin{iteMize}{$\bullet$}
  \item if $H(\ini[c]\cdot\fire)>0$ then $H(\fetch)=0$ for all $i\in\postcond{p}$ and $j\geq^\# i$, and
  \item if $H(\trans[b]{c}\cdot\fire)>0$ for some $b<^\#c$ then $H(\fetch)=0$ for all $i\in\postcond{p}$ and $j\geq^\# i$.
  \end{iteMize}
  \end{clm}
  \begin{proofclaim} Suppose that $H(t\cdot\fire)>0$, for $t=\ini[c]$ or \plat{$t=\trans[b]{c}$}.
  Then (\ref{pi-j}) resp.\ (\ref{pi-execute}) together with (\ref{T-negative}) implies that
  $H(t\cdot\reset[\omega])=0$ for each $\omega$ with $t\in\UI_\omega$.
  In order words, $H(t\cdot\reset)=0$ for each $i\confeq c$, so in particular for each $i\in\postcond{p}$.
  Furthermore, $H(t\cdot\elide)\geq 0$, by requirement (NF-\ref{NF3}) of normal forms.
  With (\ref{reset}), this yields $\sum_{j\geq^\#i}H(\fetched{j})\geq 0$, and
  (\ref{T-negative}) implies $H(\fetched{j})= 0$ for each $j\geq^\#i$.
  Now (\ref{fetch},\,\ref{T-negative}) gives $H(\fetch)= 0$ for each $j\geq^\#i\in\postcond{p}$.
  \end{proofclaim}
\noindent
  We proceed to verify the requirements (\ref{markingST})--(\ref{concurrent}) of \refthm{3ST}.

  \begin{enumerate}[(1)]
  \item[(\ref{markingST})] To show that $M_{M',U}\in\nat^S$, it suffices to apply it to the preplaces of
    transitions in $H_{M',U}+U$:\vspace{-1.5ex}
    $$\qquad\qquad\begin{array}{@{}l@{~=~}l@{}l@{}}
    M_{M',U}(p) & 0 & \mbox{for all }p\in S'\;;\\
    M_{M',U}(p_j) & \left\{\begin{array}{@{}ll@{}}
     (M'+\!\precond{U'})(p)-F'(p,j) & \mbox{if } (M'+\!\precond{U'})[j\rangle  \\
     (M'+\!\precond{U'})(p)           & \mbox{otherwise}
     \end{array}\right. & \mbox{for }p\inp S'\!,~j\inp\postcond{p};\\
    M_{M',U}(\pi_j) & \left\{\begin{array}{@{}l@{\quad}l@{}}
     \phantom{-}0 & \mbox{if } (M'+\!\precond{U'})[j\rangle  \\
     \phantom{-}1 & \mbox{otherwise}
     \end{array}\right. & \mbox{for }j\in T';\\
    M_{M',U}(\Pre^j_k) & \hspace{-1.6pt}\left\{\begin{array}{@{}l@{\quad}l@{}}
     \phantom{-}1 & \mbox{if } (M'+\!\precond{U'})[j\rangle \wedge \exec[j]{k}\notin U \\
     -1 & \mbox{if } \neg(M'+\!\precond{U'})[j\rangle \wedge \exec[j]{k}\in U \\
     \phantom{-}0 & \mbox{otherwise}
     \end{array}\right. & \mbox{for }j\leqc k\in T';\\
    M_{M',U}(\pi_{h\#j}) &  \left\{\begin{array}{@{}l@{\quad}l@{}}
     \phantom{-}0 & \mbox{if } \exists\exec[g]{h}\in U \vee (M'+\!\precond{U'})[j\rangle\\
     \phantom{-}1 & \mbox{otherwise}
     \end{array}\right. & \mbox{for }h<^\# j\in T'\\
    M_{M',U}(\transin{j}) &  \left\{\begin{array}{@{}l@{\quad}l@{}}
     \phantom{-}1 & \mbox{if } (M'+\!\precond{U'})[j\rangle \wedge \exists\exec[g]{h}\in U \\
     \phantom{-}0 & \mbox{otherwise}
     \end{array}\right. & \mbox{for }h<^\#j\in T';\\
    M_{M',U}(\transout{j}) & \left\{\begin{array}{@{}l@{\quad}l@{}}
     \phantom{-}1 & \makebox[0pt][l]{if $(M'+\!\precond{U'})[j\rangle \wedge \nexists\exec[g]{h}\in U 
                                                   \wedge \nexists\exec{j}\in U$} \\
     -1 & \makebox[0pt][l]{if $\big(\neg(M'+\!\precond{U'})[j\rangle \vee \exists\exec[g]{h}\in U\big) 
                                                   \wedge \exists\exec{j}\in U$} \\
     \phantom{-}0 & \mbox{otherwise}
     \end{array}\right. & \begin{array}{@{}l@{}}\mbox{}\\\mbox{}\\\mbox{for }h<^\#j\in T'.\end{array}\\
    \end{array}$$
    For all these places $s$ we indeed have that $M_{M',U}(s)\geq 0$,
    for the circumstances yielding the two exceptions above cannot occur:
    \begin{iteMize}{$\bullet$}
    \item Suppose \plat{$\exec[j]{k}\in U$} with $j\leqc k\in T'$. Then $j\in U'$, so
      $\precond{j} \leq M'+\!\precond{U'}$ and $(M'+\!\precond{U'})[j\rangle$.
      Consequently, $M_{M',U}(\Pre^j_k) \neq -1$ for all $j\leqc k \in T'$.
    \item Suppose $\exec{j}\in U$ with $i\leqc j\in T'$. Then $\precond{\exec{j}}\leq
      \precond{U}$, so $(M+\!\precond{U})[\exec{j}\rangle$.
      \refcl{D}(\ref{Hexecj}) with $M+\!\precond{U}$ and $M'+\!\precond{U'}$ in the
      r\^oles of $M$ and $M'$ yields $(M'+\precond{U'})[j\rangle$.

      If moreover $\exec[g]{h}\inp U$ with $g\mathbin{\leqc} h \mathbin{<^\#}\! j$, then
      $\{g\}\mathord+\{i\}\leq U'$, so $\precond\{g\}\mathord+\!\precond\{i\}\leq
      M'\mathord+\!\precond{U'}$ and
      $(M'+\precond{U'})[\{g\}\mathord+\{i\}\rangle$. In particular, $g\concurrent i$, and since $N'$
      is a \hyperlink{scn}{structural conflict net}, $\precond{g}\cap\precond{i}=\emptyset$.
      By \refcl{D}(\ref{Hexecj})---as above---$(M'\mathord+\!\precond{U'})[h\rangle$, so
      $\precond{g}\cup\precond{h}\cup \precond{j}\cup\precond{i}\leq
        M'\mathord+\!\precond{U'} \inp [M'_0\rangle_{N'}$.
      Moreover, since $g \leqc h <^\# j \geq^\# i$, we have
      $\precond{g}\cap\precond{h}\neq\emptyset$,
      $\precond{h}\cap\precond{i}\neq\emptyset$ and $\precond{i}\cap\precond{j}\neq\emptyset$.
      Now in case also $\precond{h}\cap\precond{i}\neq\emptyset$, the transitions $g$, $h$
      and $i$ constitute a \hyperlink{M}{fully reachable pure $\structuralM$};
      otherwise $h\concurrent i$ and  $h$, $j$
      and $i$ constitute a \hyperlink{M}{fully reachable pure $\structuralM$}.
      Either way, we obtain a contradiction.
      Consequently, $M_{M',U}(\transout{j}) \neq -1$ for all $h<^\# j \in T'$.
    \end{iteMize}
  \item[(\ref{matchST})] Suppose $M'\goesto[a]$; say $M'[i\rangle$ with $\ell'(i)=a$.
    Let $j$ be the largest transition in $T'$ w.r.t.\ the well-ordering $<$ on $T$
    such that $i\leqc j$ and $(M'+\!\precond{U'})[j\rangle$.
    It suffices to show that \plat{$M_{M',U} [\exec{j}\rangle$}, \ie that
    $M_{M',U}(\Pre^i_j)\mathord=1$, $M_{M',U}(\transout{j})\mathord=1$ for all $h\mathbin{<^\#}\!j$,
    and $M_{M',U}(\pi_{j\#l})\mathord=1$ for all $l\mathbin{>^\#}\!j$.
 
    If \plat{$\exec{j}\in U$} we would have $i\in U'$
    and hence $(M'+\!\precond{U'})[2\cdot\{i\}\rangle$.
    Since $N'$ is a \hyperlink{finitary}{finitary} \hyperlink{scn}{structural conflict net}, this is impossible.
    Therefore \plat{$\exec{j}\not\in U$} and, using the calculations from (a) above,
    $M_{M',U}(\Pre^i_j)=1$.
   
    Let $h<^\#j$. To establish that $M_{M',U}(\transout{j})=1$ we need to show that
    there is no $k\leqc j$ with \plat{$\exec[k]{j}\in U$} and no $g\leqc h$ with
    \plat{$\exec[g]{h}\in U$}. First suppose \plat{$\exec[k]{j}\in U$} for some $k\leqc j$.
    Then $k\in U'$ and hence $(M'+\!\precond{U'})[\{i\}\mathord+\{k\}\rangle$.
    This implies $i \smile k$, and, as $N'$ is a structural conflict net, $\precond{i}\cap \precond{k}=\emptyset$.	
    Hence the transitions $i$, $j$ and $k$ are all different, with $\precond{i}\cap \precond{j}\neq\emptyset$ and
    $\precond{j}\cap \precond{k}\neq\emptyset$ but $\precond{i}\cap \precond{k}=\emptyset$.
    Moreover, the reachable marking $M'+\!\precond{U'}$
    enables all three of them. Hence $N'$ contains a \hyperlink{M}{fully reachable pure $\structuralM$},
    which contradicts the assumptions of \refthm{correctness}.

    Next suppose \plat{$\exec[g]{h}\in U$} for some $g\leqc h$.
    Then $(M+\!\precond{U})[\exec[g]{h}\rangle$, so $(M'+\!\precond{U'})[h\rangle$
    by \refcl{D}(\ref{Hexecj}). Moreover, $g \in U'$, so $(M'+\!\precond{U'})[\{i\}\mathord+\{g\}\rangle$.
    This implies $g \smile i$, and $\precond{g}\cap \precond{i}=\emptyset$.
    Moreover, $\precond{g}\cap \precond{h}\neq\emptyset$, $\precond{h}\cap \precond{j}\neq\emptyset$ and
    $\precond{j}\cap \precond{i}\neq\emptyset$, while the reachable marking $M'+\!\precond{U'}$
    enables all these transitions. Depending on whether $\precond{h}\cap \precond{i}=\emptyset$,
    either $h$, $j$ and $i$, or $g$, $h$ and $i$ constitute a 
    \hyperlink{M}{fully reachable pure $\structuralM$},
    contradicting the assumptions of \refthm{correctness}.

    Let $l>^\#j$. To establish that \plat{$M_{M',U}(\pi_{j\#l})=1$} we need to show that
    there is no $k\leqc j$ with \plat{$\exec[k]{j}\in U$}---already done above---and that
    $\neg(M'+\!\precond{U'})[l\rangle$.
    Suppose $(M'+\!\precond{U'})[l\rangle$.
    Considering that $j$ was the largest transition with $i\leqc j$ and
    $(M'+\!\precond{U'})[j\rangle$, we cannot have $i<^\# l$.
    Hence the transitions $i$, $j$ and $l$ are all different, with $\precond{i}\cap \precond{j}\neq\emptyset$ and
    $\precond{j}\cap \precond{l}\neq\emptyset$ but $\precond{i}\cap \precond{l}=\emptyset$.
    Moreover, the reachable marking $M'+\!\precond{U'}$
    enables all three of them. Hence $N'$ contains a \hyperlink{M}{fully reachable pure $\structuralM$},
    which contradicts the assumptions of \refthm{correctness}.
  \item [(\ref{upperboundST})]
  We have to show that $H(t)\leq H_{M',U}(t)$ for each $t\in T$.
  \begin{enumerate}[$\bullet$]
  \item[$\bullet$]
    In case $t\in T_-$ this follows from (\ref{T-negative}) and \plat{$H_{M',U}\in\nat^{T_+}\!\!$}.
  \item[$\bullet$]
    In case $t=\exec{j}$ it follows since $\ell(H)\equiv\emptyset$.
  \item[$\bullet$]
    In case $t=\dist$ it follows from (\ref{p}) and (\ref{T-negative}).
  \item[$\bullet$]
    Next let $t=\ini[c]\cdot\fire$ for some $c\in T'$.
    In case $H(\ini[c]\cdot\fire)\leq 0$ surely we have $H(\ini[c]\cdot\fire)\leq H_{M',U}(\ini[c]\cdot\fire)$.
    So without limitation of generality we may assume that $H(\ini[c]\cdot\fire)>0$.
    By (\ref{pi-j},\,\ref{T-negative}) we have $H(\ini[c]\!\cdot\fire)=1$.
    Using (\ref{p_j}), \refcl{C}, (\ref{T-negative}) and (\ref{p}) we obtain, for all $p\in\precond{c}$,
    $$F'(p,c)\cdot H(\ini[c]\cdot\fire) \leq H(\dist) \leq (M'+\!\precond{U'})(p).$$
    Hence $c$ is enabled under  $M'+\!\precond{U'}$, which implies $H_{M',U}(\ini[c]\cdot\fire)=1$.
  \item[$\bullet$]
    Let $t\mathbin=\trans[b]{c}\cdot\fire$ for some $b\mathbin{<^\#}\! c\inp T'\!$.
    As above, we may assume $H(\trans[b]{c}\!\cdot\fire)\mathbin>0$.
    By (\ref{pi-execute},\,\ref{T-negative}) we have $H(\trans[b]{c}\!\cdot\fire)=1$.
    Using (\ref{T-negative}) and that $H(\exec[g]{b})=0$ for all $g\leqc b$, it follows
    that $(M+\!\precond{U})(\pi_{b\#c})=0$. Hence $\neg(M+\!\precond{U})[\exec[g]{b}\rangle$ for
    all $g\leqc b$, and thus $\nexists \exec[g]{b} \in U$.
    For all $p\in\precond{c}$ we derive
    $$\qquad\qquad\quad\begin{array}{@{}r@{~\leq~}l@{}r}
      \multicolumn{2}{@{}l@{}}{F'(p,c)\cdot H(\trans[b]{c}\cdot\fire)}\\
    \mbox{} & F'(p,c)\cdot\big(H(\trans[b]{c}\cdot\fire)-H(\trans[b]{c}\cdot\undone)\big)&
      (\ref{T-negative})\\
    & F'(p,c)\cdot\big(H(\ini[c]\cdot\fire)-H(\ini[c]\cdot\und[\mbox{$\transin[b]{c}$}])\big)&
      (\ref{transin})\\
    & F'(p,c)\cdot\big(H(\ini[c]\cdot\fire)-H(\ini[c]\cdot\undone)\big)&
      (\ref{undocount})\\
    \multicolumn{1}{r@{~=~}}{\mbox{}}
    & \displaystyle
      \mbox{[the same as above]}
      + \hspace{-.5em}\sum_{j\geq^\# i\in \postcond{p}}\hspace{-.5em}
      F'(p,i) \cdot H(\fetch) &
      (\mbox{\refcl{C}})\\[-10pt]
    & H(\dist) & (\ref{p_j}) \\
    & \displaystyle
      (M'+\!\precond{U'})(p)+ \hspace{-.5em}\sum_{\{i\in T'\mid
        p\in\postcond{i}\}}\hspace{-.5em} H(\comp{j}) & (\ref{p})\\[-10pt]
    & (M'+\!\precond{U'})(p) & (\ref{T-negative}).
    \end{array}\hspace{-2.5em}$$
    Hence $(M'+\!\precond{U'})[c\rangle$, and thus $H_{M',U}(\trans[b]{c})=1$.
  \end{enumerate}
  \item[(\ref{T-ST})]
    If $u\notin T_-$, yet $H(u)\neq 0$, then $u$ is either $\dist$, $\ini\cdot\fire$
    or $\trans{j}\cdot\fire$ for suitable $p\in S'$ or $h,j\in T'$.
    For $u=\dist$ the requirement follows from \refcl{D}(\ref{dist-positive-H});
    otherwise Property (NF-\ref{NF2}), together with (\ref{undocount}), guarantees that
    $H(u)\geq 0$.
  \item[(\ref{disjoint preplacesST})]
    If $H(t)\mathbin>0$ and $H(u)\mathbin<0$, then $t\inp T_+$ and $u\inp T_-$.
    The only candidates for $\precond{t}\cap\precond{u}\neq\emptyset$ are
    \begin{iteMize}{$\bullet$}
    \item \plat{$p_c \in \precond{(\ini[c]\cdot\fire)} \cap \precond{(\fetch)}$}
      for $p\in S'$, $c,i\in\postcond{p}$ and $j\geq^\# i$,
    \item $\transin[b]{c} \in \precond{(\trans[b]{c}\cdot\fire)}
                    \cap \precond{(\ini[c]\cdot\und[{\transin[b]{c}}])}$ for $b\leqc c\in T'$.
    \end{iteMize}
    We investigate these possibilities one by one.
    \begin{iteMize}{$\bullet$}
    \item $H(\ini[c]\cdot\fire)>0 \wedge H(\fetch)<0$ cannot occur by \refcl{C}.
    \item Suppose $H(\trans[b]{c}\cdot\fire)>0$.
      By (\ref{pi-execute},\,\ref{T-negative}) we have $H(\trans[b]{c}\!\cdot\fire)=1$.
      Through the derivation above, in the
      proof of requirement (c), using (\ref{T-negative},\,\ref{transin},\,\ref{undocount}),
      \refcl{C} and (\ref{p_j}), we obtain $H(\dist)\geq F'(p,c)$ for all $p\in\precond{c}$.
      Now \refcl{D}(\ref{dist-final}) yields $H(\comp{j})=0$ for all $i\confeq c$. By (\ref{reset}) and
      (\ref{T-negative}) we obtain $H(\ini[c]\!\cdot\reset)\mathbin=0$ for each such $i$.
      Hence {$\sum_{i\confeqscript c} H(\ini[c]\!\cdot\reset)\mathbin=0$}, and thus
      $H(\ini[c]\cdot\und[{\transin[b]{c}}])=0$ by (\ref{undocount},\,\ref{T-negative}).
    \end{iteMize}
  \item[(\ref{disjoint preplaces 2ST})]
    If $H(u)<0$ and $(M+\!\precond{U})[t\rangle$ with $\ell(t)\neq\tau$,
    then $t=\exec{j}$ for some $i\leqc j\in T'$ and $u\inp T_-$.
    The only candidates for $\precond{t}\cap\precond{u}\neq\emptyset$ are
    \begin{iteMize}{$\bullet$}
    \item $\Pre^i_j \in \precond{(\exec{j})}
                    \cap \precond{(\ini\cdot\und[\Pre^i_j])}$ and
    \item $\transout{j} \in \precond{(\exec{j})}
                    \cap \precond{(\trans{j}\cdot\und[\transout{j}])}$ for $h<^\# j$.
    \end{iteMize}
    We investigate these possibilities one by one.
    \begin{iteMize}{$\bullet$}
    \item Suppose $(M+\!\precond{U})[\exec{j}\rangle$.
      By \refcl{D}(\ref{Hexec}), \plat{$H(\comp[k]{j})\geq 0$} for each $k\confeq i$. By (\ref{reset}) and
      (\ref{T-negative}) we obtain $H(\ini[i]\!\cdot\reset[k])\mathbin=0$ for each such $k$.
      Hence \hspace{-.6pt}\plat{$\sum_{k\confeqscript i} H(\ini[i]\!\cdot\reset[k])\mathbin=0$}, and thus
      $H(\ini[i]\cdot\und[\Pre^i_j])=0$ by (\ref{undocount},\,\ref{T-negative}).
    \item Suppose $(M+\!\precond{U})[\exec{j}\rangle$ and $h<^\#j$.
      By \refcl{D}(\ref{Hexec}), \plat{$H(\comp[k]{j})\geq 0$} for each $k\confeq j$. By (\ref{reset}) and
      (\ref{T-negative}) \plat{$H(\trans{j}\!\cdot\reset[k])\mathbin=0$} for each such $k$.
      So {$\sum_{k\confeqscript j} H(\trans{j}\!\cdot\reset[k])\mathbin=0$}, and
      $H(\trans{j}\cdot\und[\transout{j}])=0$ by (\ref{undocount},\,\ref{T-negative}).
    \end{iteMize}
  \item[(\ref{concurrent})]
    Suppose $(M+\!\precond{U})[\{t\}\mathord+\{u\}\rangle_N$, and $i,k\in T'$
    with $\ell'(i)=\ell(t)$ and $\ell'(k)=\ell(u)$.
    Since the net $N'$ is \hyperlink{plain}{plain}, $t$ and $u$ must have the form $\exec{j}$ and $\exec[k]{j}$
    for some $j>^\#i$ and $l>^\#k$. \refcl{concurrency} yields $\neg(i\confeq k)$ and hence
    $\precond{i}\cap\precond{k}=\emptyset$.
    \qed
  \end{enumerate}
\end{enumerate}
\end{proofNobox}
\noindent
Thus, we have established that the conflict replicating implementation $\impl{N'}$ of a finitary plain
structural conflict net $N'$ without a fully reachable pure $\structuralM$ is branching
ST-bisimilar with explicit divergence to $N'$. It remains to be shown that $\impl{N'}$ is
essentially distributed.

\begin{lem}\label{lem-S-invariant}
Let $N$ be the conflict replicating implementation of a finitary net\\ $N'=(S',T',F',M'_0,\ell')$;
let $j,l\in T'\!$, with $l\mathbin{>^\#} j$.
Then no two transitions from the set $$\{\exec{j}\mid i\leqc j\}
  \cup\{\trans[j]{l}\cdot\fire\} \cup \{\trans[j]{l}\cdot\und[\mbox{$\transout[j]{l}$}]\}
  \cup\{\exec[k]{l}\mid k\leqc l\}$$ can fire concurrently.
\end{lem}

\begin{proof}
  For each \plat{$i\mathbin{\leqc} j$} pick an arbitrary preplace $q_i$ of $i$.
  The set $$
  \{\txf{fetch}^{q_i,i}_{i,j}\txf{-in},~\txf{fetch}^{q_i,i}_{i,j}\txf{-out}\mid i\leqc j\}
  \cup \{\pi_{j\#l},~\transout[j]{l},~\took(\transout[j]{l},\trans[j]{l}),~\rho(\trans[j]{l}\}$$
  is an \emph{S-invariant}: there is always exactly one token in this set. This is the case because
  there is exactly one token initially (on $\pi_{j\#l}$) and
  each transition from $N$ has as many (with multiplicities) preplaces as postplaces in this set.
  The transitions from
  $$\{\exec{j}\mid i\leqc j\}
  \cup\{\trans[j]{l}\cdot\fire\} \cup\{\trans[j]{l}\cdot\und[\mbox{$\transout[j]{l}$}]\}\vspace{-2pt}
  \cup\{\exec[k]{l}\mid k\leqc l\}$$
  each have a preplace in this set.
  Hence no two of them can fire concurrently.
\end{proof}

\begin{lem}\label{lem-essentially distributed}
Let $N$ be the conflict replicating implementation $\impl{N'}$ of a finitary plain
structural conflict net $N'=(S',T',F',M'_0,\ell')$ without a fully reachable pure $\structuralM$.
Then for any $i\leqc j \confeq c\in T'$ and \hyperlink{far}{$f\in (\ini[c])^{\,\it far}$},
the transitions \plat{$\exec{j}$} and $\ini[c]\cdot\und[f]$ cannot fire concurrently.
\end{lem}
\begin{proof}
Suppose these transitions can fire concurrently, say from the marking $M\in[M_0\rangle_N$.
By \refcl{extra}, there are $M'\in[M'_0\rangle_{N'}$ and $G\fin\Int^T$ such that
(\ref{r1})--(\ref{rVeryLast}) hold. Let $t:=\ini[c]$, $G_1:=G+\{t\cdot\und[f]\}$ and
$M_1\mathbin{:=}M+\marking{t\mathord\cdot\und[f]}$.
Then (\ref{undocount}), applied to the triples $(M,M',G)$ and $(M_1,M',G_1)$, yields
\[
  \sum_{\makebox[1em][l]{$\scriptstyle\{\omega\mid t\in \UI_\omega\}$}} G(t\cdot\reset[\omega])
  \leq G(t\cdot\und) < G_1(t\cdot\und)
  \leq \sum_{\makebox[1em]{$\scriptstyle\{\omega\mid t\in \UI_\omega\}$}} G_1(t\cdot\undo[\omega])
  = \sum_{\makebox[1em]{$\scriptstyle\{\omega\mid t\in \UI_\omega\}$}} G(t\cdot\undo[\omega]).
\]
Hence, there is an $\omega$ with $t\in \UI_\omega$ and $G(t\cdot\reset[\omega])< G(t\cdot\undo[\omega])$.
This $\omega$ must have the form $k\in T'$ with $k\confeq c$. We now obtain
  $$\begin{array}[b]{r@{~\leq~}ll}
  \multicolumn{1}{r@{~=~}}{0}
  & G(\comp[k]{l})
  & \mbox{(by (\ref{r2}))} \\
  & G(t\cdot\elide[k])+G(t\cdot\reset[k])
  & \mbox{(by (\ref{reset}))} \\
  \multicolumn{1}{r@{~<~}}{} & G(t\cdot\elide[k])+G(t\cdot\undo[k]) \\
  & \sum_{l\geq^\#k}G(\exec[k]{l})
  & \mbox{(by (\ref{undo}))}.
  \end{array}$$
Hence, there is an $l\geq^\# k \confeq c$ with $G(\exec[k]{l})>0$.
By (\ref{rExec5}) we obtain $\neg(j\confeq k)$, so $\precond{j}\cap\precond{k}=\emptyset$.
Additionally, we have $\precond{j}\cap\precond{c}\neq\emptyset$ and
$\precond{c}\cap\precond{k}\neq\emptyset$.
By (\ref{rVeryLast}) we obtain $M'[j\rangle$, and
by (\ref{r3}) and (\ref{r4}) $M'[k\rangle$.
Furthermore, by (\ref{undocount}), $G(t\cdot\und)<G_1(t\cdot\und)\leq G_1(t\cdot\fire) =
G(t\cdot\fire)$, so, for all $p\inp\precond{c}$,
  $$\begin{array}[b]{r@{~\leq~}ll}
  F'(p,c)
  & F'(p,c)\cdot \big(G(t\cdot\fire)-G(t\cdot\und)\big)\\
  & F'(p,c)\cdot \big(G(t\cdot\fire)-G(t\cdot\undone)\big)
  & \mbox{(by (\ref{undocount}))} \\
  & G(\dist) - \sum_{j\geq^\# i\in \postcond{p}} F'(p,i) \cdot G(\fetch)
  & \mbox{(by (\ref{p_j}))} \\
  & G(\dist)
  & \mbox{(by (\ref{rFp}) and (\ref{fetch}))} \\
  & M'(p)
  & \mbox{(by (\ref{r3}).} \\
  \end{array}$$
It follows that $M'[c\rangle$.
Thus $N'$ contains a \hyperlink{M}{fully reachable pure $\structuralM$},
which contradicts the assumptions of \reflem{essentially distributed}.
\end{proof}

\begin{thm}\label{thm-cri-distributed}
  Let $N$ be the conflict replicating implementation $\impl{N'}$ of a finitary plain
  structural conflict net $N'$ without a fully reachable pure $\structuralM$.
  Then $N$ is essentially distributed.
\end{thm}
\begin{proof}
  We take the canonical distribution $D$ of $N$, in which $\equiv_D$ is the
  equivalence relation on places and transitions generated by Condition (1) of \refdf{distributed}.
  We need to show that this distribution satisfies Condition ($2'$) of \refdf{externally distributed}.
  A given transition $t$ with $\ell(t)\neq\tau$ must have the form \plat{$\exec{j}$} for some $i\leqc j\in T'$.
  By following the flow relation of $N$ one finds the places and transitions that, under
  the canonical distribution, are co-located with \plat{$\exec{j}$}:
  $$\begin{array}{@{}l@{}}
  \pi_{j\#l} \rightarrow \trans[j]{l}\cdot\fire \leftarrow \transin[j]{l} \rightarrow
  \ini[l]\cdot\und[\mbox{$\transin[j]{l}$}] \leftarrow \take(\transin[j]{l},\ini[l]) \\
  ~~\downarrow\\
  \hspace*{-1em}\exec{j} \\
  ~~\uparrow\\
  \transout{j}  \rightarrow \trans{j}\cdot\und[\transout{j}] \leftarrow \take(\transout{j},\trans{j}) \\
  ~~\downarrow\\
  \exec[g]{j} \\
  ~~\uparrow\\
  \Pre^g_j  \rightarrow \ini[g]\cdot\und[\Pre^g_j] \leftarrow \take(\Pre^g_j,\ini[g])
  \end{array}$$
  for all $l\mathbin{>^\#} j$, $h\mathbin{<^\#} j$ and $g\leqc j$.
  We need to show that none of these transitions can happen concurrently with \plat{$\exec{j}$}.
  For transitions $\trans[j]{l}\cdot\fire$ and $\exec[g]{j}$ this follows directly from \reflem{S-invariant}.
  For \plat{$\trans{j}\cdot\und[\transout{j}]$} this also follows from \reflem{S-invariant}, in
  which $j$, $k$ and $l$ play the r\^ole of the current $h$, $i$ and $j$.
  For the transitions \plat{$\ini[l]\cdot\und[\mbox{$\transin[j]{l}$}]$}
  and \plat{$\ini[g]\cdot\und[\Pre^g_j]$} this has been established in \reflem{essentially distributed}.
\end{proof}
\noindent
Our main result follows by combining Theorems~\ref{thm-correctness},
\ref{thm-cri-distributed} and~\ref{thm-bothdistributedequal}:

\begin{thm}\label{thm-fullmgttrulysync}
  Let $N$ be a finitary plain structural conflict net
  without a fully reachable \visible pure \structuralM.
  Then $N$ is distributable up to $\approx^\Delta_{bSTb}$.
    \qed
\end{thm}

\begin{cor}\label{cor-fullmeqtrulysync}
  Let $N$ be a finitary  plain structural conflict net.
  Then $N$ is distributable iff it has no fully reachable
  \visible pure~\structuralM.
    \qed
\end{cor}

\section{Conclusion}

In this paper, we have given a precise characterisation of
distributable Petri nets in terms of a semi-structural property. Moreover, we
have shown that our notion of distributability corresponds to an
intuitive notion of a distributed system by establishing that any
distributable net may be implemented as a network of asynchronously
communicating components.

In order to formalise what qualifies as a valid implementation, we needed a suitable
equivalence relation. We have chosen step failures equivalence for showing the
impossibility part of our characterisation, since it is one of the simplest and least
discriminating semantic equivalences imaginable that abstracts from internal actions but
preserves branching time, concurrency and divergence to some small degree. For the
positive part, stating that all other nets are implementable, we have introduced a
combination of several well known rather discriminating equivalences, namely a divergence
sensitive version of branching bisimulation adapted to ST-semantics.  Hence our
characterisation is rather robust against the chosen equivalence; it holds in fact for all
equivalences between these two notions.  However, ST-equivalence (and our version of it)
preserves the causal structure between action occurrences only as far as it can be expressed in terms
of the possibility of durational actions to overlap in time. Hence a natural question is
whether we could have chosen an even stronger causality sensitive equivalence for our
implementability result, respecting e.g.\ pomset equivalence or history preserving
bisimulation.  Our conflict replicating implementation does not fully preserve the causal
behaviour of nets; we are convinced that we have chosen the strongest possible equivalence
for which our implementation works.  It is an open problem to find a class of nets that
can be implemented distributedly while preserving divergence, branching time and causality
in full.
Another line of research is to investigate which Petri nets can be
implemented as distributed nets when relaxing the requirement of
preserving the branching structure.
We conjecture that there exists a notion of equivalence that captures
some branching time aspects, but not as strongly as step failures
equivalence, under which all Petri nets become distributable.
However, also in this case it is problematic, in fact even impossible in our setting, to
preserve the causal structure, as has been shown in \cite{EPTCS64.9}.
A similar impossibility result has been obtained in the world of the $\pi$-calculus in
\cite{EPTCS64.7}.

In this paper we have sought a characterisation of distributability only for \emph{plain} nets,
in which all transitions have a different label and none are internal.
Naturally, any distributed implementation that applies to plain nets having a semi-structural
property---in particular the one contributed here---also applies to non-plain nets having the same
semi-structural property. Namely to implement a non-plain net $N$, note that $N$ can be written as
$\rho(N')$, where $N'$ is a plain net and $\rho$ a relabelling function. A correct implementation of $N$
is now obtained as $\rho(\impl{N'})$, where $\impl{N'}$ is the distributed implementation of $N'$.
Yet, it appears unlikely that there is a semi-structural characterisation that captures
\emph{all} non-plain distributable nets: for any non-trivial semi-structural
property there probably are nets that do not have that property, but are semantically equivalent to
nets that do. This may happen for instance when some essential transitions that violate the property
are labelled $\tau$ and can be abstracted away. Thus, we do not expect that a natural
characterisation of distributability for non-plain nets exists---where ``natural'' excludes
characterisations that just say, in other words, ``being equivalent to a distributed net''.

Our work shows that the main problem in creating distributed implementations of systems 
arises from the interplay between choice and synchronous communication. This issue
has already been investigated
in the context of distributed algorithms. Rabin and Lehmann observed in
\cite{rabin94advantageoffreechoice} that there is no fully symmetric distributed solution
to the dining philosophers problem. In \cite{bouge88symmetricleader} Luc Boug\'e
considers the problem of implementing symmetric leader election in the
sublanguages of CSP obtained by allowing different forms of
communication, combining input and output guards in guarded choice in
different ways. He finds that the possibility of implementing leader
election depends heavily on the structure of the communication
graphs. Truly symmetric schemes are only possible in CSP with
arbitrary input and output guards in choices.

Synchronous interaction is a basic concept in many languages for
system specification and design, e.g.\ in statechart-based approaches
and in process calculi. For process calculi, language hierarchies have
been established which exhibit the expressive power of different forms
of synchronous and asynchronous interaction.  In
\cite{boer91embedding} Frank de Boer and Catuscia Palamidessi consider
various dialects of CSP with differing degrees of asynchrony.  Similar
work is done for the $\pi$-calculus in \cite{palamidessi97comparing}
by Catuscia Palamidessi, in \cite{nestmann00what} by Uwe Nestmann and
in \cite{G:FoSSaCS06} by Daniele Gorla.  A rich hierarchy of
asynchronous $\pi$-calculi has been mapped out in these papers.
Similar to the findings of Boug\'e,
mixed-choice, i.e.\ the ability to combine input and output guards in a
single choice, plays a central r\^ ole in the implementation of
synchronous behaviour.

In \cite{selinger97firstorder}, Peter Selinger considers labelled
transition systems whose visible actions are partitioned into input and
output actions. He defines asynchronous implementations of such a
system by composing it with in- and output queues, and then
characterises the systems that are behaviourally equivalent to their
asynchronous implementations. The main difference with our approach is
that we focus on asynchrony within a system, whereas Selinger focuses
on the asynchronous nature of the communications of a system with the
outside world.

Dirk Taubner has in \cite{taubner88zurverteiltenimpl} given various
protocols by which to implement arbitrary Petri nets in the OCCAM programming
language. Although this programming language offers synchronous communication
he makes no substantial use of that feature in the protocols, thereby
effectively providing an asynchronous implementation of Petri nets. He does not
indicate a specific equivalence relation, but is effectively using
linear-time equivalences to compare implementations to the specification.

Also in hardware design it is an intriguing quest to use interaction
mechanisms which do not rely on a global clock, in order to gain
performance. Here the simulation of synchrony by asynchrony can be a
crucial issue, see for instance \cite{lamport78ordering} and
\cite{lamport02arbitration}.

The idea of modelling asynchronously communicating sequential components by sequential
Petri nets interacting though buffer places has already been considered in
\cite{reisig82buffersync}. There Wolfgang Reisig introduces a class of systems, represented
as Petri nets, where the relative speeds of different components are guaranteed to be
irrelevant.  His class is a strict subset of our LSGA nets, requiring additionally,
amongst others, that all choices in sequential components are free, \ie do not depend upon
the existence of buffer tokens, and that places are output buffers of only one component.
Another quite similar approach was taken in \cite{EHH10}, where
transition labels are classified as being either input or output.
There, asynchrony is introduced by adding new buffer places during
net composition.
This framework does not allow multiple senders for a single receiver.

\begin{figure}[tb]
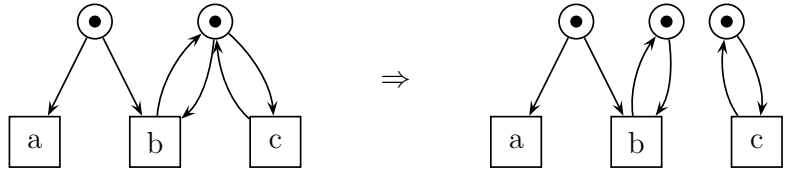

  \begin{center}
    \begin{petrinet}(14,3.4)
      \P (2,3):p;
      \P (4,3):q;

      \t (1,1):a:a;
      \t (3,1):b:b;
      \t (5,1):c:c;
      
      \a p->a; \a p->b;
      \A b->q; \A q->b; \A c->q; \A q->c;

      \rput(7,2){\Large $\Rightarrow$}

      \P (10,3):implp;
      \P (11.5,3):implqb;
      \P (12.5,3):implqc;
      
      \t (9,1):impla:a;
      \t (11,1):implb:b;
      \t (13,1):implc:c;

      \a implp->impla; \a implp->implb;
      \A implqb->implb; \A implb->implqb;
      \A implqc->implc; \A implc->implqc;
    \end{petrinet}
  \end{center}
\vspace{-1.5em}
  \caption{A specification and its Hopkins-implementation which added concurrency.}
  \label{fig-hopkins-added-concurrency}
\end{figure}

Other notions of distributed and distributable Petri nets are proposed in
\cite{hopkins91distnets,BCD02,BD11}. In these works, given a distribution of the
transitions of a net, the net is distributable iff it can be implemented by a net that is
distributed w.r.t.\ that distribution. The requirement that concurrent transitions may not
be co-located is absent; given the fixed distribution, there is no need for such a
requirement. These papers differ from each other, and from ours, in what counts as a valid
implementation.
Hopkins \cite{hopkins91distnets} uses an interleaving equivalence to compare an
implementation to the original net, and while allowing a range of implementations, he
does require them to inherit some of the structure of the original net.
The net classes he describes in his paper are incomparable with our class of
distributable nets. One direction of this inequality depends on his
choice of interleaving semantics, which allows the implementation in
\reffig{hopkins-added-concurrency}. The step failures equivalence we use
does not tolerate the added concurrency and the depicted net is not distributable in our sense.
The other direction of the inequality stems from the fact that we allow
implementations which do not share structure with the specification but only
emulate its behaviour. That way, the net in \reffig{distr-not-hopkins} can be
implemented in our approach as depicted.

\begin{figure}
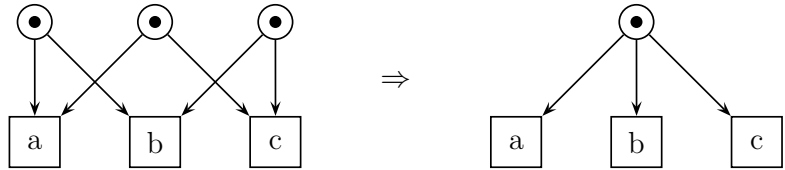

  \begin{center}
    \begin{petrinet}(14,3.4)
      \P (1,3):p;
      \P (3,3):q;
      \P (5,3):r;

      \t (1,1):a:a;
      \t (3,1):b:b;
      \t (5,1):c:c;

      \a p->a; \a p->b;
      \a q->a; \a q->c;
      \a r->b; \a r->c;

      \rput(7,2){\Large $\Rightarrow$}

      \P (11,3):pimpl;

      \t (9,1):aimpl:a;
      \t (11,1):bimpl:b;
      \t (13,1):cimpl:c;

      \a pimpl->aimpl; \a pimpl->bimpl; \a pimpl->cimpl;
    \end{petrinet}
  \end{center}
\vspace{-1.5em}
  \caption{A distributable net which is not considered distributable in
    \cite{hopkins91distnets}, and its implementation.}
  \label{fig-distr-not-hopkins}
\end{figure}

A more abstract approach to the same underlying problem of correctly executing
an arbitrary Petri net as a distributed system has been taken in
\cite{katoen13taming}. The authors provide a modified net semantics and an
algorithm to split the net into agents which can locally decide most choices
and resort to a global scheduler in case multiple agents must be coordinated.
While such an approach looses branching time equivalence between a net and its
implementation, it provides a clear separation of concerns between 
executing the net and solving the distributed coordination problems.

In \cite{glabbeek08syncasyncinteractionmfcs} we have obtained a characterisation similar to
Corollary~\ref{cor-fullmeqtrulysync}, but for a much more restricted notion of distributed
implementation (\emph{plain distributability}), disallowing nontrivial transition
labellings in distributed implementations.  We also proved that fully reachable pure
\structuralM s are not implementable in a distributed way, even when using transition
labels (\refthm{trulysyngltfullm}). However, we were not able to show that this upper
bound on the class of distributable systems was tight.  Our current work implies the
validity of Conjecture~1 of \cite{glabbeek08syncasyncinteractionmfcs}.
While in \cite{glabbeek08syncasyncinteractionmfcs} we considered only one-safe
place/transition systems, the present paper employs a more general class of
place/transition systems, namely structural conflict nets. This enables us to give a
concrete characterisation of distributed nets as systems of sequential components
interacting via non-safe buffer places.

On the level of applications, we expect our results to be useful for language
design. We would like to make a thorough comparison of our results to those
on communication patterns in process algebras, versions of the $\pi$-calculus and
I/O-automata \cite{lynch96}. Using a Petri net semantics of a suitable system description
language, we could compare our class of distributed nets to the class of nets expressible
in the language, especially when restricting the allowed communication patterns in the
ways considered in~\cite{boer91embedding,bouge88symmetricleader} or in \cite{lynch96}.
A first step in that direction is \cite{ESOP13}.

\subsection*{Acknowledgment}
The authors gratefully thank the referees of this paper for their very 
thorough examination and helpful suggestions.

\bibliographystyle{eptcsalpha}

\end{document}